\documentclass[runningheads]{llncs}
\usepackage[T1]{fontenc}
\usepackage{amsmath,amssymb,amsfonts}
\usepackage{graphicx}
\usepackage{empheq}
\usepackage{color}
\usepackage{dsfont}
\usepackage{subfigure}
\usepackage{balance}
\usepackage{textcomp}
\usepackage{cite}
\usepackage{ulem}
\usepackage{hyperref}
\usepackage{url}
\usepackage{multicol,lipsum}
\usepackage{algorithm}
\usepackage{tabularx}
\usepackage{caption} 
\usepackage{adjustbox}
\usepackage{booktabs}
\usepackage{algorithmic}
\usepackage{multirow}
\usepackage{mathtools}

\usepackage{hyperref}
\DeclareMathOperator*{\argmin}{arg\,min}
\usepackage{bm}
\usepackage{bbding}
\usepackage{listings}
\usepackage{graphicx}
\usepackage{subcaption}
\usepackage{float}
\usepackage{pifont}
\newcommand{\oomit}[1]{}
\newtheorem{assumption}{Assumption}

\begin{document}

\title{PAC One-Step Safety Certification for Black-Box Discrete-Time Stochastic Systems}

\titlerunning{PAC One-Step Safety Certification for Black-Box Stochastic Systems} 

 \author{Taoran Wu\inst{1,2}, Dominik Wagner\inst{3}, Jingduo Pan\inst{1,2},
    Luke Ong\inst{3}, Arvind Easwaran\inst{3}, and Bai Xue\inst{1,2}}
\institute{Key Laboratory of System Software, Institute of Software, Chinese Academy of Sciences, Beijing, China \email{\{wutr,panjd,xuebai\}@ios.ac.cn}
\and
University of Chinese Academy of Sciences, Beijing, China
\and
Nanyang Technological University, Singapore
\email{\{luke.ong,dominik.wagner,arvinde\}@ntu.edu.sg}}

\authorrunning{T. Wu, D. Wagner, J. Pan, L. Ong, A. Easwaran, and B. Xue}

\maketitle

%% Abstract
\begin{abstract}
This paper investigates the problem of safety certification for black-box discrete-time stochastic systems, where both the system dynamics and disturbance distributions are unknown, and only sampled data are available. Under such limited information, ensuring robust or classical quantitative safety over finite or infinite horizons is generally infeasible. To address this challenge, we propose a data-driven framework that provides theoretical one-step safety guarantees in the Probably Approximately Correct (PAC) sense. This one-step guarantee can be applied recursively online at each time step, thereby yielding step-by-step safety assurances over extended horizons. Our approach formulates barrier certificate conditions based solely on sampled data and establishes PAC safety guarantees by leveraging the VC dimension, scenario approaches, Markov’s inequality, and Hoeffding’s inequality. Two sampling procedures are proposed, and three methods are proposed to derive PAC safety guarantees. The properties and comparative advantages of these three methods are thoroughly discussed. Finally, the effectiveness of the proposed methods are demonstrated through several numerical examples.
\end{abstract}

\section{Introduction}
\label{sec:intro}
With the rapid advancement of artificial intelligence, the deployment of intelligent autonomous systems has significantly increased. Ensuring the safety of such systems is of great significance. A system is considered \textit{safe} if its trajectories do not leave a designated set of safe states, the \emph{safe set}.
% \changed[dw]{a designated set of safe states, the \emph{safe set}.}
In practice, however, system dynamics are often stochastic due to uncertainties and disturbances present in real-world environments, posing major challenges to such strong safety guarantees.

In recent years, several methods have been developed for safety/reach-avoid verification of stochastic systems. Notable approaches include abstraction-based techniques (e.g., \cite{zamani2014symbolic,lahijanian2015formal}) and barrier certificates (e.g., \cite{prajna2007framework, chakarov2013probabilistic, jagtap2020formal,kenyon2021supermartingales, zhi2024unifying,xue2024sufficient}). However, these methods typically require accurate system dynamics and disturbance distributions, which are often unavailable or too costly to obtain, leaving many systems effectively black-box. Although one can try to learn these models from data, doing so is difficult and time-consuming, and the resulting models are often too imprecise to support reliable formal guarantees. As a result, providing robust or quantitative safety guarantees over either finite or infinite horizons is especially challenging, and sometimes impossible, for black-box stochastic systems. Data-driven methods (e.g., \cite{nejati2023formal, nazeri2025data} \footnote{{In our view, the statistical guarantee in \cite{nazeri2025data} requires additional justification due to the reuse of the confidence parameter $\beta$ across multiple inference stages. Statistically, this amounts to treating a joint estimation problem as a series of independent marginal estimations without appropriately allocating the confidence budget. This raises a multiple comparisons concern: to ensure a valid system-level guarantee, the confidence budget should be distributed across all inference stages rather than reused, as otherwise the probability of system-wide failure could exceed $\beta$. A possible resolution is to employ confidence-sequence techniques, which provide uniform-in-time guarantees while appropriately controlling the overall error probability \cite{waudby2024estimating}.}}) have been proposed recently to provide safety guarantees for such systems. {Although these works target finite or infinite-horizon safety verification, they typically require strong prior knowledge, such as Lipschitz constants, to generalize from finite samples to the entire state space. In a more general black-box setting where such constants are unknown, ensuring such strong safety guarantees is generally infeasible.}
%but they often rely on knowing the system's Lipschitz constants, which are difficult to determine in practice. 
To address these challenges, and inspired by \cite{wu2025convex,xue2020pac}, this paper uses barrier certificates to provide formal one-step safety guarantees for black-box discrete-time stochastic systems, relying solely on sampled data instead of explicit models.

%When the system dynamics are unknown, the proposed method uses sampled data to formally verify whether the system remains within a predefined safe set in the next time step with high confidence. If the system remains safe in one step, the procedure can be recursively applied to establish step-by-step safety guarantees.

In this paper, we propose a data-driven framework for safety certification of black-box discrete-time stochastic systems. The system dynamics and disturbance distributions are assumed to be unknown, and only sampled data is available. Based on these samples, the proposed approaches establish formal one-step safety guarantees in the PAC sense by solving barrier certificate conditions. We investigate two sampling strategies in combination with two types of barrier certificates, robust and stochastic, resulting in three distinct forms of PAC safety certification, as formulated in Subsection \ref{sub:ps}. The first method solves a robust barrier certificate condition using one-to-one state-disturbance sample pairs. PAC safety certification is then established by combining VC dimension theory or scenario approaches with Markov's inequality. The second method also solves a robust barrier certificate condition but employs one-to-many state-disturbance samples. The corresponding PAC safety certification is derived through scenario approaches, combined with Markov's and Hoeffding's inequalities.  The third method addresses a stochastic barrier certificate condition using one-to-many state-disturbance samples, and the PAC safety certification is similarly established via scenario approaches, together with Markov's and Hoeffding's inequalities. The properties and respective advantages of these three methods are also discussed, and the appropriate choice depends on the specific problem setting. Finally, the effectiveness of the proposed approaches are demonstrated through several numerical examples.

The main contributions of this work are summarized as follows:
\begin{enumerate}
\item We propose three PAC methods for one-step safety certification of black-box discrete-time stochastic systems. To our knowledge, this is the first work to provide PAC guarantees for safety certification in such systems.
%\item We develop three methods based on robust and stochastic barrier certificates, combined with complexity measures from statistical learning theory—such as VC dimension, scenario approaches, and Hoeffding's inequality—to establish PAC safety guarantees. 
%\item We introduce an iterative algorithm that employs gradient descent to synthesize controllers satisfying the proposed PAC safety guarantees for black-box systems.
\item We correct the PAC statement in Theorem 4 of \cite{xue2020pac}, although our present characterization is inspired by its original formulation. See Remark \ref{correction} in Subsection \ref{sec:probabilistic} for details.
\item We implement a prototype tool to support the proposed framework, available at \href{https://github.com/TaoranWu/PSC-BDSS}{https://github.com/TaoranWu/PSC-BDSS}. The effectiveness of our approach is demonstrated through several numerical examples.
\end{enumerate}

\textbf{Related Work.}  Barrier certificates were first introduced for deterministic systems as Lyapunov-like tools for proving safety and reachability \cite{prajna2004safety,prajna2007convex}, and later extended to stochastic systems for finite and infinite-horizon safety and reach-avoid analysis. Subsequent work developed probability bounds using Ville’s inequality \cite{prajna2007framework,steinhardt2012finite,jagtap2018temporal,jagtap2020formal,santoyo2021barrier,vzikelic2023learning,zhi2024unifying} and equation relaxations \cite{xue2021reach,xue2022reach,yu2023safe,xue2024sufficient,xue2024finite,chen2025construction}, and applied barrier methods to probabilistic programs and $\omega$-regular properties \cite{chakarov2013probabilistic,mciver2017new,kenyon2021supermartingales,abate2024stochastic,wang2025verifying}. %For practical applications with bounded horizons, finite-time verification is more relevant. Here, (c)-martingale–based barrier conditions are widely used, building on \cite{kushner1967,steinhardt2012finite} and later extended to finite-time temporal-logic verification and controller synthesis \cite{jagtap2018temporal,jagtap2020formal,santoyo2021barrier}. Recent work further sharpens these conditions and their probability bounds \cite{zhi2024unifying,xue2024finite} and studies systems modeled by difference inclusions \cite{ghanbarpour2025characterization}.
 %All the above works  Recent work %\cite{zhi2024unifying,xue2024finite} introduced barrier-like conditions that cover both lower and upper bounds over finite and infinite horizons, and %\cite{xue2024finite} proposed barrier-like conditions for both lower- and upper-bounding finite-time safety and reach-avoid probabilities.  
The above methods assume full knowledge of the system, including both dynamics and disturbance distributions. In contrast, this work considers systems with unknown dynamics and disturbance distributions. %, reflecting practical scenarios—such as autonomous systems operating in uncertain environments—where exact disturbance statistics are unavailable.

On the other hand, data-driven methods for formal verification and design of black-box systems have received increasing attention, e.g., \cite{samari2024single,wooding2024learning}. A representative data-driven approach is scenario optimization, originally proposed in \cite{calafiore2006scenario} to address robust optimization problems. 
% By sampling from the distribution of uncertain variables, this approach replaces an infinite set of constraints with a finite set, and provides an upper bound on the probability of constraint violation. 
It has since been integrated with barrier certificates for safety verification and controller synthesis in deterministic systems \cite{xue2019probably,nejati2023formal, wu2025convex, rickard2025data, samari2025data}. Later, \cite{mathiesen2023inner,gracia2024data,mathiesen2024data,wwox2026} investigated finite-horizon safety guarantees for stochastic discrete-time systems with known dynamics but unknown disturbance distributions. Subsequently, \cite{nejati2023data,salamati2024data} addressed the more challenging case of black-box stochastic systems with both unknown dynamics and disturbance distributions. Furthermore, \cite{salamati2022data} and \cite{salamati2022safety} aimed to reduce the sample complexity by employing the wait-and-judge technique~\cite{campi2018wait} and the repetitive scenario technique~\cite{calafiore2016repetitive}, respectively. {Although these works target finite or infinite-horizon safety verification, they typically require strong prior knowledge, such as Lipschitz constants, to generalize from finite samples to the entire state space. In a more general black-box setting where such constants are unknown, ensuring such strong safety guarantees is generally infeasible. In contrast, our approach eliminates the need for such prior constants. Moreover, we would like to highlight a technical subtlety in these approaches. To the best of our knowledge, works such as \cite{salamati2024data,salamati2022data,salamati2022safety} apply Chebyshev’s inequality to obtain uniform characterizations of expectations (e.g., via Rademacher complexity). However, the validity of this step in the current setting requires additional justification, since applying a pointwise bound to a data-dependent selection does not automatically extend to a uniform guarantee.}

%In \cite{mathiesen2023inner}, the authors studied finite-horizon safety guarantees for discrete-time piecewise affine stochastic systems with unknown disturbance distributions. Their approach uses stochastic barrier certificates to reformulate the problem as a chance-constrained optimization problem and synthesizes piecewise affine stochastic barrier functions via scenario optimization. Because it relies on this chance-constrained formulation, the method is conceptually similar to the robust barrier certificate approach proposed in this work, which enforces constraint satisfaction over all sampled disturbances. More recently, \cite{mathiesen2024data} extended this method to nonlinear systems with affine disturbances by over-approximating them as uncertain piecewise affine systems. Further, \cite{gracia2024data} addressed the synthesis of provably correct controllers for general nonlinear switched systems with unknown disturbance distributions under LTL$_f$ specifications. Their method is based on learning an ambiguity set that contains the unknown disturbance distribution and subsequently constructing a robust Markov decision process as a finite abstraction of the system. However, it is well known that such discretization-based approaches suffer from the curse of dimensionality.

%The main focus is on handling unknown stochastic disturbance distributions, a common scenario in real-world applications. 

\textbf{Notations.} Let $\mathbb{R}$, $\mathbb{R}^n$, and $\mathbb{N}$ denote real numbers, $n$-dimensional vectors, and non-negative integers. For $\bm{v} \in \mathbb{R}^n$, $\bm{v}[j]$ is its $j$-th component. For sets $\mathbb{A}$ and $\mathbb{B}$, $\mathbb{A} \setminus \mathbb{B}$ is the difference, $\overline{\mathbb{A}}$ the closure, and $\mathbb{A}^c$ the complement. The indicator $1_{\mathbb{S}}(\bm{x})$ equals 1 if $\bm{x} \in \mathbb{S}$, 0 otherwise. For a probability space $(\mathbb{Y}, \mathcal{F}, \textnormal{P}_{\bm{y}})$, $\textnormal{P}_{\bm{y}}[\mathbb{E}]$ is the probability of $\mathbb{E}$, and $\textnormal{E}_{\bm{y}}[\cdot]$ its expectation.

\section{Preliminaries}
\label{sec:pre}
In this section, we define black-box stochastic systems and three PAC one-step safety certification problems, and then briefly review the scenario approach.
\subsection{Problem Statement}
\label{sub:ps}
We consider a class of black-box discrete-time stochastic systems
\begin{equation}
\label{eq:system}
    \bm{x}(t+1) = \bm{f}(\bm{x}(t), \bm{d}(t)), \quad \forall t \in \mathbb{N},
\end{equation}
where $\bm{x}(t)\in\mathbb{R}^n$ is the state and $\bm{d}(t)\in\mathbb{D}$ is a stochastic disturbance. The disturbance sequence $\{\bm d(t)\}_{t\ge0}$ is i.i.d.\ with distribution $\mathrm{P}_{\bm d}$ supported on $\mathbb D$. Both the dynamics $\bm f$ and the disturbance distribution $\mathrm{P}_{\bm d}$ are unknown.

Let the safe set be a compact set $\mathbb{X}\subset\mathbb{R}^n$. We equip $\mathbb X$ with the probability space $(\mathbb{X},\mathcal{F}_{\bm x},\mathrm{P}_{\bm x})$, where $\mathrm{P}_{\bm x}$ denotes the uniform probability measure on $\mathbb X$ (i.e., the normalized Lebesgue measure). {It is worth remarking that the paper illustrates the problem using a uniform distribution to determine the ``volume'' of the safe set. For black-box systems, assuming a uniform distribution is a general and natural choice, since it assigns equal importance to all states and naturally characterizes the volume of a computed set. However, the underlying theorems and proofs are distribution-agnostic. The method remains valid for any state distribution $\mathrm{P}_x$, provided that the training samples are drawn from the same distribution.} 
Starting from any $\bm x\in\mathbb X$ and under any disturbance $\bm d\in\mathbb D$, the next state is given by $\bm x^+=\bm f(\bm x,\bm d)$.

We now formulate the three safety certification problems considered.

\begin{problem}[PAC One-Step Safety Certification I]
\label{prob:1}
The objective is to design a certification procedure $\mathcal{A}$ that provides one-step safety guarantees in the single-layer PAC sense by leveraging one-to-one state–disturbance samples. Specifically, with confidence at least $1-\delta$ over the sample set $
\mathbb{S} = \{(\bm{x}^{(i)}, \bm{d}^{(i)})\}_{i=1}^N 
\stackrel{\text{i.i.d.}}{\sim} 
\textnormal{P}_{\bm{x}} \times \textnormal{P}_{\bm{d}}$ \footnote{{We clarify that our method does not require observing the explicit values of the disturbances $d^{(i)}$. In our notation $\mathbb{S} = \{(\bm{x}^{(i)}, \bm{d}^{(i)})\}_{i=1}^N$, $\bm{d}^{(i)}$ formally denotes the realization of the random disturbance that generated the transition.
However, the optimization constraints in this paper only require evaluating the barrier function at the current state $\bm{x}^{(i)}$ and the observed next state $\bm{x}^{(i)+} = \bm{f}(\bm{x}^{(i)}, \bm{d}^{(i)})$. The numerical value of $\bm{d}^{(i)}$ itself is never used in the computation. This is a key feature of our black-box approach, as it relies solely on observed state transitions $(\bm{x}^{(i)}, \bm{x}^{(i)+})$}.}, if $\mathcal{A}(\mathbb{S})$ accepts, then for at least a $(1-\alpha_1)$ fraction of states $\bm{x}$ in $\mathbb{X}$, the following holds:  the probability that the system~\eqref{eq:system}, starting from $\bm{x}$, remains within $\mathbb{X}$ at the next step is no less than $1-\alpha_2$, where $\alpha_1,\alpha_2,\delta \in (0,1)$.

Formally, the single-layer one-step PAC safety certification is expressed as
\[
\textnormal{P}_{\mathbb{S}} \!\left[
\mathcal{A}(\mathbb{S}) \Rightarrow 
\textnormal{P}_{\bm{x}}\left[
\textnormal{P}_{\bm{d}}\left[\bm{f}(\bm{x},\bm{d})\in \mathbb{X}\right]
\ge 1-\alpha_2
\right] \ge 1-\alpha_1
\right] \ge 1-\delta,
\]
where $\textnormal{P}_{\mathbb{S}} := (\textnormal{P}_{\bm{x}} \times \textnormal{P}_{\bm{d}})^N$.
\end{problem}

In Problem~\ref{prob:1}, the confidence level $1-\delta$ corresponds to the joint probability measure $(\textnormal{P}_{\bm{x}} \times \textnormal{P}_{\bm{d}})^N$, which is the $N$-fold product of $\textnormal{P}_{\bm{x}} \times \textnormal{P}_{\bm{d}}$ defined on the sample space $(\mathbb{X} \times \mathbb{D})^N$ (the Cartesian product of $\mathbb{X} \times \mathbb{D}$ with itself $N \in \mathbb{N}$ times).

%\dw{I think it should be stressed why this is strictly stronger then having the joint distribution also in the inner probability.}
%\wtr{I added a remark after Theorem 1 to clarify this point. However, the paper is currently more than one page over the page limit, so this remark may eventually need to be moved to the appendix.}

\begin{problem}[PAC One-Step Safety Certification II]
\label{prob:2}
The goal is to design a certification procedure $\mathcal{A}$ that provides one-step safety guarantees in the PAC sense by leveraging one-to-many state–disturbance samples. 
Specifically, with confidence at least $1-\delta$ over the sample set 
$\mathbb{S} = \{(\bm{x}^{(i)}, \{\bm{d}^{(i,j)}\}_{j=1}^M)\}_{i=1}^N
\stackrel{\text{i.i.d.}}{\sim} 
\textnormal{P}_{\bm{x}} \times \textnormal{P}_{\bm{d}}^M$, where $\{\bm{d}^{(i,j)}\}_{j=1}^M \stackrel{\text{i.i.d.}}{\sim} \textnormal{P}_{\bm{d}}$, 
if $\mathcal{A}(\mathbb{S})$ accepts, then for at least a $(1-\alpha_1)$ fraction of states $\bm{x}$ in $\mathbb{X}$, it holds:  
the probability that the system~\eqref{eq:system}, starting from $\bm{x}$, remains within $\mathbb{X}$ at the next step is no less than $1-\alpha_2$, where $\alpha_1,\alpha_2,\delta \in (0,1)$.

Formally, the PAC one-step safety certification is expressed as
\[
\textnormal{P}_{\mathbb{S}} \!\left[
\mathcal{A}(\mathbb{S}) \Rightarrow 
\textnormal{P}_{\bm{x}}\left[ 
\begin{split}
\textnormal{P}_{\bm{d}}\left[\bm{f}(\bm{x},\bm{d})\in \mathbb{X}\right]\ge 1-\alpha_2
\end{split}
\right] \ge 1-\alpha_1
\right] \ge 1-\delta,
\]
where $\textnormal{P}_{\mathbb{S}} := (\textnormal{P}_{\bm{x}} \times \textnormal{P}_{\bm{d}}^M)^N$.
\end{problem}

In Problem \ref{prob:2}, the probability $1-\delta$ corresponds to the joint probability $(\textnormal{P}_{\bm{x}}\times \textnormal{P}_{\bm{d}}^M)^N$ (the product of the probability $\textnormal{P}_{\bm{x}}\times \textnormal{P}_{\bm{d}}^M $ repeated $N$ times) in the space $(\mathbb{X}\times \mathbb{D}^M)^N$(the Cartesian product of $\mathbb{X}\times \mathbb{D}^M$ with itself $N$ times). Setting $M=1$ recovers Problem \ref{prob:1} from Problem \ref{prob:2}. \textit{However, we do not merge Problems \ref{prob:1} and \ref{prob:2} in this paper; see Remark \ref{remark_com_single_and_nested} for an explanation.} 

\begin{problem}[PAC One-Step Safety Certification III]
\label{prob:3}
Similar to Problem~\ref{prob:2}, the goal is to design a certification procedure $\mathcal{A}$ that yields PAC one-step safety guarantees using one-to-many state--disturbance samples. The key difference is that the fixed threshold $\alpha_2$ is replaced by a state-dependent function $\epsilon(\cdot): \mathbb{X} \to [0,1]$, which may also depend on the sampled dataset $\mathbb{S}$. 

Formally, the PAC one-step safety requirement is
\[
\textnormal{P}_{\mathbb{S}}\!\left[
\mathcal{A}(\mathbb{S}) \Rightarrow 
\textnormal{P}_{\bm{x}}\!\left[
\textnormal{P}_{\bm{d}}\!\left[\bm{f}(\bm{x},\bm{d}) \in \mathbb{X}\right]
\ge 1 - \epsilon(\bm{x})
\right] \ge 1 - \alpha_1
\right] \ge 1 - \delta,
\]
where $\textnormal{P}_{\mathbb{S}} := (\textnormal{P}_{\bm{x}} \times \textnormal{P}_{\bm{d}}^M)^N$ and $\alpha_1, \delta \in (0,1)$.
\end{problem}
%Problem 1 provides a single-layer PAC guarantee over states and disturbances. While conceptually simple, it may be conservative for complex stochastic systems. Problem 2 introduces a nested structure by using multiple i.i.d. disturbance samples per state, with separate confidence levels for the sample set ($\delta_1$) and the disturbance realizations ($\delta_2$). This improves robustness and reduces conservatism, but restricts the guarantee to a computed subset $\mathbb{X}_s$. Problem 3 further generalizes Problem 2 by allowing a state-dependent probability threshold $\epsilon(\bm{x})$, capturing variability in safety requirements across states, and represents the most flexible and general nested PAC formulation.

% \sout{Problems \ref{prob:1}–\ref{prob:3} form a hierarchy of PAC one-step safety guarantees with increasing flexibility. Problem \ref{prob:1} uses paired state–disturbance samples, which is simple but can be conservative when state samples are limited. Problem \ref{prob:2} improves this by drawing multiple i.i.d. disturbance samples per state, giving stronger guarantees. Problem \ref{prob:3} further generalizes by allowing state-dependent probability thresholds, making the PAC guarantee most flexible.}

{The three problems were introduced to establish a hierarchy of increasing flexibility and reduced conservatism for safety certification.
\begin{enumerate}
    \item Problem \ref{prob:1} establishes the baseline framework using standard one-to-one sampling (pairs of states and disturbances). While simple, it assesses the safety of a state based on a single realization of the disturbance. Consequently, a state might be deemed `unsafe' in the worst-case analysis due to one rare, unlucky disturbance, without accounting for its robustness against the broader disturbance distribution. 
    \item Problem \ref{prob:2} improves upon this by utilizing one-to-many sampling (multiple disturbances per state). This decoupling allows for tighter variance reduction and stronger guarantees, particularly in settings where we can query the system multiple times from the same initial state (e.g., via a simulator or reset mechanism), rather than relying solely on single trajectories. It is worth remarking here that setting $M=1$ recovers Problem \ref{prob:1}. 
    \item Problem \ref{prob:3} provides the most general formulation by allowing state-dependent safety thresholds. Instead of enforcing a uniform safety probability across the entire state space (as in Problems 1 and 2), it adapts the guarantee based on the specific location in the state space, which is essential for complex systems where safety margins vary.
\end{enumerate}
}

\begin{remark}
It is worth remarking here that our one-step guarantee is designed to be applied iteratively online for ongoing safety evaluation. Fundamentally, our framework establishes a probabilistic invariance property for the safe set $\mathbb{X}$. Specifically, the guarantee states that, with confidence at least $1-\delta$, a large fraction (at least $1-\alpha_1$) of states in $\mathbb{X}$ are `safe'—meaning they will remain in $\mathbb{X}$ with probability $\ge 1-\alpha_2$.

 In our framework, the initial sampling distribution $\mathrm{P}_{\bm{x}}$ serves solely as a fixed reference measure to quantify the fraction of the safe region. Since this measure characterizes the extent of the safe set itself rather than the system's dynamics, it remains a valid metric for safety evaluation independent of the system's actual evolving state distribution. In other words, safety is a conditional property: it depends on the state's location within the safe set, not on the probability density of visiting that location.

Consequently, this probabilistic guarantee supports iterative application online: at any time step $t$, provided the system is in $\mathbb{X}$, there is a high probability (\emph{measured by the reference volume $1-\alpha_1$ with confidence $1-\delta$}) that the current state is one for which the safety transition holds. While this does not guarantee safety for every possible trajectory (which is infeasible for black-box stochastic systems), it provides a rigorous, step-by-step statistical assurance that the system remains in the `safe' region. This constitutes a valid recursive framework for online safety monitoring and risk assessment over extended horizons.
\end{remark}
\color{black}
\subsection{Scenario Optimization}
\label{sec: sop}
We review the scenario optimization framework, based on the foundational work in \cite{calafiore2006scenario}, which addresses robust  convex programs of the form:
\begin{equation}
\label{ROP}
    \begin{split}
        \bm{z}^*=&\argmin_{\bm{z}\in \mathcal{Z}\subset \mathbb{R}^m} \bm{c}^{\top} \bm{z}, \quad\texttt{s.t.~} \bm{g}(\bm{z},\bm{\delta})\leq 0, ~\forall \bm{\delta}\in \Delta,
    \end{split}
\end{equation}
where the uncertain parameter $\bm{\delta} \in \Delta$ follows a probability distribution $\textnormal{P}_\delta$. 
\begin{assumption}
\label{as:convex}
$\mathcal{Z} \subset \mathbb{R}^m$ is a convex and closed set, $\Delta\subseteq \mathbb{R}^{n_{\delta}}$, and $\bm{g}(\bm{z}, \bm{\delta}): \mathcal{Z}\times \Delta \rightarrow [-\infty, \infty]$ is continuous and convex in $\bm{z}$, for any fixed value of $\bm{\delta} \in \Delta$.  
\end{assumption}

Since the constraint set is infinite, a fully robust solution $\bm{z}^*$ is generally unattainable. Scenario optimization addresses this by enforcing the constraints only on a finite number of sampled scenarios. {Conceptually, this approach relies on the idea that if a solution is robust to a sufficiently large number of random samples, it is statistically likely to be robust to the entire uncertainty set. This allows us to convert the intractable infinite-constraint problem into a standard convex optimization problem, while providing rigorous probabilistic guarantees on the solution's validity.}
% \textcolor{red}{Assume the uncertain parameter $\bm{\delta} \in \Delta$ follows a probability distribution $\textnormal{P}_\delta$.}
For $N$ i.i.d. samples $\Delta^N=\{\bm{\delta}_i\}_{i=1}^N\stackrel{\text{i.i.d.}}{\sim}\textnormal{P}_\delta$, we consider the following scenario optimization:
\begin{equation}
% \tag{ROP-N}
\label{ROP-N}
    \begin{split}
        \bm{z}_N^*=&\argmin_{\bm{z}\in \mathcal{Z}\subset \mathbb{R}^m} \bm{c}^{\top} \bm{z}, \quad\texttt{s.t.~}\bm{g}(\bm{z},\bm{\delta}_i)\leq 0, ~i=1,\ldots, N, 
    \end{split}
\end{equation}
under the following assumption:
\begin{assumption}[\cite{calafiore2006scenario}]
\label{as:solution}
For any given set of N samples $\Delta^N=\{\bm{\delta}_i\}_{i=1}^N$, program \eqref{ROP-N} has a unique solution $\bm{z}_N^*$ after applying a  tie-break rule.
\end{assumption}

To assess the performance of a candidate solution $\bm{z}_N^*$ across all uncertainties in $\Delta$, we define the violation set $F(\bm{z}) = \{\bm{\delta} \in \Delta \mid \bm{g}(\bm{z}, \bm{\delta}) \not\leq 0\}$ and the violation probability $V(\bm{z})=\textnormal{P}_{\delta}[\bm{\delta}\in F(\bm{z})]$ of a given $\bm{z}\in \mathcal{Z}$, which is the probability of sampling a constant $\bm{\delta}$ to which $\bm{z}$ is not robust. A central result in scenario optimization is as follows.
\begin{proposition}[\cite{campi2009scenario}]
\label{prop:pac}
For $\alpha, \beta, N, m$ satisfying $\alpha,\beta\in (0,1)$ and $\alpha \ge \frac{2}{N}(\ln{\frac{1}{\beta}} + m)$, we have
$\textnormal{P}_{\Delta^N}[V(\bm{z}_N^*) \le \alpha] \ge 1 - \beta$, where $\textnormal{P}_{\Delta^N} = \textnormal{P}_\delta^N$ is the probability measure over $N$ i.i.d. samples of $\bm{\delta}$.
\end{proposition}

\section{PAC Safety Certification with One-to-One Samples}
\label{sec:verification}
In this section, we introduce our data-driven method for solving Problem~\ref{prob:1}.
Subsections~\ref{subsec:PACRC} and~\ref{subsec:PACSA} address it using two approaches: one combines robust barrier certificates (RBC) with VC-dimension theory, and the other uses the scenario approach. The VC-based method accommodates more general classes of RBCs but typically requires a larger sample size, whereas the scenario approach achieves comparable guarantees with fewer samples but assumes convexity of the barrier certificate in its unknown parameters.

A RBC $h(\bm{x}):\mathbb{R}^n \rightarrow \mathbb{R}$ whose existence ensures safety of system~\eqref{eq:system}: if the system starts in $\mathbb{X}$, it will remain within $\mathbb{X}$ for all future time steps.
\begin{definition}
\label{def:rbc}
    Given $\gamma \in (0,1)$ and the safe set $\mathbb{X}$, a function $h(\cdot):\mathbb{R}^n\rightarrow\mathbb{R}$ is a \textnormal{RBC} for the system 
    \eqref{eq:system} if the following inequalities hold:
    \begin{subequations}
    \label{eq:rbf}
    \begin{empheq}[left=\empheqlbrace]{align}
        &h(\bm{x}) < 0, & \forall \bm{x} \in \mathbb{R}^n\setminus\mathbb{X}, \label{eq:rbf1} \\
        &h(\bm{x}) \geq 0, &\forall \bm{x}\in \mathbb{X}, \label{eq:rbf2} \\
        &h(\bm{f}(\bm{x},\bm{d})) \geq \gamma h(\bm{x}), & \forall \bm{x} \in \mathbb{X}, ~\forall \bm{d} \in \mathbb{D}. \label{eq:rbf3}
    \end{empheq}
    \end{subequations}
\end{definition}

\begin{proposition}[\cite{prajna2004safety}]
If there exists a \textnormal{RBC} $h(\cdot): \mathbb{R}^n\rightarrow\mathbb{R}$, then for any disturbance $\bm{d} \in \mathbb{D}$ and initial state $\bm{x} \in \mathbb{X}$, the system \eqref{eq:system} will remain within the safe set $\mathbb{X}$ at the next time step, i.e., $\forall \bm{d} \in \mathbb{D}. \forall \bm{x} \in \mathbb{X}. \bm{f}(\bm{x},\bm{d}) \in \mathbb{X}$. 
\end{proposition}

In the following, we use the sample set $\mathbb{S} = \{(\bm{x}^{(i)}, \bm{d}^{(i)})\}_{i=1}^N$ to solve \eqref{eq:rbf1}--\eqref{eq:rbf3} and thereby address Problem~\ref{prob:1}. 
\subsection{PAC Safety Certification Based on VC Dimension}
\label{subsec:PACRC}
In this section, we present our first method for addressing Problem~\ref{prob:1}. In this approach, the sample complexity required to achieve the desired safety guarantees is derived using the VC dimension. {Intuitively, the VC dimension measures the `expressiveness' or complexity of the chosen class of barrier functions. A class with a finite VC dimension prevents the algorithm from `overfitting' to the sampled data, ensuring that a barrier certificate validated on finite samples will generalize to unseen states with high probability. }

We consider a parameterized function $h(\bm{a},\cdot):\mathbb{R}^n \to \mathbb{R}$ of the form
\begin{equation}
\label{barrier_form}
h(\bm{a},\bm{x})
= 1_{\mathbb X}(\bm{x}) \cdot h_1(\bm{a},\bm{x})
+ C \cdot 1_{\mathbb{X}^c}(\bm{x}),
\end{equation}
where $C<0$ is a fixed scalar.  
For $\bm{a}\in [0,U_a]^m$ and $\xi\in[0,\overline{\xi}]$ with $U_a>0$ and $\overline{\xi}>-C$ being user-specified bounds for each component of $\bm{a}$ and for $\xi$, respectively, we define the indicator function $\phi_{\bm{a},\xi}(\bm{x},\bm{d})
=\mathbf{1}_{\left\{(\bm{x},\bm{d}):
h(\bm{a}, \bm{f}(\bm{x},\bm{d}))
<\gamma\, h(\bm{a},\bm{x}) - \xi
\right\}}(\bm{x},\bm{d})$, collect these classifiers in the family
$\mathbb{H}
=\left\{
\phi_{\bm{a},\xi} :
\bm{a}\in[0,U_a]^m,\;
\xi\in[0,\overline{\xi}]
\right\}$, and impose the following assumptions: 1). The parametric real-valued function class
     $\{h(\bm{a},\cdot) : \bm{a}\in[0,U_a]^m\}$
   has finite pseudo-dimension (equivalently, finite VC--subgraph dimension); 2). $h(\bm{a},\bm{x})\ge 0$ for all $(\bm{x},\bm{a})\in \mathbb{X}\times[0,U_a]^m$, and 
   $h_1(\bm{0},\bm{x})=0$ for all $\bm{x}\in\mathbb{X}$. Under these assumptions, the induced classifier family $\mathbb{H}$ has finite VC dimension; that is,
$\operatorname{vc}(\mathbb H) < \infty$.

\begin{remark}
\label{remark:barrier}
Using the barrier function in \eqref{barrier_form} has two main advantages. 
First, it simplifies computation by avoiding the need to enforce 
$h(\bm{a},\bm{x}) \ge 0$ for all $\bm{x}\in\mathbb{X}$ and 
$h(\bm{a},\bm{x}) < 0$ for all $\bm{x}\notin\mathbb{X}$, conditions that cannot be reliably checked by sampling.  Second, the structure of \eqref{barrier_form} prevents the safety test from requiring the system to enter the unsafe set. 
If a prediction shows $\bm{f}(\bm{x},\bm{d}) \in \mathbb{X}^c$, we can simply set 
$h(\bm{a},\bm{f}(\bm{x},\bm{d})) = C$ without needing to observe the true next state, which could put the system at risk. This choice provides a practical balance between computational tractability, conservativeness, and safety. More general barrier forms will be explored in future work. \qed
\end{remark}

Then, we search for $\bm{a}$ by solving the following optimization \eqref{eq:rbc_rc}, which is constructed with the sample set $\mathbb{S}=\{(\bm{x}^{(i)},\bm{d}^{(i)})\}_{i=1}^N$ % \stackrel{\text{i.i.d.}}{\sim} \mathbb{P}_{\bm{x}} \times \mathbb{P}_{\bm{d}}$ 
and the constraint \eqref{eq:rbf3}:
\begin{equation}
    \label{eq:rbc_rc}
  \begin{split}
    &\textstyle \min_{\xi\in\mathbb{R},\bm{a}\in \mathbb{R}^m } \xi \\
    \text{s.t.}&
    \begin{cases}
        h(\bm{a}, \bm{f}(\bm{x}^{(i)},\bm{d}^{(i)})) \geq \gamma h(\bm{a}, \bm{x}^{(i)}) - \xi,  \\
        \xi \in [0,\overline{\xi}]; ~ \bm{a}[k]  \in [0, U_a], \\
        i=1,\ldots,N;~ k=1,\ldots,m. \\
    \end{cases}
    \end{split}
\end{equation}
Here, $\gamma \in (0,1)$ is a user-specified scalar. It is noted that \eqref{eq:rbc_rc} is always feasible, since  $(\bm{a}=\bm{0},\xi=-C)$ constitutes a feasible solution.

%Let $\bm{a}^*(\mathbb{S})$ denote the solution to the optimization problem~\eqref{eq:rbc_rc1}. Obviously, $(\xi^*(\mathbb{S})=0, \bm{a}^*(\mathbb{S}))$ is an optimal solution to \eqref{eq:rbc_rc}.

%When a solution with $\xi^*(\mathbb{S})=0$ to is found, we proceed with solving the following optimization \eqref{eq:rbc_rc1}. 

%\begin{equation}
%\label{eq:rbc_rc1}
%\begin{split}
%    &\min_{\bm{a}\in \mathbb{R}^m } - \textstyle\frac{1}{N_o}\textstyle\sum_{i=1}^{N_o}h(\bm{a},\bm{x}'_i) \\
%    \text{s.t.}&\begin{cases}
%        h(\bm{a}, \bm{f}(\bm{x}^{(i)},\bm{d}^{(i)})) \geq \gamma h(\bm{a}, \bm{x}^{(i)}),  \\
%        \bm{a}[l]  \in [-U_a, U_a], \\
%        ~ i=1,\ldots,N;~ l=1,\ldots,m, \\
%    \end{cases}
%\end{split}
%\end{equation}
%

Let $\widehat{\textnormal{E}}_{(\bm{x},\bm{d})}[\phi_{\bm{a},\xi}]$ denote the empirical mean of $\phi_{\bm{a},\xi}(\bm{x},\bm{d})$ over the sample set %the set of $N$ i.i.d.\ disturbance samples 
$\mathbb{S}=\{(\bm{x}^{(i)},\bm{d}^{(i)})\}_{i=1}^N$, i.e.,
$\widehat{\textnormal{E}}_{(\bm{x},\bm{d})}[\phi_{\bm{a},\xi}] := \frac{1}{N} \sum_{i=1}^N \phi_{\bm{a},\xi}(\bm{x}^{(i)},\bm{d}^{(i)})$.
%Since there always exist solutions to \eqref{eq:rbc_rc}, we conclude that there always exists a solution $\big(\bm{a}(\mathbb{S}),\xi(\mathbb{S})\big)$ such that $\widehat{\textnormal{E}}_{(\bm{x},\bm{d})}[\phi_{\bm{a}(\mathbb{S}),\xi(\mathbb{S})}]=0$. According to Corollary~4 in~\cite{alamo2009randomized}, there exists an explicit bound on the number of samples $N$ required to ensure that the empirical probability $\widehat{\textnormal{E}}_{(\bm{x},\bm{d})}[\phi_{\bm{a}(\mathbb{S}),\xi(\mathbb{S})}]$ deviates from the true probability $\textnormal{E}_{(\bm{x},\bm{d})}[\phi_{\bm{a}(\mathbb{S}),\xi(\mathbb{S})}]$ by at most $\alpha_1$ from one side with confidence at least $1-\delta$. Specifically,
%\textcolor{blue}{$\textnormal{P}_{(\bm{x},\bm{d})}[\phi_{\bm{a}(\mathbb{S}),\xi(\mathbb{S})}]$, it seems that not correct? Mathematically, $\textnormal{P}_{(\bm{x},\bm{d})}[{\big\{(\bm{x},\bm{d})\mid h(\bm{a}, \bm{f}(\bm{x}, \bm{d})) < \gamma\, h(\bm{a}, \bm{x})-\xi\big\}}]=E[1_{\big\{(\bm{x},\bm{d})\mid h(\bm{a}, \bm{f}(\bm{x}, \bm{d})) < \gamma\, h(\bm{a}, \bm{x})-\xi\big\}}(\bm{x},\bm{d})]$. And $\textnormal{P}_{(\bm{x},\bm{d})}[1_{\big\{(\bm{x},\bm{d})\mid h(\bm{a}, \bm{f}(\bm{x}, \bm{d})) < \gamma\, h(\bm{a}, \bm{x})-\xi\big\}}(\bm{x},\bm{d})]$ is not defined?} 
Since \eqref{eq:rbc_rc} always has solutions, there exists $(\bm{a}(\mathbb{S}), \xi(\mathbb{S}))$ with  $\widehat{\textnormal{E}}_{(\bm{x},\bm{d})}[\phi_{\bm{a}(\mathbb{S}),\xi(\mathbb{S})}] = 0$. By Corollary 4 in~\cite{alamo2009randomized}, the value $N$ can be explicitly bounded to ensure that the empirical probability/estimate $\widehat{\textnormal{E}}_{(\bm{x},\bm{d})}[\phi_{\bm{a}(\mathbb{S}),\xi(\mathbb{S})}]$ differs from the true probability $
\textnormal{E}_{(\bm{x},\bm{d})}[\phi_{\bm{a}(\mathbb{S}),\xi(\mathbb{S})}]$
by at most $\alpha$ with confidence at least $1-\delta$.
Specifically,
\begin{equation}
\label{NN1}
\textstyle N \ge \frac{5}{\alpha} \left( \ln{\frac{4}{\delta}} + \operatorname{vc}(\mathbb{H}) \ln{\frac{40}{\alpha}} \right)
\end{equation}
can guarantee 
$\textnormal{P}_{\mathbb{S}}\!\left[\textnormal{E}_{(\bm{x},\bm{d})}[\phi_{\bm{a}(\mathbb{S}),\xi(\mathbb{S})}]\le \alpha \right] \ge 1-\delta$,
where $\textnormal{P}_{(\bm{x},\bm{d})}:=\textnormal{P}_{\bm{x}} \times \textnormal{P}_{\bm{d}}$.
\begin{lemma}[Corollary 4, \cite{alamo2009randomized}]
\label{coro4}
Suppose $\alpha_1 \in (0,1)$ and $\delta\in (0,1)$ are given. Then, the probability of one-sided constrained failure $p(N,\alpha_1,0)$ is smaller than $\delta$, where $p(N,\alpha_1,0)$ is defined as
\[p(N,\alpha_1,0):=\textnormal{P}_{\mathbb{S}}\left[
\begin{split}
&\exists (\bm{a},\xi)\in [-U_a,U_a]^m \times [0,\overline{\xi}], \\&\Big(
\widehat{\textnormal{E}}_{(\bm{x},\bm{d})}[\phi_{\bm{a},\xi}]\leq 0 \wedge \textnormal{E}_{(\bm{x},\bm{d})}[\phi_{\bm{a},\xi}]>\alpha\Big)
\end{split}
\right], \]  
if the sample size $N$ satisfies \eqref{NN1}.
\end{lemma}

Based on Lemma \ref{coro4} and Markov’s inequality, we conclude that, with confidence at least $1-\delta$ over the sample set $\mathbb{S}$, if the sampled program \eqref{eq:rbc_rc} returns $(\xi^*(\mathbb{S})=0, \bm{a}^*(\mathbb{S}))$, then at least a fraction $1-\alpha_1$ of the states $\bm{x} \in \mathbb{X}$ satisfy ``\text{the probability that } $h(\bm{a}^*(\mathbb{S}),\bm{f}(\bm{x},\bm{d})) \ge \gamma h(\bm{a}^*(\mathbb{S}),\bm{x})$ \text{ is at least} $1-\alpha_2$'', 
\[
\begin{split}
\text{i.e.,~}\textnormal{P}_{\mathbb{S}}\left[\xi^*(\mathbb{S})=0 \Rightarrow \textnormal{P}_{\bm{x}} \left[
\begin{split}
&\textnormal{P}_{\bm{d}}\left[
\begin{split}
   &h(\bm{a}^*(\mathbb{S}),\bm{f}(\bm{x},\bm{d}))\\
   &\ge\gamma\,h(\bm{a}^*(\mathbb{S}),\bm{x}) 
\end{split}
\right]\geq 1-\alpha_2
\end{split}
\right] \ge 1-\alpha_1\right]\geq 1-\delta,
\end{split}
\]
further implying $\mathbb{P}_{\mathbb{S}}\left[\xi^*(\mathbb{S})=0 \Rightarrow \textnormal{P}_{\bm{x}} \left[
\textnormal{P}_{\bm{d}}[\bm{f}(\bm{x},\bm{d}) \in \mathbb{X}] \geq 1-\alpha_2
\right] \ge 1-\alpha_1\right]\geq 1-\delta$.
%where 
%\[
%\mathbb{X}_s = \{\bm{x}\in\mathbb{X} \mid h(\bm a^*(\mathbb{S}),\bm{x}) > 0\}.
%\]
 The above statements hold provided that $N$ satisfies \eqref{NN1} with $\alpha=\alpha_1\alpha_2$. The proof is shown in Appendix \ref{appendix:a1}.
\begin{theorem}[PAC Safety Certification I Based on VC Dimension]
\label{thm:main-fixedpoint}
Fix probability threshold $\alpha_1, \alpha_2 \in (0,1)$  and cofidence level $\delta \in (0,1)$.  Let $(\bm{a}^*(\mathbb{S}), \xi^*(\mathbb{S}))$ denote the solution obtained by solving~\eqref{eq:rbc_rc}, %where $\mathbb{S} = \{(\bm{x}^{(i)}, \bm{d}^{(i)})\}_{i=1}^N$ is an i.i.d. sample drawn from $\mathbb{P}_{\bm{x}} \times \mathbb{P}_{\bm{d}}$. 
Define the computational procedure:
\begin{equation}
\label{cp}
\mathcal{A}(\mathbb{S}) = 1_{\{\xi^*(\mathbb{S}) = 0\}}(\mathbb{S}).
\end{equation}
If $N$ satisfies \eqref{NN1} with $\alpha=\alpha_1 \alpha_2$, then, with confidence at least $1 - \delta$ over the sample set $\mathbb{S}$, if $\xi^*(\mathbb{S})=0$, $
\textnormal{P}_{\bm{x}}\!\Big[
\textnormal{P}_{\bm{d}}\big[h(\bm{a}^*(\mathbb{S}),\bm{f}(\bm{x},\bm{d}))\geq \gamma\,h(\bm{a}^*(\mathbb{S}),\bm{x})] \geq 1-\alpha_2
\Big]
\ge 1-\alpha_1$.
This implies $\textnormal{P}_{\mathbb{S}}\Big[\mathcal{A}(\mathbb{S})\Rightarrow \textnormal{P}_{\bm{x}}\!\Big[
\textnormal{P}_{\bm{d}}\big[
\bm{f}(\bm{x},\bm{d})\in \mathbb{X}
\big] \geq 1-\alpha_2
\Big]
\ge 1-\alpha_1\Big]\geq 1-\delta$.
%where $\mathbb{X}_s=\{\bm{x}\in \mathbb{X}\mid h(\bm{a}^*(\mathbb{S}),\bm{x})\geq 0\}$.
%and 
%$\epsilon(\bm{x}):=\begin{cases}
%    \alpha_2, &\text{if~}\bm{x}\in \mathbb{X}_s,\\
%    1, &\text{otherwise}.
%\end{cases}$
\end{theorem}
%\begin{proof}
%    The proof is shown in Appendix \ref{appendix:a1}. \qed
%\end{proof}
\begin{remark}
    \label{rm:prob}
    As shown in the proof of Theorem \ref{thm:main-fixedpoint}, one can also derive $\textnormal{P}_{\mathbb{S}}\Big[\mathcal{A}(\mathbb{S})\Rightarrow \textnormal{P}_{(\bm{x},\bm{d})}\big[ \bm{f}(\bm{x},\bm{d})\in \mathbb{X} \big] \geq 1-\alpha \Big]\geq 1-\delta$, where the inner probability is taken with respect to the joint distribution $\textnormal{P}_{\bm{x}}\times\textnormal{P}_{\bm{d}}$. In contrast, the guarantee provided in Theorem \ref{thm:main-fixedpoint}, i.e., $\textnormal{P}_{\mathbb{S}}\Big[\mathcal{A}(\mathbb{S})\Rightarrow \textnormal{P}_{\bm{x}}\!\Big[ \textnormal{P}_{\bm{d}}\big[ \bm{f}(\bm{x},\bm{d})\in \mathbb{X} \big] \geq 1-\alpha_2 \Big] \ge 1-\alpha_1\Big]\geq 1-\delta$,   explicitly decouples the uncertainty arising from the state and disturbance distributions. This separation yields a more refined probabilistic characterization over the state and disturbance spaces, resulting in a clearer and more interpretable one-step safety guarantee for the system. \qed
\end{remark}

The solution to Problem~\ref{prob:1} is relatively general, since the function $h(\bm{a},\cdot): \mathbb{R}^n \to \mathbb{R}$ can be nonlinear in both $\bm{a}$ and $\bm{x}$. This allows the use of expressive models, such as neural networks (e.g., ReLU networks with nonnegative weights for $h_1$ in \eqref{barrier_form}) with finite pseudo-dimension. However, this generality requires many samples, which can be computationally expensive when solving the program \eqref{eq:rbc_rc}. To reduce this burden, we next focus on functions that are linear in the parameters $\bm{a}$, which lowers the sample complexity.
\subsection{PAC Safety Certification Based on Scenario Approaches}
\label{subsec:PACSA}
In this subsection, we present an approach that integrates scenario optimization with RBCs to address Problem \ref{prob:1}.

%\begin{assumption}
%\label{as:h}
%The function $h(\bm{a},\bm{x}): \mathbb{R}^n \to \mathbb{R}$ is linear in the parameters $\bm{a} \in [0,U_a]^m$ and continuous in $\bm{a}$ for any $\bm{x} \in \mathbb{R}^n$. Moreover, $h(\bm{a},\bm{x})$ is non-negative for $\bm{x} \in \mathbb{X}$ and negative for $\bm{x} \notin \mathbb{X}$, for all $\bm{a} \in [0,U_a]^m$.
%The parameterized function $h(\bm{a},\cdot): \mathbb{R}^n\rightarrow \mathbb{R}$ is linear over unknown parameters $\bm{a} \in [0,U_a]^m$, and continuous with respect to $\bm{a}\in [0,U_a]^m$ for any $\bm{x}\in \mathbb{R}^n$. Moreover, $h(\bm{a},\bm{x})$ is negative for all $(\bm{x},\bm{a})\in \mathbb{R}^n\setminus \mathbb{X} \times [0,U_a]^m$ and non-negative for all $(\bm{x},\bm{a})\in \mathbb{X}\times [0,U_a]^m$.
%\end{assumption}

We consider a class of functions $h(\bm{a},\bm{x})$ of the form \eqref{barrier_form}, where
%\begin{equation}
%\label{eq:h_r}
%h(\bm{a},\bm{x})=1_{\mathbb{X}} (\bm{x})h_1(\bm{a},\bm{x})+1_{\mathbb{R}^n\setminus \mathbb{X}}(\bm{x})C,
%\end{equation}
% \color{red}
% \begin{equation}
% h(\bm{a},\bm{x})=1_{\mathbb{X}} (\bm{x})-1_{\mathbb{R}^n\setminus \mathbb{X}}(\bm{x}),
% \end{equation}
% \begin{equation}
% h(\bm{a},\bm{x})=1_{\mathbb{X}} (\bm{x})a-1_{\mathbb{R}^n\setminus \mathbb{X}}(\bm{x}),
% \end{equation}
% \color{black} 
\begin{equation}
\label{eq:h_r}
h_1(\bm{a},\bm{x})=\sum_{i_{1}=0}^{\kappa}\ldots\sum_{i_{n}=0}^{\kappa} a_{i_{1},\ldots,i_{n}} \prod_{j=1}^n \big(\bm{x}[j] - \underline{\mathbb{X}}[j]\big)^{i_{j}}\big(\overline{\mathbb{X}}[j] - \bm{x}[j]\big)^{\kappa-i_{j}}.
\end{equation}
{This representation, known as the Handelman Representation, offers key computational advantages. By restricting the coefficients $a_{i_1,\ldots,i_n}$ to be non-negative, the polynomial $h_1(\bm{a}, \bm{x})$ is guaranteed to be non-negative on the interval $[\underline{\mathbb{X}}, \overline{\mathbb{X}}]$, which simplifies the enforcement of the positivity constraint. Furthermore, the representation remains linear in the decision variables $\bm{a}$, preserving the convexity of the optimization problem while maintaining the capability to approximate positive polynomials.}
% where $a_{i_1,\ldots,i_n} \geq 0$.
With $a_{i_1,\ldots,i_n} \in [0, U_a]$ and $\underline{\mathbb{X}}, \overline{\mathbb{X}} \subseteq \mathbb{R}^n$ such that $[\underline{\mathbb{X}}, \overline{\mathbb{X}}]$ is an interval containing the safe set $\mathbb{X}$, we have that $h(\bm{a},\bm{x}) \equiv C < 0, \forall (\bm{x},\bm{a}) \in \mathbb{R}^n \setminus \mathbb{X} \times [0,U_a]^m$ and $h(\bm{a},\bm{x}) \geq 0, \forall (\bm{x},\bm{a}) \in \mathbb{X} \times [0,U_a]^m$. Thus, for any $\bm{a}\in [0,U_a]^m$, $h(\bm{a},\bm{x})$ satisfies constraints \eqref{eq:rbf1} and \eqref{eq:rbf2}. %Therefore, we only need to solve the parameter $\bm{a}$ so that $h(\bm{a},\bm{x})$ satisfies the constraint \eqref{eq:rbf3}, which can be transformed into the following uncertain convex optimization problem:
%\begin{equation}
%\begin{split}
%    &\min_{\xi\in\mathbb{R},\bm{a}\in \mathbb{R}^m } \xi \label{eq:rbc_linear1} \\
%    \text{s.t.}&\begin{cases}
%        h(\bm{a}, \bm{f}(\bm{x},\bm{d})) \geq \gamma h(\bm{a}, \bm{x})-\xi, \forall \bm{x} \in \mathbb{X}, \forall \bm{d} \in \mathbb{D};\\
%        \xi \in [0,\overline{\xi}]; ~ \bm{a}[l]  \in [0, U_a]; ~ l=1,\ldots,m, 
%    \end{cases} 
%\end{split}
%\end{equation}

%Let $\bm{a}^*$ and $\xi^*$ be the optimal solution to \eqref{eq:rbc_linear1}. If $\xi^* = 0$, we conclude that $h(\bm{a}^*, \bm{x})$ is a RBC. However, optimization problem \eqref{eq:rbc_linear1} is a robust optimization, which is fundamentally challenging to solve. Moreover, the lack of knowledge about the system dynamics $\bm{f}(\bm{x}, \bm{d})$ further complicates the solution process. To address these challenges, we extract $N\times M$ samples of the form $\{(\bm{x}^{(i)}, \bm{x}^+_{i,j})\}_{i=1,j=1}^{N, M}$ satisfying Assumption \ref{as:sample} and use these samples to solve the optimization \eqref{eq:rbc_linear1}. 

With $\mathbb{S}=\{(\bm{x}^{(i)},\bm{d}^{(i)})\}_{i=1}^N$ and the constraint \eqref{eq:rbf3}, %where $(\bm{x}^{(i)},\bm{d}^{(i)})\stackrel{\text{i.i.d.}}{\sim} \mathbb{P}_{\bm{x}} \times \mathbb{P}_{\bm{d}}$, 
we solve the optimization \eqref{eq:rbc_rc}, which is a linear optimization over $(\bm{a},\xi)$.
In this context, the optimization \eqref{eq:rbc_rc} is a scenario optimization to the following robust convex optimization:
\begin{equation}
    \label{eq:rbc_linear1}
  \begin{split}
    &\textstyle \min_{\xi\in\mathbb{R},\bm{a}\in \mathbb{R}^m } \xi \\
    \text{s.t.}&
    \begin{cases}
        h(\bm{a}, \bm{f}(\bm{x},\bm{d})) \geq \gamma h(\bm{a}, \bm{x}) - \xi, \forall \bm{x}\in \mathbb{X},\forall \bm{d} \in \mathbb{D},  \\
        \xi \in [0,\overline{\xi}]; ~ \bm{a}[k]  \in [0, U_a], \\
         k=1,\ldots,m. \\
    \end{cases}
    \end{split}
\end{equation}
It is observed that optimization \eqref{eq:rbc_linear1} satisfies Assumption \ref{as:convex} and the optimization \eqref{eq:rbc_rc} satisfies Assumption \ref{as:solution}. Following Proposition \ref{prop:pac}, we have the conclusion below. The corresponding proof is shown in Appendix \ref{appendix:proof_theo2}.
\begin{theorem}[PAC Safety Certification I Based on Scenario Approaches]
\label{theo:robust}
%Let $\mathbb{S} = \{(\bm{x}^{(i)},\bm{d}^{(i)})\}_{i=1}^N$ be an i.i.d. sample drawn from $\mathbb{P}_{\bm{x}} \times \mathbb{P}_{\bm{d}}$. 
Fix probability thresholds $\alpha_1, \alpha_2 \in (0,1)$ and confidence level $\delta \in (0,1)$, and let the sample size satisfy
$N \;\geq\; \frac{2}{\alpha_1 \alpha_2}\!\left(\ln\frac{1}{\delta} + m + 1\right)$.

Let $(\bm{a}^*(\mathbb{S}), \xi^*(\mathbb{S}))$ be the solution to~\eqref{eq:rbc_rc}, with the computational procedure in~\eqref{cp}. Then,  with confidence at least $1 - \delta$ over the random draw of $\mathbb{S}$, if $\xi^*(\mathbb{S}) = 0$, 
$\textnormal{P}_{\bm{x}}\!\Big[
\textnormal{P}_{\bm{d}}\!\big[
h(\bm{a}^*(\mathbb{S}), \bm{f}(\bm{x},\bm{d})) \geq \gamma h(\bm{a}^*(\mathbb{S}), \bm{x})
\big] \geq 1- \alpha_2
\Big]
\;\ge\; 1-\alpha_1$ holds, i.e., $\textnormal{P}_{\mathbb{S}}\Big[
\mathcal{A}(\mathbb{S}) \Rightarrow 
\textnormal{P}_{\bm{x}}\Big[\textnormal{P}_{\bm{d}}\!\big[
\bm{f}(\bm{x},\bm{d}) \in \mathbb{X} 
\big] \ge 1 - \alpha_2
\Big]
\ge 1 - \alpha_1
\Big]
\ge 1 - \delta$.
%where $\mathbb{X}_s = \big\{\bm{x} \in \mathbb{X} \mid h(\bm{a}^*(\mathbb{S}), \bm{x}) \ge 0 \big\}$.%and $\epsilon(\bm{x}) =
%\begin{cases}
%\alpha_2, & \text{if } \bm{x} \in \mathbb{X}_s,\\[3pt]
%1, & \text{otherwise.}
%\end{cases}$
\end{theorem}
%\begin{proof}
%    The proof is shown in Appendix \ref{appendix:proof_theo2}.
%\end{proof}
%\dw{Why is this better than Theorem 1?}
When $h(\bm{a},\bm{x})$ takes the form \eqref{barrier_form} with $h_1(\bm{a},\bm{x})$ as in \eqref{eq:h_r}, Theorem \ref{thm:main-fixedpoint} requires $N \ge \frac{5}{\alpha_1 \alpha_2} \left( \ln{\frac{4}{\delta}} + (m+1) \ln{\frac{40}{\alpha_1 \alpha_2}} \right)$, whereas Theorem \ref{theo:robust} only requires $N \ge \frac{2}{\alpha_1 \alpha_2}\left(\ln\frac{1}{\delta} + m + 1\right)$.
Clearly, Theorem \ref{theo:robust} achieves the same probabilistic guarantee with significantly fewer samples.
\section{PAC Safety Certification with One-to-Many Samples}
\label{sec:nested}
In this section, we present our methods for solving Problem~\ref{prob:2} and Problem~\ref{prob:3}. Subsection~\ref{sec:probabilistic} addresses Problem~\ref{prob:2} using RBCs with scenario approaches, Markov’s inequality, and Hoeffding’s inequality, while Subsection~\ref{sub:pacnpgsb} tackles Problem~\ref{prob:3} using stochastic barrier certificates (SBCs) with the same tools. 
\subsection{PAC Safety Certification for Problem \ref{prob:2}}
\label{sec:probabilistic}
In this section, we present our data-driven approach to address Problem~\ref{prob:2}, which primarily involves solving a linear optimization problem. 

 We work with the class of functions $h(\bm{a},\bm{x})$ defined by \eqref{barrier_form}, with $h_1$ given by \eqref{eq:h_r}. Then, we search for $\bm{a}$ by solving the following optimization \eqref{eq:rbc_linear3}, which is constructed based on the condition \eqref{eq:rbf3} and sample set $\mathbb{S}=\{(\bm{x}^{(i)},\{\bm{d}^{(i,j)}\}_{j=1}^M)\}_{i=1}^N$, %where $\{(\bm{x}^{(i)}, \{\bm{d}^{(i,j)}\}_{j=1}^M)\}_{i=1}^N\stackrel{\text{i.i.d.}}{\sim} \textnormal{P}_{\bm{x}} \times \textnormal{P}_{\bm{d}}^M$ and  $\{\bm{d}^{(i,j)}\}_{j=1}^M \stackrel{\text{i.i.d.}}{\sim} \textnormal{P}_{\bm{d}}$, 
\begin{equation}
    \label{eq:rbc_linear3}
  \begin{split}
    &\textstyle \min_{\xi\in\mathbb{R},\bm{a}\in \mathbb{R}^m } \xi \\
    \text{s.t.}&
    \begin{cases}
        h(\bm{a}, \bm{f}(\bm{x}^{(i)},\bm{d}^{(i,j)})) \geq \gamma h(\bm{a}, \bm{x}^{(i)}) - \xi,  \\
        \xi \in [0,\overline{\xi}]; ~ \bm{a}[k]  \in [0, U_a], \\
        i=1,\ldots,N;~ j=1,\ldots,M;~ k=1,\ldots,m, \\
    \end{cases}
    \end{split}
\end{equation}
where $\gamma \in (0,1)$, and $U_a$ and $\overline{\xi}$ are the same with the ones in \eqref{eq:rbc_rc}. It is a scenario optimization to the following robust convex optimization:
\begin{equation}
\label{eq:rbc_linear1_M}
  \begin{split}
    &\textstyle \min_{\xi\in\mathbb{R},\bm{a}\in \mathbb{R}^m } \xi \\
    \text{s.t.}&
    \begin{cases}
        \wedge_{j=1}^M h(\bm{a}, \bm{f}(\bm{x},\bm{d}_j)) \geq \gamma h(\bm{a}, \bm{x}) - \xi, \forall (\bm{x},\{\bm{d}_j\}_{j=1}^M) \in \mathbb{X}\times \mathbb{D}^M,  \\
        \xi \in [0,\overline{\xi}]; ~ \bm{a}[k]  \in [0, U_a];~ k=1,\ldots,m. \\
    \end{cases}
    \end{split}
\end{equation}

%Similar to the approaches in Section \ref{sec:verification}, if $\xi^*(\mathbb{S})=0$ is computed via solving \eqref{eq:rbc_linear3}, we proceed with solving the optimization \eqref{eq:rbc_rc1}. 
It is observed that the optimization \eqref{eq:rbc_linear1_M} satisfies Assumption \ref{as:convex} and the optimization \eqref{eq:rbc_linear3} satisfies Assumption \ref{as:solution}. Assume a solution $(\xi^*(\mathbb{S}), \bm{a}^*(\mathbb{S}))$ is obtained by solving \eqref{eq:rbc_linear3}. Then, using scenario approach theory and Markov’s inequality, we conclude: with confidence at least $1-\delta_1$ over the sample set $\mathbb{S}$, if $\xi^*(\mathbb{S}) = 0$, then for at least a $1 - \frac{\alpha_1}{l\delta_2}$ fraction of states $\bm{x} \in \mathbb{X}$, the system \eqref{eq:system}, starting from $\bm{x}$, remains safe at the next step with probability at least $1-\alpha_2$, provided that $\alpha_1,\alpha_2,\delta_1,\delta_2, l \in (0,1)$, $\alpha_1 < l\delta_2$, $N \ge \frac{2}{\alpha_1}\Big(\ln\frac{1}{\delta_1} + m + 1\Big)$, and $M \ge \frac{1}{2\alpha_2^2}\ln\frac{1}{(1-l)\delta_2}$. The proof is shown in Appendix \ref{appendix:a2}.
\begin{theorem}[PAC Safety Certification II Based on RBC]
\label{thm:main-vc-final-corrected}
Fix $\delta_1 \in (0,1)$, $\delta_2 \in (0,1)$, $\alpha_1 \in (0,1)$, $\alpha_2 \in (0,1)$, and $l\in (0,1)$, where $\alpha_1<l\delta_2$.  %Let $\mathbb{S}=\{(\bm{x}^{(i)},\{\bm{d}^{(i,j)}\}_{j=1}^M)\}_{i=1}^N$ is an i.i.d. sample drawn from  $\textnormal{P}_{\bm{x}} \times \textnormal{P}_{\bm{d}}^M$, where $\{\bm{d}^{(i,j)}\}_{j=1}^M \stackrel{\text{i.i.d.}}{\sim} \textnormal{P}_{\bm{d}}$, 
Let $
N \geq \frac{2}{\alpha_1}\left(\ln\frac{1}{\delta_1} + m+1\right)$ and $M \geq \frac{1}{2\alpha_2^2}\ln\frac{1}{(1-l)\delta_2}$.% Define
%\[h_{\bm{a},\xi}(\bm{x},\{\bm{d}_i\}_{i=1}^M) := 1_{\big\{(\bm{x},\{\bm{d}_i\}_{i=1}^M)\mid \min_{i=1}^M h(\bm{a},\bm{f}(\bm{x},\bm{d}_j)) < \gamma h(\bm{a},\bm{x})-\xi\big\}}(\bm{x},\{\bm{d}_i\}_{i=1}^M),
%\]where

Let $(\bm{a}^*(\mathbb{S}), \xi^*(\mathbb{S}))$ be the solution obtained by solving~\eqref{eq:rbc_linear3}.  Then, we have 
$\textnormal{P}_{\mathbb{S}}\left[\textnormal{P}_{\bm{x}}\left[
\textnormal{P}_{\bm{d}}\left[
h(\bm{a}^*(\mathbb{S}),\bm{f}(\bm{x},\bm{d})) \geq  \gamma h(\bm{a}^*(\mathbb{S}),\bm{x})-\xi^*(\mathbb{S})
\right] \geq 1-\alpha_2 
\right] \geq 1 - \frac{\alpha_1}{l\delta_2}\right]\geq 1-\delta_1$
%\end{split}
%\end{equation*}
%\dw{I don't see where the $\bm d$ of $\textnormal{P}_{\bm{d}}^M$ appear within its body $\textnormal{P}_{\bm{d}}^M[\cdot]$}\textcolor{red}{Since the related term equals zero, it cancels out. Please refer to Step 3 in the proof. It is quite okay that this term does not appear when using Hoeffiding inequality in certain cases, such as when the condition is always true.}
and thus, $\textnormal{P}_{\mathbb{S}}\left[\xi^*(\mathbb{S})=0\Rightarrow \textnormal{P}_{\bm{x}}\left[
\textnormal{P}_{\bm{d}}\left[
\bm{f}(\bm{x},\bm{d}) \in \mathbb{X}
\right] \geq 1-\alpha_2
\right] \geq 1 - \frac{\alpha_1}{l\delta_2}\right]\geq 1-\delta_1$.
\end{theorem}
%\begin{proof}
%    The proof is shown in Appendix \ref{appendix:a2}. \qed
%\end{proof}

\begin{remark}
\label{remark_com_single_and_nested}
For $M>1$, each sampled state is paired with multiple disturbance samples. If $\xi^*(\mathbb{S}) = 0$, none of these disturbances lead to a violation, so the probability $\textnormal{P}_{\bm{d}}[\bm{f}(\bm{x},\bm{d}) \in \mathbb{S}]$ is generally higher than in the $M=1$ case. This one-to-many state–disturbance sampling is particularly advantageous when only a limited number of states can be sampled. {Intuitively, this improvement stems from decoupling the verification of state and disturbance uncertainties. By evaluating each state against multiple disturbances, the algorithm effectively `averages out' the disturbance noise locally, obtaining a cleaner estimate of each state's safety level. This reduces the variance seen by the outer optimization, allowing it to potentially establish stronger guarantees (lower $\alpha_1$) with significantly fewer state samples ($N$) compared to the single-sample approach.} On the other hand, when $M=1$, Theorem~\ref{thm:main-vc-final-corrected} can theoretically be reduced to Theorem~\ref{theo:robust}, since the portion of the proof relying on Hoeffding’s inequality is no longer needed (see the proof). However, the final guarantee in Theorem~\ref{thm:main-vc-final-corrected} for $M=1$ does not exactly recover Theorem~\ref{theo:robust} and may even be looser in practice, which can be verified via a simple calculation. This looseness stems from our use of Hoeffding’s inequality to handle the case 
$M>1$ (see \textbf{Step 3. Application of Hoeffding’s Inequality} in the proof), which is inherently conservative and leads to a looser probability bound. In future work, we plan to reduce this gap by exploring alternative concentration inequalities. \qed
\end{remark}

\begin{remark}
\label{correction}
When proving Theorem \ref{thm:main-vc-final-corrected}, we obtain the nested PAC certification:
\begin{equation*}
\label{eq:nested_guarantee_final01}
\textnormal{P}_{\mathbb{S}}\left[\textnormal{P}_{\bm{x}}\left[
\textnormal{P}_{\bm{d}}^M\left[
\mu(\bm{a}^*(\mathbb{S}),\bm{x},\xi^*(\mathbb{S}))
\geq 1-\alpha_2 
\right]\geq 1-\delta_2
\right] \geq 1 - \frac{\alpha_1}{l\delta_2}\right]\geq 1-\delta_1.
\end{equation*}
where $\mu(\bm{a}^*(\mathbb{S}),\bm{x},\xi^*(\mathbb{S})):=\textnormal{P}_{\bm{d}}\left[ h(\bm{a}^*(\mathbb{S}),\bm{f}(\bm{x},\bm{d})) \geq \gamma h(\bm{a}^*(\mathbb{S}),\bm{x})-\xi^*(\mathbb{S})\right]$. Further, since $\mu(\bm{a}^*(\mathbb{S}),\bm{x},\xi^*(\mathbb{S})) \geq 1-\alpha_2$ is deterministic with respect to $\textnormal{P}_{\bm{d}}^M$, $\textnormal{P}_{\bm{d}}^M[\mu(\bm{a}^*(\mathbb{S}),\bm{x},\xi^*(\mathbb{S})) \geq 1-\alpha_2]\in \{0,1\}$. Also, since $\textnormal{P}_{\bm{d}}^M[\mu(\bm{a}^*(\mathbb{S}),\bm{x},\xi^*(\mathbb{S})) \geq 1-\alpha_2] > 0$ (as $1-\delta_2 > 0$), $\mu(\bm{a}^*(\mathbb{S}),\bm{x},\xi^*(\mathbb{S}))\geq 1-\alpha_2$ holds deterministically with respect to $\textnormal{P}_{\bm{d}}^M$, which yields the conclusion of Theorem~\ref{thm:main-vc-final-corrected}. \textit{For more details, please refer to the proof.}

In Theorem 4 of \cite{xue2020pac}, a nested PAC characterization similar to the above is established. The analysis is over $(\bm{x},t)$ rather than $(\bm{x},\bm{d})$, with $t$ representing time. The result is obtained by applying the scenario approach twice: once for the outer layer $\bm{x}$ and once for the inner layer $t$. However, this method reuses the same samples for $t$ across both layers, creating a data-dependence issue that invalidates the PAC guarantee. Our approach resolves this by using Markov's inequality and fresh disturbance samples $\{\bm{d}_j\}_{j=1}^M$ in the inner confidence characterization $\textnormal{P}_{\bm{d}}^M$, which are independent of the samples in $\mathbb{S}$.  \qed
\end{remark}

If a more general parametric function $h_1(\bm{a},\bm{x})$ is used, e.g., one that is nonlinear in $\bm{a}$ rather than linear, then the VC-dimension method in Theorem \ref{thm:main-fixedpoint} can replace the scenario approach. We leave this for future work.
\subsection{PAC Safety Certification for Problem \ref{prob:3}}
\label{sub:pacnpgsb}
In this section, we address Problem \ref{prob:3} using SBCs. The concept of SBCs adopted here is based on the formulation proposed in \cite{prajna2007convex}.

\begin{definition} [\cite{prajna2007convex}]
    \label{def:sbf}
    Given the safe set $\mathbb{X}$, a function $h(\cdot):\mathbb{R}^n\rightarrow\mathbb{R}$ is a \textnormal{SBC} for the system 
    \eqref{eq:system} if there exists a scalar $\lambda \in [0,1]$ satisfying the condition:
    \begin{subequations}
    \label{eq:sbf}
    \begin{empheq}[left=\empheqlbrace]{align}
        &h(\bm{x}) \geq 0, &\forall \bm{x}\in \mathbb{X}, \label{eq:sbf1} \\
        % &h(\bm{x}) \leq M, \quad \forall \bm{x}\in \mathbb{X}, \label{eq:bf2} \\
        &h(\bm{x}) \geq 1, & \forall \bm{x} \in \mathbb{R}^n\setminus\mathbb{X}, \label{eq:sbf2} \\
        &\textnormal{E}_{\bm{d}}[h(\bm{f}(\bm{x},\bm{d}))] - h(\bm{x}) \leq \lambda, &\forall \bm{x} \in \mathbb{X}. \label{eq:sbf3}
    \end{empheq}
    \end{subequations}
\end{definition}

\begin{proposition}[\cite{prajna2007convex}]
\label{pro_one}
    If there exist a \textnormal{SBC} $h(\cdot): \mathbb{R}^n\rightarrow\mathbb{R}$ and a scalar $\lambda\in[0,1]$, then $\textnormal{P}_{\bm{d}}[\bm{f}(\bm{x}, \bm{d}) \in \mathbb{X}] \geq 1 - \lambda - h(\bm{x}), \forall  \bm{x}\in\mathbb{X}$.
\end{proposition}

Similarly, we parameterize $h(\bm{a},\bm{x})$ as follows:
\begin{equation}
\label{eq:h_p}
h(\bm{a},\bm{x})=1_{\mathbb{X}} (\bm{x})h_1(\bm{a},\bm{x})+1_{\mathbb{R}^n\setminus \mathbb{X}}(\bm{x}),
\end{equation}
where $h_1(\bm{a},\bm{x}): =\sum_{i_1=0}^{\kappa}\ldots\sum_{i_n=0}^{\kappa} a_{i_1,\ldots,i_n} \prod_{j=1}^n\dbinom{{\kappa}}{i_j}  \psi_j^{i_j}(1-\psi_j)^{{\kappa}-i_j}$ is a Bernstein polynomial, $a_{i_1,\ldots,i_n} \in [0, U_a]$, $U_a\geq 1$, and $\psi_j = \frac{\bm{x}[j]-\underline{\mathbb{X}}[j]}{\overline{\mathbb{X}}[j]-\underline{\mathbb{X}}[j]}$. 
Parameterization in the form of \eqref{eq:h_p} allows constraints \eqref{eq:sbf1} and \eqref{eq:sbf2} to hold naturally. Moreover, compared with the polynomial in \eqref{eq:h_r}, parameterization in  \eqref{eq:h_p} ensures that $h(\bm{a},\bm{x}) \leq U_a$ for all $(\bm{x},\bm{a}) \in \mathbb{X} \times [0,U_a]^m$ \cite{farouki2012bernstein}. Moreover, we obtain that $h(\bm{a},\bm{x})\geq 0, \forall (\bm{a},\bm{x}) \in [0,U_a]^m \times \mathbb{X}$ and $h(\bm{a},\bm{x})\equiv 1, \forall (\bm{a},\bm{x}) \in [0,U_a]^m \times \mathbb{R}^n\setminus \mathbb{X}$.

%which is important for the validity of Theorem \ref{thm:main-vc-final-sbc1}.

We search for $\bm{a}$ by solving the  optimization \eqref{eq:sbc_linear3} below, which is constructed using the condition \eqref{eq:sbf3} and sample set  $\mathbb{S}=\{(\bm{x}^{(i)}, \{\bm{d}^{(i,j)}\}_{j=1}^M)\}_{i=1}^N$,%\stackrel{\text{i.i.d.}}{\sim} \textnormal{P}_{\bm{x}} \times \textnormal{P}_{\bm{d}}^M$ and  $\{\bm{d}^{(i,j)}\}_{j=1}^M \stackrel{\text{i.i.d.}}{\sim} \textnormal{P}_{\bm{d}}$,
\begin{equation}
   \label{eq:sbc_linear3}
   \begin{split}
    &\textstyle \min_{\lambda\in \mathbb{R}, ~\bm{a}\in \mathbb{R}^m} ~\lambda + \frac{1}{N_o}\textstyle\sum_{k=1}^{N_o}h(\bm{a},\bm{x}'_k)\\
    \text{s.t.}&
    \begin{cases}
        \frac{1}{M}\sum_{j=1}^{M}h(\bm{a}, \bm{f}(\bm{x}^{(i)},\bm{d}^{(i,j)})) \leq h(\bm{a},\bm{x}^{(i)})  + \lambda - \tau, \\
        \lambda \in [0,1]; ~ \bm{a}[k] \in [0, U_a]; ~ i = 1,\ldots,N;~ k=1,\ldots,m,
    \end{cases}
    \end{split}
\end{equation}
where $\tau \in (0, 1)$ is a specified positive value and $\{\bm{x}'_i\}_{i=1}^{N_o}$ is a family of specified states distributed evenly over $\mathbb{X}$, which can be generated by random sampling from $\mathbb{X}$. This data set $\{\bm{x}'_i\}_{i=1}^{N_o}$ is fixed over the whole sampling process and is independent of $\mathbb{S}$. We can also see that \eqref{eq:sbc_linear3} is feasible:  the constant function 
$h_1(\bm{x}) \equiv U_a \in [0, U_a]$ on $\mathbb{X}$, together with $\lambda=\tau$, satisfies the constraints.  %Similarly, \eqref{eq:sbc_linear3} satisfies Assumption \ref{as:solution}. 
\begin{remark}
Ideally, program \eqref{eq:sbc_linear3} searches for a function $\lambda + h(\bm{a},\bm{x})$ that stays small over all $\bm{x}\in\mathbb{X}$, giving the tightest lower bound on the one-step safety probability. The natural objective would be the integral $\lambda + \int_{\mathbb{X}} h(\bm{a},\bm{x}), d\bm{x}$, but this integral is usually intractable. We therefore approximate it using samples, replacing the integral by the empirical sum $\sum_{k=1}^{N_o} h(\bm{a},\bm{x}'_k)$.
\end{remark}

It is observed that the linear problem \eqref{eq:sbc_linear3} is a scenario optimization to the following uncertain convex optimization \eqref{eq:sbc_un_con}: 
\begin{equation}
   \label{eq:sbc_un_con}
   \begin{split}
    &\textstyle \min_{\lambda\in \mathbb{R}, ~\bm{a}\in \mathbb{R}^m}~ \lambda + \frac{1}{N_o}\textstyle\sum_{k=1}^{N_o}h(\bm{a},\bm{x}'_k)\\
    \text{s.t.}&
    \begin{cases}
        \frac{\sum_{j=1}^{M}h(\bm{a}, \bm{f}(\bm{x},\bm{d}_{j}))}{M} \leq h(\bm{a},\bm{x})  + \lambda - \tau, ~\forall (\bm{x},\{\bm{d}_j\}_{j=1}^M) \in \mathbb{X}\times \mathbb{D}^M, \\
        \lambda \in [0,1]; ~ \bm{a}[k] \in [0, U_a];~ k=1,\ldots,m.
    \end{cases}
    \end{split}
\end{equation}

%To find $\bm{a}$ such that $h(\bm{a},\bm{x})$ satisfies constraint \eqref{eq:sbf3}, we formulate the following linear program:
%\begin{equation}
%    \label{eq:sbc_linear1}
%  \begin{split}
%    &\min_{\lambda \in \mathbb{R}, ~\bm{a}\in \mathbb{R}^m} \lambda + \textstyle\frac{1}{N_o}\textstyle\sum_{k=1}^{N_o}h(\bm{a},\bm{x}'_k)\\
%    \text{s.t.}&
%    \begin{cases}
%        \mathbb{E}_{\bm{d}}[h(\bm{a}, \bm{f}(\bm{x},\bm{d}))] \leq h(\bm{a}, \bm{x}) + \lambda, ~ \forall \bm{x} \in \mathbb{X}; \\
%        \lambda \in [0,1]; ~ \bm{a}[l] \in [0, U_a]; ~ l=1,\ldots,m; 
%    \end{cases}
%    \end{split}
%\end{equation}
%where $\{\bm{x}'_k\}_{k=1}^{N_o}$ is a family of specified states distributed evenly over $\mathbb{X}$, which can be generated by random sampling from a uniform distribution over the state space $\mathbb{X}$. 

Condition \eqref{eq:sbc_un_con} meets Assumption~\ref{as:convex}, and \eqref{eq:sbc_linear3} meets Assumption~\ref{as:solution}.  Let $(\bm{a}^*(\mathbb{S}), \lambda^*(\mathbb{S}))$ be the optimal solution of \eqref{eq:sbc_linear3}.  
By the scenario approach, Markov's inequality,  and Hoeffding's inequality, with confidence at least $1 - \delta_1$ over the sampled set $\mathbb{S}$, at least a $1 - \frac{\alpha_1}{l\delta_2}$ fraction of states $\bm{x} \in \mathbb{X}$ satisfy: the system \eqref{eq:system}, starting from such $\bm{x}$, stays safe at the next step with probability at least
 $1-h(\bm{a}^*(\mathbb{S}),\bm{x})-\lambda^*(\mathbb{S})$, provided $\alpha_1,\delta_1,\delta_2, l\in (0,1)$, $
N \geq \frac{2}{\alpha_1}\left(\ln\frac{1}{\delta_1} + m+1\right)$, and $M \geq \frac{U_a^2}{2\tau^2}\ln\frac{1}{(1-l)\delta_2}$.
 The proof is shown in Appendix \ref{appendix:a3}.
\begin{theorem}[PAC Safety Certification III Based on SBC]
\label{thm:main-vc-final-sbc1}
Fix $\delta_1 \in (0,1)$, $\delta_2 \in (0,1)$, $\alpha_1 \in (0,1)$, $\alpha_2 \in (0,1)$, and $l\in (0,1)$, where $\alpha_1<l\delta_2$, % Let $\mathbb{S}=\{(\bm{x}^{(i)},\{\bm{d}^{(i,j)}\}_{j=1}^M)\}_{i=1}^N$ is an i.i.d. sample drawn from  $\textnormal{P}_{\bm{x}} \times \textnormal{P}_{\bm{d}}^M$, where $\{\bm{d}^{(i,j)}\}_{j=1}^M \stackrel{\text{i.i.d.}}{\sim} \textnormal{P}_{\bm{d}}$, $
and let $N \geq \frac{2}{\alpha_1}\left(\ln\frac{1}{\delta_1} + m+1\right)$, and $M \geq \frac{U_a^2}{2\tau^2}\ln\frac{1}{(1-l)\delta_2}$.

Let $(\bm{a}^*(\mathbb{S}), \lambda^*(\mathbb{S}))$ be the solution obtained by solving~\eqref{eq:sbc_linear3}. Then, we have 
$\textnormal{P}_{\mathbb{S}}\left[\textnormal{P}_{\bm{x}}\left[
\textnormal{E}_{\bm{d}}[h(\bm{a}^*(\mathbb{S}),\bm{f}(\bm{x},\bm{d}))] \leq  h(\bm{a}^*(\mathbb{S}),\bm{x})+\lambda^*(\mathbb{S})
\right] \geq 1 - \frac{\alpha_1}{l\delta_2}\right]\geq 1-\delta_1$ and thus $\textnormal{P}_{\mathbb{S}}\left[\textnormal{P}_{\bm{x}}\left[
\textnormal{P}_{\bm{d}}[\bm{f}(\bm{x},\bm{d})\in \mathbb{X}]\geq 1-h(\bm{a}^*(\mathbb{S}),\bm{x})-\lambda^*(\mathbb{S})
\right] \geq 1 - \frac{\alpha_1}{l\delta_2}\right]\geq 1-\delta_1$.
%\begin{equation*}
% \label{eq:nested_guarantee_final0}
%\begin{split}
%\textnormal{P}_{\mathbb{S}}\left[\textnormal{P}_{\bm{x}}\left[
%\begin{split}
%\begin{split}
%\textnormal{E}_{\bm{d}}[h(\bm{a}^*(\mathbb{S}),\bm{f}(\bm{x},\bm{d}))] \leq  h(\bm{a}^*(\mathbb{S}),\bm{x})+\lambda^*(\mathbb{S})
%\end{split}
%\end{split}
%\right] \geq 1 - \frac{\alpha_1}{l\delta_2}\right]\geq 1-\delta_1,
%\end{split}
%\end{equation*}
%and thus, 
%\begin{equation*}
%\label{eq:nested_guarantee_final1}
%\begin{split}
%\textnormal{P}_{\mathbb{S}}\left[\textnormal{P}_{\bm{x}}\left[
%\begin{split}
%\textnormal{P}_{\bm{d}}^M\left[
%\begin{split}
%&\textnormal{E}_{\bm{d}}[h(\bm{a}^*(\mathbb{S}),\bm{f}(\bm{x},\bm{d}))]\\
%&> h(\bm{a}^*(\mathbb{S}),\bm{x})
%\end{split}
%\right] \geq 1-\delta_2
%\end{split}
%\right] \geq 1 - \frac{\alpha_1}{l\delta_2}\right]\geq 1-\delta_1,
%\end{split}
%\end{equation*}
%implying 
\end{theorem}
%\begin{proof}
%    The proof is shown in Appendix \ref{appendix:a3}. \qed
%\end{proof}

Similarly, if a more general parametric function $h_1(\bm{a},\bm{x})$ is used, e.g., one that is nonlinear in $\bm{a}$ rather than linear, then the VC-dimension method in Theorem \ref{thm:main-fixedpoint} can replace the scenario approach. We leave this for future work.

%It is worth noting that the same analysis described in Remark~\ref{remark1} can be extended to more general forms of $h(\bm{a},\bm{x})$ by applying Lemma~\ref{coro4} in conjunction with the Markov's and Hoeffding's inequalities.

%Here, $\epsilon : \mathbb{X} \rightarrow [0,1]$ denotes a state-dependent probability threshold that may depend on the sampled data $\mathbb{S}$. 

%for each $\bm{x}$, with confidence at least $1-\beta_2$, the probability that the system \eqref{eq:system} starting from $\bm{x}^{(i)}$, remains within the safe set $\mathbb{X}$ at the next step is at least $1-\lambd\bm{a}^*-h(\bm{a}^*,\bm{x})$, provided that \textcolor{red}{$M \geq \frac{U_a^2}{\tau^2} \left( 2\sqrt{m} + \sqrt{2 \ln \frac{1}{\beta_2}} \right)^2$} holds. 
%where $D=\max_{\bm{d}\in \mathcal{D}}\|\bm{d}\|$.
%$\textcolor{red}{M \geq -\frac{U_a^2\ln{\frac{\beta_2}{m+1}}}{2\tau^2}}$ holds. 
%We have the following conclusion, the proof of which is available in Appendix B.

%This holds provided that the sample sizes satisfy:
%\begin{equation*}
%\begin{split}
%&N \geq \frac{2}{\alpha_1}\left(\ln\frac{1}{\delta_1} + m+1\right), \\
%&M \geq \frac{1}{2\tau^2}\ln\frac{4}{\delta_2}.
%\end{split}
%\end{equation*}
%\end{theorem}

\section{Experiments}
\label{sec:ex}
This section provides several examples to evaluate our three PAC safety certification methods. To demonstrate practical applicability, a case study using the CARLA autonomous driving simulator is provided in Appendix~\ref{sec:carla}. 
% The implementation is available at \href{https://github.com/TaoranWu/PSC-BDSS}{https://github.com/TaoranWu/PSC-BDSS}.
%To comprehensively evaluate our approach, we also consider white-box baselines that assume both the system dynamics and disturbance distributions are known. A representative model-based method for computing RBCs and SBCs is the sum-of-squares (SOS) programming. Specifically, when $\bm{f}(\cdot,\cdot)$ is a known polynomial and the disturbance set $\mathbb{D}$ is known and semi-algebraic, constraint \eqref{eq:rbf} can be reformulated as an SOS programming, yielding an RBC baseline (\textit{RBC-SOS}). We compare it with our RBC-based single-layer PAC certification (Theorem \ref{theo:robust}, denoted \textit{RBC-I}) and nested PAC certification (Theorem \ref{thm:main-vc-final-corrected}, denoted \textit{RBC-II}). Likewise, when $\bm{f}(\cdot,\cdot)$ is a known polynomial function and the disturbance distribution is known, constraint \eqref{eq:sbf} can be recast as an SOS programming to obtain an SBC baseline (\textit{SBC-SOS}), which we compare with our SBC-based method in Theorem \ref{thm:main-vc-final-sbc1} (\textit{SBC-III}).
For evaluation, we include white-box baselines that assume full knowledge of the system and disturbances. When $\bm{f}(\cdot,\cdot)$ is a known polynomial and $\mathbb{D}$, $\mathbb{X}$ are semi-algebraic, constraint \eqref{eq:rbf} can be formulated as an SOS program to produce \textbf{RBC-SOS}, which we compare with our PAC methods \textbf{RBC-I} (Theorem~\ref{theo:robust}) and \textbf{RBC-II} (Theorem~\ref{thm:main-vc-final-corrected}). Similarly, if both $\bm{f}(\cdot,\cdot)$ and the disturbance distribution are known, constraint \eqref{eq:sbf} can be cast as an SOS program to yield \textbf{SBC-SOS}, which we compare with \textbf{SBC-III} (Theorem~\ref{thm:main-vc-final-sbc1}).
%To provide a comprehensive evaluation, we also include white-box baselines that assume full knowledge of the system dynamics and disturbance distributions. A standard model-based approach for computing RBCs and SBCs is sum-of-squares (SOS) programming, applicable when the dynamics $\bm{f}(\cdot,\cdot)$ are known polynomials and the disturbance set $\mathbb{D}$ and safe set $\mathbb{X}$ are semi-algebraic. In this setting, constraint \eqref{eq:rbf} can be formulated as an SOS program, yielding an RBC baseline (denoted RBC-SOS). We compare this with the PAC method (Theorem \ref{theo:robust}, \textnormal{RBC-I}) and the PAC method (Theorem \ref{thm:main-vc-final-corrected}, \textnormal{RBC-II}). Similarly, when $\bm{f}(\cdot,\cdot)$ and the disturbance distribution are both known, constraint \eqref{eq:sbf} can be reformulated as an SOS program to produce an SBC baseline (\textnormal{SBC-SOS}), which we compare against our PAC approach from Theorem \ref{thm:main-vc-final-sbc1} (\textnormal{SBC-III}).
 \vspace{-0pt}
\begin{table*}[!h]
\centering
\caption{Summary of performance and computation times ({seconds}).}
\footnotesize% \scriptsize
\setlength{\tabcolsep}{0.7mm}{
    \begin{tabular}{*{15}{c}}
    \toprule
    \multirow{2}{*}{EX} & \multirow{2}{*}{n} & \multicolumn{2}{c}{RBC-SOS} & \multicolumn{3}{c}{RBC-I} & \multicolumn{3}{c}{RBC-II} & \multicolumn{1}{c}{SBC-SOS} & \multicolumn{2}{c}{SBC-III}\\\cmidrule(lr){3-4} \cmidrule(lr){5-7} \cmidrule(lr){8-10} \cmidrule(lr){11-11} \cmidrule(lr){12-13}
    & & $T$ & $\xi^*\!\!\!=\!0$ & $N$& $T$ & $\xi^*\!\!\!=\!0$ &$(N,M)$  & $T$ & $\xi^*\!\!\!=\!0$ & $T$ & $T$ &$(N,M)$  \\ \midrule
    \ref{ex:vinc} & 2 & 1.0 & \ding{55} &15053& 0.2 & \ding{51} &(3764,45) & 0.9 & \ding{51}& 2.8 & 1.3 &(3764,1357)\\
    \ref{ex:arch} & 2 & 1.3 & \ding{51} &15053& 0.2 & \ding{51} &(3764,45)& 0.9 & \ding{51} & 2.5 & 0.8 &(3764,1357)\\
    % \ref{ex:bc4} & 2 & 1.2 & \ding{55} & 0.2 & \ding{51} & 11.7 & \ding{55} & 2.1 & 0.2201 & 16.8 & 0.0826 \\
    \ref{ex:stable} & 3 & 8.7 & \ding{51} &18253& 0.3 & \ding{51} &(4564,45)& 1.0 & \ding{51} & 2.4 & 1.6 &(4564,1357)\\
    \ref{ex:4d} & 4 & 1.4 & \ding{55} &24653& 0.3 & \ding{51}&(6164,45) & 1.8 & \ding{51} & 2.6 & 2.8 &(6164,1357)\\
    \ref{ex:6d} & 6 & 17.7 & \ding{51} &63053& 0.8 & \ding{51} &(15764,45)& 8.2 & \ding{51} & 21.4 & 16.1 &(15764,1357)\\
    \ref{ex:lotka} & 2 & 1.0 & \ding{55} &15053& 0.2 & \ding{55} &(3764,45)& 0.9 & \ding{55} & 10.3 & 20.4 &(27164,631)\\
    \ref{ex:pendulum} & 2 & - & - &15053& 0.2 & \ding{55} &(3764,45)& 0.9 & \ding{55} & -  & 19.2 &(27164,631)\\
    \ref{ex:sank} & 4 & 3.9 & \ding{55} &24653& 0.3 & \ding{55} &(6164,45)& 1.8 & \ding{55} & 153.4 & 9.2 &(19164,340)\\
    % \ref{ex:jet} & 2 &  & \ding{55} & 0.2 & \ding{55} & 11.9 & \ding{55} &  &  &  &  \\
    % \ref{ex:temp}   & 3 & 1.3 & \ding{55} & 0.3 & \ding{55} & 19.3 & \ding{55} & - & - & 182.6 & 0.5492 \\
    % \ref{ex:lorenz} & 6 & 2.4 & \ding{55} & 0.7 & \ding{55} & 8.5 & \ding{55} & 113.6 & 0.8004 & 2.0 & 0.4799 \\
    \ref{ex:lorenz} & 7 & 2.4 & \ding{55}&114253 & 2.5 & \ding{55} &(28564,45)& 38.3 & \ding{55} & 84.4 & 17.6 &(28564,340) \\
    % \ref{ex:lorenz} & 8 & 2.4 & \ding{55} & 7.6 & \ding{55} & 117.4 & \ding{55} & 77.4 & 0.9352 &  & 0.6193 \\
\bottomrule
\end{tabular}}
\label{tab:results}
\end{table*}
\vspace{-0pt}

%In all examples, we use the following settings.
%For RBC-I, we set $\delta = 10^{-6}$ and $\alpha_1 = \alpha_2 = 0.05$. For RBC-II, we use $\delta_1 = 10^{-6}$, $\alpha_1 = 0.01$,  $\delta_2=0.999$, $\alpha_2 = 0.05$, and $l = 0.2$, which gives $1 - \frac{\alpha_1}{l\delta_2} \approx 0.95$. For SBC-III, we use $\delta_1 = 10^{-6}$, $\alpha_1 = 0.01$, $\delta_2 = 0.999$, and $l = 0.2$. Table~\ref{tab:results} summarizes the results, where $T$ is the total computation time (in seconds) and $J^*$ is the optimum of \eqref{eq:sbc_linear3}. %\dw{I think it would also be interesting to see a table with the required number of samples, since this motivates some of the refinements of the method.}
%In the table, $\xi^* = 0$ means the method yields a valid safety guarantee. The experimental settings are given in Appendix~\ref{sec:app_setting}, and additional results, including those under different hyperparameters, are presented in Appendix~\ref{sec:app_para}. We then discuss Examples~\ref{ex:vinc} and~\ref{ex:lotka} as representative cases.
%\dw{I don't see data on how often it occurs that your method says a state is safe but it is in fact not safe.} \wtr{For SBC-III, I provide Monte Carlo estimates of one-step safety in Table \ref{tab:lotka} of Example \ref{ex:lotka}. For the RBC-based methods, the systems in Examples 2, 3, and 6 never leave the safe set, and in the remaining examples the probability of leaving the safe set is also extremely small; Example 1 includes a Monte Carlo illustration of this as well.}

For all examples, the settings are as follows. \textbf{RBC-I:} $\delta = 10^{-6}$, $\alpha_1 = \alpha_2 = 0.05$. \textbf{RBC-II:} $\delta_1 = 10^{-6}$, $\alpha_1 = 0.01$, $\delta_2 = 0.999$, $\alpha_2 = 0.05$, $l = 0.2$ ($1 - \frac{\alpha_1}{l\delta_2} \approx 0.95$). \textbf{SBC-III:} $\delta_1 = 10^{-6}$, $\alpha_1 = 0.01$, $\delta_2 = 0.999$, $l = 0.2$. Table~\ref{tab:results} reports the results, where $T$ denotes the total computation time in seconds, %and $J^*$ is the optimum of \eqref{eq:sbc_linear3};
$\xi^* = 0$ indicates a valid safety guarantee, and $N$ and $M$ denote the numbers of state and disturbance samples, respectively. Additional experimental details are in Appendix~\ref{sec:app_setting}, and results for other hyperparameters are in Appendix~\ref{sec:app_para}. Examples~\ref{ex:vinc} and~\ref{ex:lotka} are discussed as representative cases.

\begin{example}
    \label{ex:vinc}
    Consider the following system adapted from \cite{vincent1997nonlinear}:
\begin{equation*}
    \begin{cases}
x(t+1)=x(t)+0.01\big(y(t) - x(t)  (d(t)+0.5)\big),\\
y(t+1)=y(t)+0.01\big(-(1-x^2(t))x(t) - y(t)\big),
\end{cases}
\end{equation*}
where $d(t)$ follows a truncated normal distribution with mean 0 and standard deviation 
$0.1$, limited to $[-0.7, 0.7]$. The safe set is $\mathbb{X}=\{\,(x,y)\mid x^2+y^2 \leq 0.64\,\}$.

%In this example, RBC-SOS method fails to produce a valid RBC, indicating that some disturbance realizations in $\mathbb{D}$ may violate \eqref{eq:rbf}, even if their sampling probability is extremely low. In contrast, the proposed RBC-I and RBC-II methods provide high-confidence probabilistic guarantees. As a reference, we conduct a Monte Carlo study in which $10^6$ states are uniformly sampled from the safe set $\mathbb{X}$, and for each state $10^3$ successor states are simulated; all sampled successors remain within $\mathbb{X}$. This empirical result aligns with the guarantees delivered by RBC-I and RBC-II. Experimental results for different parameter settings are provided in Appendix~\ref{sec:app_para}. 
In this example, the RBC-SOS method fails to provide a valid RBC and some disturbances lead to the safety violation, as shown in Fig. \ref{fig:vinc_traj}. In contrast, RBC-I and RBC-II provide high-confidence probabilistic guarantees. %\dw{how many seeds were used?}\wtr{All experiments reported in Tab. \ref{tab:results} are conducted with the random number seed set to 0. This is explained in the Appendix. I think that averaging results over multiple random seeds seems difficult to explain theoretically. Therefore, in the appendix I will provide two examples that report the results obtained under five different random seeds.} 
For reference, we performed a Monte Carlo study by uniformly sampling $10^6$ states from $\mathbb{X}$ and simulating $10^3$ successors for each; all sampled successors stayed within $\mathbb{X}$. This result agrees with the guarantees from RBC-I and RBC-II. 
% Additional results under different parameters are given in Appendix~\ref{sec:app_para}.

Besides, although RBC-II is more time-consuming, it would provide higher one-step safety probability by using multiple disturbance samples per state (see Remark~\ref{remark_com_single_and_nested}). This becomes evident when we fix the sample size rather than parameters such as $\alpha_1$ and $\alpha_2$. For RBC-I, using $10^3$ samples of $(\bm{x},\bm{d})$ gives the guarantee $\textnormal{P}_{\mathbb{S}}\Big[ \textnormal{P}_{\bm{x}}\Big[\textnormal{P}_{\bm{d}}\!\big[ \bm{f}(\bm{x},\bm{d}) \in \mathbb{X} \big] \ge 0.80 \Big] \ge 0.80 \Big] \ge 1-10^{-6}$. For RBC-II, using $10^3$ samples of $\bm{x}$ and $10$ samples of $\bm{d}$ per $\bm{x}$ (a one-to-many strategy not possible in RBC-I) yields $\textnormal{P}_{\mathbb{S}}\Big[ \textnormal{P}_{\bm{x}}\Big[\textnormal{P}_{\bm{d}}\!\big[ \bm{f}(\bm{x},\bm{d}) \in \mathbb{X} \big] \ge 0.96  \Big] \ge 0.80 \Big] \ge 1-10^{-6}$. The resulting safety probability $\textnormal{P}_{\bm{d}}[ \bm{f}(\bm{x},\bm{d}) \in \mathbb{X}]$ from RBC-II is higher.%, but guaranteed with confidence of at least $0.80$. %reflecting its greater disturbance-level expressiveness.

% In addition, SBC-III achieves \textcolor{red}{performance} \textcolor{blue}{please explain what the performance is. }that closely matches SBC-SOS,
%\begin{figure}
%    \centering
%    \includegraphics[width=0.3\linewidth]{Fig/n2_sos.pdf}
%    \caption{Contour of SBC-SOS in Example \ref{ex:vinc}.}
%    \label{fig:contour_sos_1}
%\end{figure}
%\textcolor{blue}{(This part largely overlaps with the discussion in Example 2, and the trajectory plots and contour visualizations for SBC-III in this example are not particularly informative. Due to page limitations, I think this paragraph can be removed.) 
The contour plot of $1-\lambda^*-h(\bm{a}^*,\bm{x})$ from SBC-SOS is shown in Fig. \ref{fig:vinc_sos}. In contrast, SBC-III yields $1-\lambda^*(\mathbb{S})-h(\bm{a}^*(\mathbb{S}),\bm{x}) \approx 0.99$ for every $\bm{x} \in \mathbb{X}$. Although SBC-III provides only a one-step guarantee, the lower probability bound it delivers is significantly tighter than that of SBC-SOS.
\vspace{-0pt}
\begin{figure}[!h]
%\label{ex1:fig}
    \centering
    \subfigure[\scriptsize Trajectories in Ex. \ref{ex:vinc}]{\includegraphics[width=0.3\linewidth]{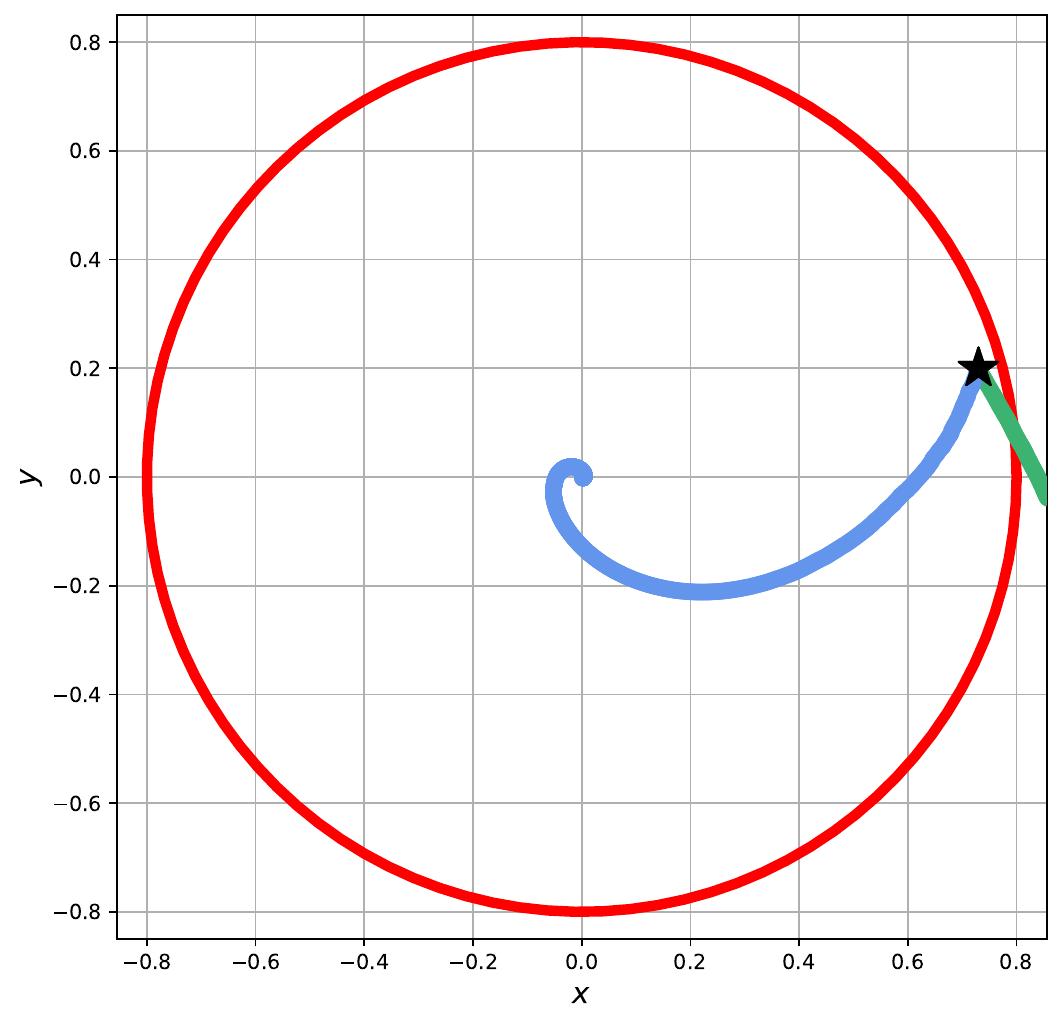}
    \label{fig:vinc_traj}}
    \subfigure[\scriptsize SBC-SOS in Ex. \ref{ex:vinc}]{\includegraphics[width=0.3\linewidth]{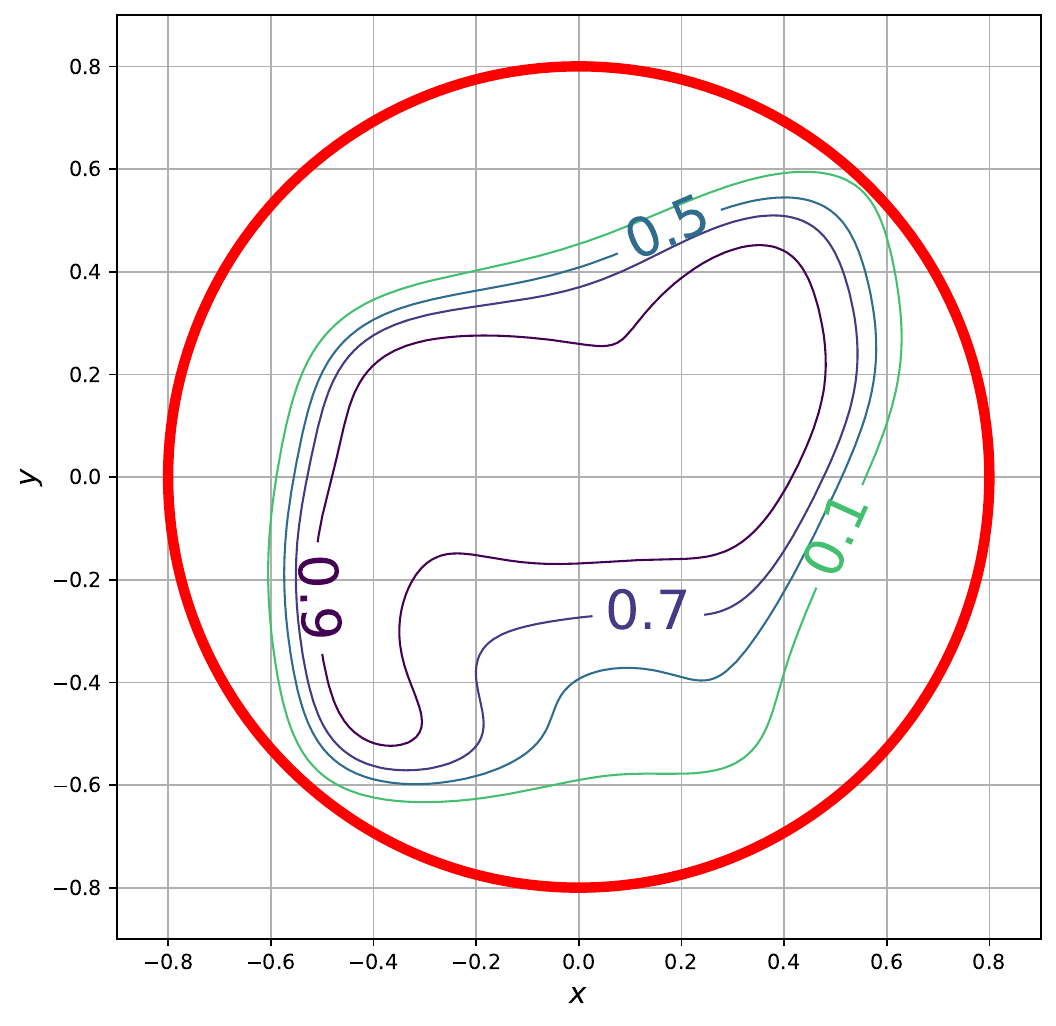}
    \label{fig:vinc_sos}}
    \vspace{-0pt}
    \caption{Results for Ex. \ref{ex:vinc}}
\end{figure}
   % \caption{Contour of SBC-SOS in Example \ref{ex:vinc}.}
\vspace{-0pt}
%In addition, SBC-III achieves \textcolor{blue}{objective value $J^*$} \textcolor{red}{Actually, this value can not mean too much. It would be better to explain what it means briefly here. Otherwise, remove this discussions. If we can show that their contour lines for computed $1-\lambda^*(\mathbb{S})-h(\bm{a}^*(\mathbb{S}),\bm{x})$ collide, we can say something like "although only conservative one-step guarantee is provided in our method, its practical performance may match the one from SBC-SOS." } that closely matches SBC-SOS, which relies on full model information, demonstrating the effectiveness of our approach. However, SBC does not require the safety condition to hold for every disturbance sample, but only in expectation, so its guarantees are less stringent than those offered by the RBC-based methods in this example.
\end{example}

\setcounter{example}{5}
\begin{example}
\label{ex:lotka}
    Consider the discrete-time Lotka–Volterra model, widely used to describe the dynamics of interacting biological species \cite{lotka2002contribution}:
\begin{equation*}
    \begin{cases}
x(t+1)=rx(t)-ay(t)x(t),\\
y(t+1)=sy(t)+acy(t)x(t),
\end{cases}
\end{equation*}
where $r=0.5$, $a=c=1$, and $s=-0.5+d(t)$ with $d(\cdot)\colon\mathbb{N}\rightarrow \mathbb{D}=[-1,1]$ drawn from the uniform distribution on $\mathbb{D}$. The safe set is $\mathbb{X}=\{\,(x,y)\mid x^2+y^2 \leq 1\,\}$.

In this example, five trajectories initialized at $(-0.8, -0.5)$ and simulated for five time steps are shown in Fig. \ref{fig:lotka_traj}. It is evident that the system leaves the safe set under certain disturbances, indicating that no RBC exists in this example. Therefore, RBC-SOS, RBC-I(for $\alpha_1 \leq 0.9$ and $\alpha_2\leq0.9$), and RBC-II(for $\alpha_1\leq0.15$, $\alpha_2\leq0.9$) cannot provide valid guarantees, and we focus instead on the SBC-based methods.

The contour plots of $1-\lambda^*(\mathbb{S})-h(\bm{a}^*(\mathbb{S}),\bm{x})$ from SBC-III and $1-\lambda^*-h(\bm{a}^*,\bm{x})$ from SBC-SOS, are shown in Fig. \ref{fig:lotka_sbc_nested} and Fig. \ref{fig:lotka_sbc_sos}, respectively. %As reported in Tab. \ref{tab:results}, SBC-III yields a better objective value $J^*$. Moreover,
Fig.~\ref{fig:lotka_sbc_sos} shows that SBC-SOS provides probability lower bounds that quickly drop to zero as the system state approaches the boundary of the safe set $\mathbb{X}$, which will be trivial and uninformative. In contrast, SBC-III avoids this issue, though it provides only a PAC one-step safety guarantee. This performance advantage likely stems from the piecewise template in \eqref{eq:h_p}, which enforces constraints \eqref{eq:sbf1} and \eqref{eq:sbf2} by construction. By comparison, SBC-SOS seeks a single polynomial satisfying \eqref{eq:sbf1}, \eqref{eq:sbf2}, and \eqref{eq:sbf3} simultaneously, imposing significantly stronger requirements on expressiveness and global approximation accuracy. We plan to investigate this issue in detail in future work. %In addition, SBC-III relies on linear programming, avoiding the relaxations inherent in SOS programming.
Tab.~\ref{tab:lotka} (in Appendix~\ref{sec:app_results}) compares SBC-III and SBC-SOS along the purple trajectory shown in Fig.~\ref{fig:lotka_traj}. For reference, we also report Monte Carlo estimates. As observed in Tab.~\ref{tab:lotka}, SBC-III gives a less conservative lower bound than SBC-SOS. 
% Additional results for different parameter settings are in Appendix~\ref{sec:app_para}.
% Tab.~\ref{tab:lotka} (In Appendix~\ref{sec:app_para}) compares SBC-III and SBC-SOS along the purple trajectory in Fig.~\ref{fig:lotka_traj}, where $\textnormal{P}_{\text{SBC-III}}(\bm{x}) = 1-\lambda^*(\mathbb{S})-h(\bm{a}^*(\mathbb{S}),\bm{x})$ and $\textnormal{P}_{\text{SBC-SOS}}(\bm{x})=1-\lambda^*-h(\bm{a}^*,\bm{x})$. For reference, we also report the Monte Carlo estimate $\textnormal{P}_{\textit{MC}}(\bm{x})$, computed by generating $10^6$ successor samples for each state $\bm{x}$ and calculating the fraction that remains in the safe set $\mathbb{X}$. As shown in Tab.~\ref{tab:lotka}, SBC-III gives a less conservative lower bound than SBC-SOS. Additional results for different parameter settings are in Appendix~\ref{sec:app_para}.% denote the corresponding lower bounds. 
%For reference, we also report the empirical estimate $\textnormal{P}_{\textit{MC}}(\bm{x})$ via Monte Carlo sampling. For each state $\bm{x}$ along the trajectory, we generate $10^6$ successor samples and compute the fraction that remains within the safe set $\mathbb{X}$. As shown in Tab. \ref{tab:lotka}, compared to the SBC-SOS, SBC-III yields less conservative lower bound. Additional results under different parameter settings are provided in Appendix~\ref{sec:app_para}.
\vspace{-0pt}
\begin{figure}[!h]
    \centering
    \subfigure[\scriptsize Trajectories in Ex. \ref{ex:lotka}]{\includegraphics[width=0.3\linewidth]{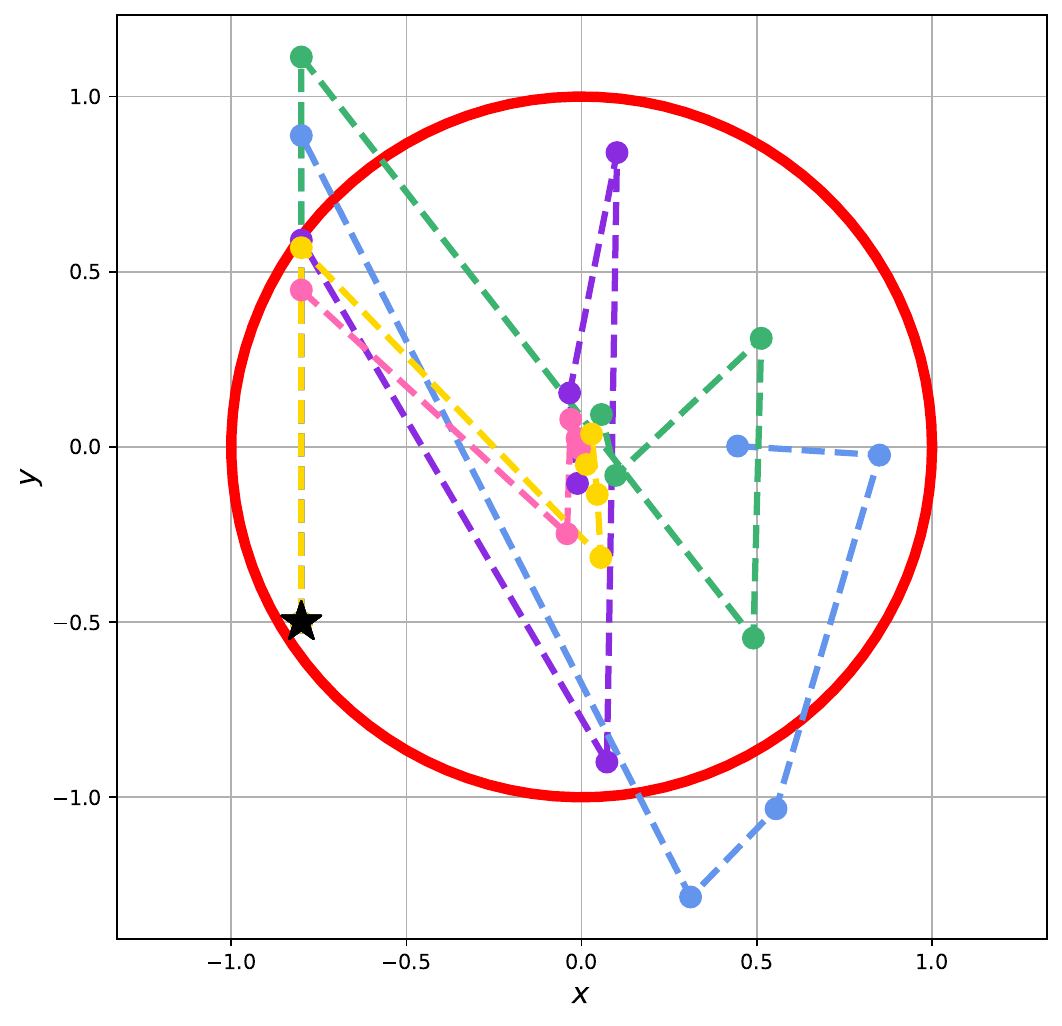}
    \label{fig:lotka_traj}}
    \subfigure[ \scriptsize SBC-III in Ex. \ref{ex:lotka}]{\includegraphics[width=0.3\linewidth]{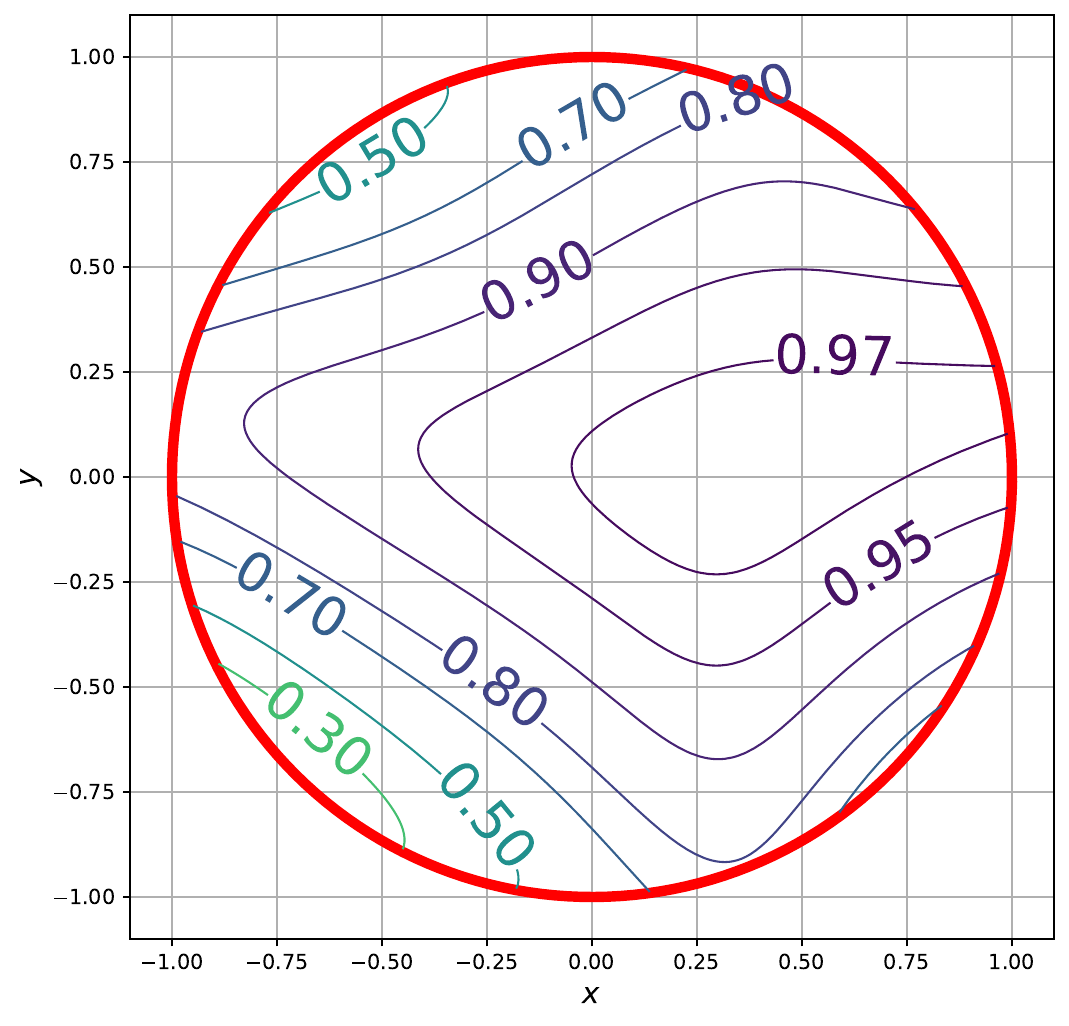}
    \label{fig:lotka_sbc_nested}}
    \subfigure[\scriptsize SBC-SOS in Ex. \ref{ex:lotka}]{\includegraphics[width=0.3\linewidth]{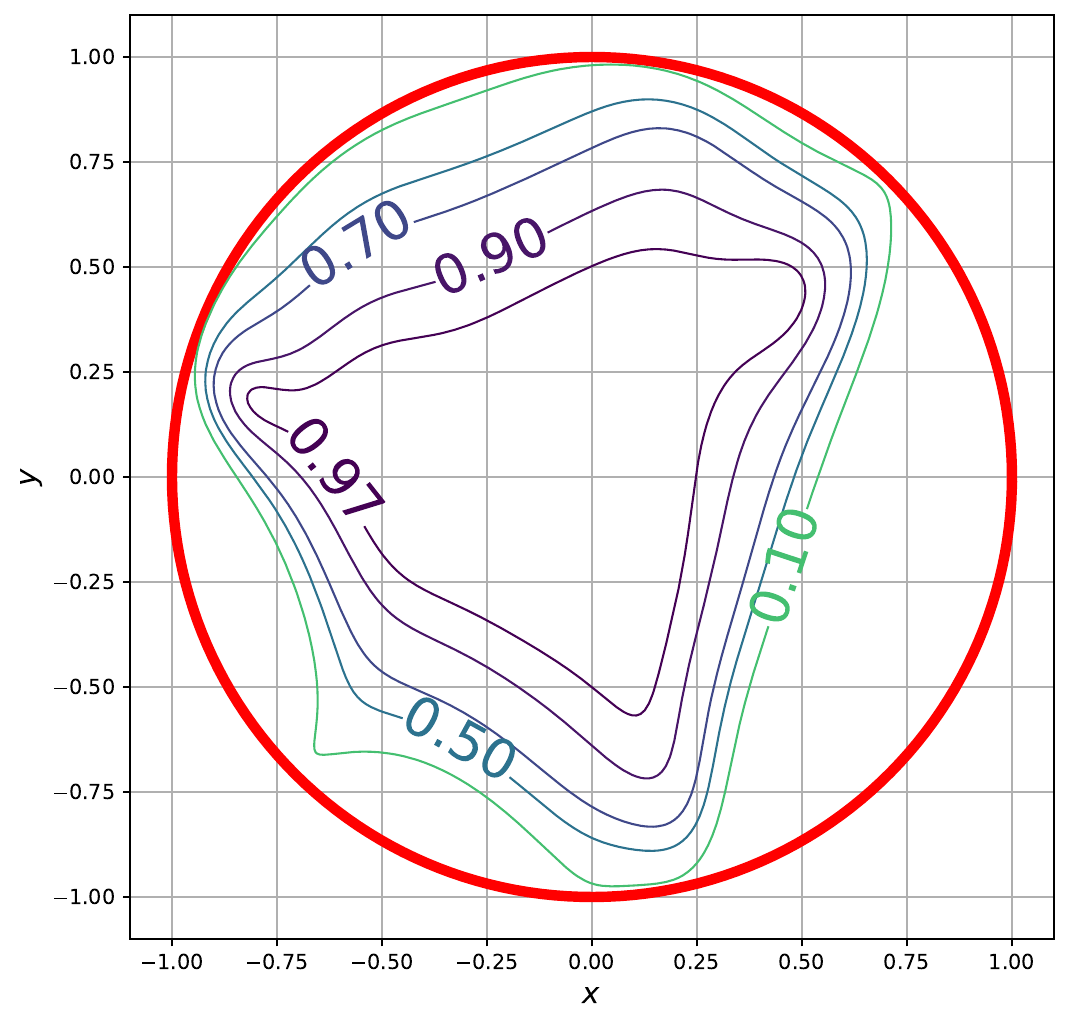}
    \label{fig:lotka_sbc_sos}}
    \vspace{-0pt}
    \caption{Results for Example \ref{ex:lotka}}
    \vspace{0pt}
    \label{fig:placeholder}
\end{figure}
\end{example}

\textbf{Summary}: Overall, the results demonstrate that the proposed methods provide formal PAC one-step safety guarantees without any model knowledge. In Examples~\ref{ex:vinc} and~\ref{ex:4d}, RBC-SOS fails to produce a valid certificate, while both RBC-I and RBC-II still succeed. Similarly, SBC-III consistently achieves probability lower bounds that match or exceed those of SBC-SOS. In addition, SOS-based methods require the system dynamics to be a known polynomial and the safe set $\mathbb{X}$ to be semi-algebraic, whereas our methods have no such restrictions.

% Across Examples \ref{ex:vinc}-\ref{ex:4d}, the RBC-based approaches generally yield stronger guarantees than the SBC-III, whereas in Examples \ref{ex:lotka}-\ref{ex:lorenz} the RBC-based methods fail to provide valid guarantee. Thus, the RBC and SBC paradigms are complementary, and practitioners should select between them based on system characteristics.

% Regarding scalability, RBC-I has low sample complexity, making it well suited for high-dimensional systems.
%\textcolor{red}{Regarding scalability, the computation time of RBC-I increases only moderately with the system dimension, demonstrating strong scalability to high-dimensional systems.} 
Our computational methods primarily involve solving linear programs, whose size mainly depends on the sample complexity. The sample complexity, in turn, depends only on parameters such as $m$, $\alpha_1$, $\alpha_2$, $\delta_1$, and/or $(\delta_2,l)$, and is independent of the system dimension $n$. Hence, all three methods theoretically scale well to high-dimensional systems. In practice, RBC-I performs particularly efficiently, with minimal computation time across all examples. For SBC-III, the number of constraints in the optimization problem \eqref{eq:sbc_linear3} depends only on the state sample size $N$ and not on $M$, keeping the problem tractable. For example, in Example~\ref{ex:lorenz}, solving SBC-III requires only 0.8 seconds, with most of the total time spent on data collection and preprocessing.

\section{Conclusion}
\label{sec:con}
In this paper, we presented a data-driven framework for one-step safety certification of black-box discrete-time stochastic systems, where both the system dynamics and disturbance distributions are unknown and only sampled data are available. Based on robust and stochastic barrier certificates, we proposed three methods for establishing formal one-step safety guarantees based on scenario approaches, VC dimension, and Markov's and Hoeffding's inequalities. The effectiveness of our approaches was demonstrated through several examples. 

Besides the future work discussed in Remarks \ref{remark:barrier} and \ref{remark_com_single_and_nested}, we plan to incorporate recent advances in scenario optimization \cite{campi2018wait} and other methods to further reduce sample complexity and improve computational efficiency, as well as to extend the proposed approaches to PAC verification for probabilistic programs \cite{zaiser2025guaranteed} and quantum computations \cite{ying2010quantum}. {In this context, the program or algorithm's input-output execution is modeled as a stochastic transition, where input states are sampled from their intrinsic (often non-uniform) distributions and the program or algorithm executes according to its probabilistic semantics.} We also plan to extend our approach to (controlled) continuous-time systems, which are ubiquitous in the physical world.

\bibliographystyle{elsarticle-num}
\bibliography{reference}

\newpage
\section*{Appendix}
\setcounter{secnumdepth}{2}
\setcounter{section}{0}
\renewcommand\thesubsection{\thesection.\arabic{subsection}}
\renewcommand\thesection{\Alph{section}}
%\input{appendix1}
%\newpage
%\input{appendix2}
%\newpage
%\input{appendix3}
%\newpage
%\input{appendix4}
\section{Proofs of the Theorems}
 \label{appendix:pot}
\subsection{The Proof of Theorem \ref{thm:main-fixedpoint}}
\label{appendix:a1}

\setcounter{theorem}{0}
\begin{theorem}[PAC Safety Certification I Based on VC Dimension]
%\label{thm:main-fixedpoint}
Fix probability threshold $\alpha_1, \alpha_2 \in (0,1)$  and cofidence level $\delta \in (0,1)$.  Let $(\bm{a}^*(\mathbb{S}), \xi^*(\mathbb{S}))$ denote the solution obtained by solving~\eqref{eq:rbc_rc}. %where $\mathbb{S} = \{(\bm{x}^{(i)}, \bm{d}^{(i)})\}_{i=1}^N$ is an i.i.d. sample drawn from $\mathbb{P}_{\bm{x}} \times \mathbb{P}_{\bm{d}}$. 
Define the computational procedure:
\begin{equation*}
\mathcal{A}(\mathbb{S}) = 1_{\{\xi^*(\mathbb{S}) = 0\}}(\mathbb{S}).
\end{equation*}
If $N$ satisfies \eqref{NN1} with $\alpha=\alpha_1 \alpha_2$, then, with confidence at least $1 - \delta$ over the sample set $\mathbb{S}$, if $\xi^*(\mathbb{S})=0$, $
\textnormal{P}_{\bm{x}}\!\Big[
\textnormal{P}_{\bm{d}}\big[h(\bm{a}^*(\mathbb{S}),\bm{f}(\bm{x},\bm{d}))\geq \gamma\,h(\bm{a}^*(\mathbb{S}),\bm{x})] \geq 1-\alpha_2
\Big]
\ge 1-\alpha_1$.
This implies $\textnormal{P}_{\mathbb{S}}\Big[\mathcal{A}(\mathbb{S})\Rightarrow \textnormal{P}_{\bm{x}}\!\Big[
\textnormal{P}_{\bm{d}}\big[
\bm{f}(\bm{x},\bm{d})\in \mathbb{X}
\big] \geq 1-\alpha_2
\Big]
\ge 1-\alpha_1\Big]\geq 1-\delta$.
%where $\mathbb{X}_s=\{\bm{x}\in \mathbb{X}\mid h(\bm{a}^*(\mathbb{S}),\bm{x})\geq 0\}$.
%and 
%$\epsilon(\bm{x}):=\begin{cases}
%    \alpha_2, &\text{if~}\bm{x}\in \mathbb{X}_s,\\
%    1, &\text{otherwise}.
%\end{cases}$
\end{theorem}
\begin{proof}
The proof follows three steps.
%\dw{I tried to streamline the argument.}

\textbf{1. PAC Guarantees Based on Lemma \ref{coro4}.} According to Lemma \ref{coro4}, we have 
\[\textnormal{P}_{\mathbb{S}}[\textnormal{P}_{(\bm{x},\bm{d})}[h(\bm{a}^*(\mathbb{S}),\bm{f}(\bm{x},\bm{d}))\ge\gamma\,h(\bm{a}^*(\mathbb{S}),\bm{x})-\xi^*(\mathbb{S})]\geq 1-\alpha]\geq 1-\delta,\]
if 
$N\geq \frac{5}{\alpha}\big( \ln{\frac{4}{\delta}}+\operatorname{vc}(\mathbb{H}) \ln{\frac{40}{\alpha}}\big)$.
Thus, 
\[\textnormal{P}_{\mathbb{S}}[\textnormal{P}_{(\bm{x},\bm{d})}[h(\bm{a}^*(\mathbb{S}),\bm{f}(\bm{x},\bm{d}))<\gamma\,h(\bm{a}^*(\mathbb{S}),\bm{x})-\xi^*(\mathbb{S})]\leq \alpha]\geq 1-\delta.\]

\textbf{2. Guarantee Decomposition over $\textnormal{P}_{(\bm{x},\bm{d})}$ Based on Markov's Inequality.} 
Due to Markov's, we have for all $\mathbb S$,
\[
\begin{split}
&\textnormal{P}_{\bm{x}}[\textnormal{P}_{\bm{d}}[h(\bm{a}^*(\mathbb{S}),\bm{f}(\bm{x},\bm{d}))<\gamma\,h(\bm{a}^*(\mathbb{S}),\bm{x})-\xi^*(\mathbb{S})\mid \bm{x}]
\geq \alpha_2]\\
&\leq  \frac{\textnormal{E}_{\bm{x}}[\textnormal{P}_{\bm{d}}[h(\bm{a}^*(\mathbb{S}),\bm{f}(\bm{x},\bm{d}))<\gamma\,h(\bm{a}^*(\mathbb{S}),\bm{x})-\xi^*(\mathbb{S})\mid \bm{x}]]}{\alpha_2}\\
&=\frac{\textnormal{P}_{(\bm{x},\bm{d})}[h(\bm{a}^*(\mathbb{S}),\bm{f}(\bm{x},\bm{d}))<\gamma\,h(\bm{a}^*(\mathbb{S}),\bm{x})-\xi^*(\mathbb{S})]}{\alpha_2} \leq \alpha_1,
\end{split}
\]

 \textbf{3. Conclusion.}
Thus, we can obtain 
\begin{align*}
1-\delta&\leq \textnormal{P}_{\mathbb{S}}[\textnormal{P}_{(\bm{x},\bm{d})}[h(\bm{a}^*(\mathbb{S}),\bm{f}(\bm{x},\bm{d}))<\gamma\,h(\bm{a}^*(\mathbb{S}),\bm{x})-\xi^*(\mathbb{S})]\leq \alpha]\\
&\leq \textnormal{P}_{\mathbb{S}}[\textnormal{P}_{\bm{x}}[\textnormal{P}_{\bm{d}}[h(\bm{a}^*(\mathbb{S}),\bm{f}(\bm{x},\bm{d}))<\gamma\,h(\bm{a}^*(\mathbb{S}),\bm{x})-\xi^*(\mathbb{S})\mid \bm{x}]
\geq \alpha_2]\leq \frac\alpha{\alpha_2}]\\
&= \textnormal{P}_{\mathbb{S}}[\textnormal{P}_{\bm{x}}[\textnormal{P}_{\bm{d}}[h(\bm{a}^*(\mathbb{S}),\bm{f}(\bm{x},\bm{d}))<\gamma\,h(\bm{a}^*(\mathbb{S}),\bm{x})-\xi^*(\mathbb{S})\mid \bm{x}]
< \alpha_2]\geq 1-\alpha_1]\\
&= \textnormal{P}_{\mathbb{S}}[\textnormal{P}_{\bm{x}}[\textnormal{P}_{\bm{d}}[h(\bm{a}^*(\mathbb{S}),\bm{f}(\bm{x},\bm{d}))\geq\gamma\,h(\bm{a}^*(\mathbb{S}),\bm{x})-\xi^*(\mathbb{S})\mid \bm{x}]
\geq 1-\alpha_2]\geq 1-\alpha_1],
\end{align*}
%\dw{not sure whether there's a small mistake because I get a strict inequality}
%\[\textnormal{P}_{\mathbb{S}}[\textnormal{P}_{\bm{x}}[\textnormal{P}_{\bm{d}}[h(\bm{a}^*(\mathbb{S}),\bm{f}(\bm{x},\bm{d}))\ge\gamma\,h(\bm{a}^*(\mathbb{S}),\bm{x})-\xi^*(\mathbb{S})\mid \bm{x}]
%\geq 1-\alpha_2]\geq 1-\alpha_1]\geq 1-\delta,\]
implying
\[
\begin{split}
&\textnormal{P}_{\mathbb{S}}\Bigg[\xi^*(\mathbb{S})=0 \Rightarrow \textnormal{P}_{\bm{x}}\left[
\begin{split}&\textnormal{P}_{\bm{d}}\left[
\begin{split}
   & h(\bm{a}^*(\mathbb{S}),\bm{f}(\bm{x},\bm{d}))\\
   &\ge\gamma\,h(\bm{a}^*(\mathbb{S}),\bm{x})
\end{split}
\right]\geq 1-\alpha_2
\end{split}
\right]\geq 1-\alpha_1 \Bigg]\geq 1-\delta,
\end{split}
\]
and 
$\textnormal{P}_{\mathbb{S}}\Bigg[\xi^*(\mathbb{S})=0 \Rightarrow \textnormal{P}_{\bm{x}} \left[
\textnormal{P}_{\bm{d}}[\bm{f}(\bm{x},\bm{d}) \in \mathbb{X}]\geq 1-\alpha_2
\right] \ge 1-\alpha_1\Bigg]\geq 1-\delta$.
This proves the claim. \qed
\end{proof}

\subsection{The Proof of Theorem \ref{theo:robust}}
\label{appendix:proof_theo2}
\begin{theorem}[PAC Safety Certification I Based on Scenario Approaches]
%\label{theo:robust}
Fix probability thresholds $\alpha_1, \alpha_2 \in (0,1)$ and confidence level $\delta \in (0,1)$, and let the sample size satisfy
$N \;\geq\; \frac{2}{\alpha_1 \alpha_2}\!\left(\ln\frac{1}{\delta} + m + 1\right)$.

Let $(\bm{a}^*(\mathbb{S}), \xi^*(\mathbb{S}))$ be the solution to~\eqref{eq:rbc_rc}, with the computational procedure in~\eqref{cp}. Then,  with confidence at least $1 - \delta$ over the random draw of $\mathbb{S}$, if $\xi^*(\mathbb{S}) = 0$, 
$\textnormal{P}_{\bm{x}}\!\Big[
\textnormal{P}_{\bm{d}}\!\big[
h(\bm{a}^*(\mathbb{S}), \bm{f}(\bm{x},\bm{d})) \geq \gamma h(\bm{a}^*(\mathbb{S}), \bm{x})
\big] \geq 1- \alpha_2
\Big]\ge 1-\alpha_1$ holds, i.e., $\textnormal{P}_{\mathbb{S}}\Big[
\mathcal{A}(\mathbb{S}) \Rightarrow 
\textnormal{P}_{\bm{x}}\Big[\textnormal{P}_{\bm{d}}\!\big[
\bm{f}(\bm{x},\bm{d}) \in \mathbb{X} 
\big] \ge 1 - \alpha_2
\Big]
\ge 1 - \alpha_1
\Big]
\ge 1 - \delta$.
\end{theorem}

\begin{proof}
%Consider the scenario program defined on the sample $\mathbb{S} = \{(\bm{x}^{(i)},\bm{d}^{(i)})\}_{i=1}^N$:
%\[
%\begin{aligned}
%&\min_{\xi \in \mathbb{R}, a \in \mathcal{A}} \xi \\
%&\text{s.t. } h(a,f(x_j,d_j)) \geq \gamma h(a,x_j) - \xi, \quad j=1,\dots,N, \\
%&\quad \xi \in [0,\overline{\xi}], \quad a \in \mathcal{A}.
%\end{aligned}
%\]

%Let $(\bm{a}^*(\mathbb{S}),\xi^*(\mathbb{S}))$ be an optimal solution to \eqref{eq:rbc_rc1}. 
According to Proposition \ref{prop:pac}, we have that, with probability at least $1-\delta$ over the sample $\mathbb{S}$ of size $N$, the optimal solution $(\bm{a}^*(\mathbb{S}),\xi^*(\mathbb{S}))$ satisfies $
\textnormal{P}_{(\bm{x},\bm{d})}\left[h(\bm{a}^*(\mathbb{S}),\bm{f}(\bm{x},\bm{d})) < \gamma h(\bm{a}^*(\mathbb{S}),\bm{x}) - \xi^*(\mathbb{S})\right] \le \alpha_1 \alpha_2$, where $\textnormal{P}_{(\bm{x},\bm{d})}:=\textnormal{P}_{\bm{x}}\times \textnormal{P}_{\bm{d}}$. Following the proof of Theorem \ref{thm:main-fixedpoint}, we obtain the conclusion. \qed
\end{proof}

\subsection{The Proof of Theorem \ref{thm:main-vc-final-corrected}}
\label{appendix:a2}

\begin{theorem}[PAC Safety Certification II Based on RBC]
%\label{thm:main-vc-final-corrected}
Fix $\delta_1 \in (0,1)$, $\delta_2 \in (0,1)$, $\alpha_1 \in (0,1)$, $\alpha_2 \in (0,1)$, and $l\in (0,1)$, where $\alpha_1<l\delta_2$.  %Let $\mathbb{S}=\{(\bm{x}^{(i)},\{\bm{d}^{(i,j)}\}_{j=1}^M)\}_{i=1}^N$ is an i.i.d. sample drawn from  $\textnormal{P}_{\bm{x}} \times \textnormal{P}_{\bm{d}}^M$, where $\{\bm{d}^{(i,j)}\}_{j=1}^M \stackrel{\text{i.i.d.}}{\sim} \textnormal{P}_{\bm{d}}$, 
Let $
N \geq \frac{2}{\alpha_1}\left(\ln\frac{1}{\delta_1} + m+1\right)$ and $M \geq \frac{1}{2\alpha_2^2}\ln\frac{1}{(1-l)\delta_2}$.% Define
%\[h_{\bm{a},\xi}(\bm{x},\{\bm{d}_i\}_{i=1}^M) := 1_{\big\{(\bm{x},\{\bm{d}_i\}_{i=1}^M)\mid \min_{i=1}^M h(\bm{a},\bm{f}(\bm{x},\bm{d}_j)) < \gamma h(\bm{a},\bm{x})-\xi\big\}}(\bm{x},\{\bm{d}_i\}_{i=1}^M),
%\]where

Let $(\bm{a}^*(\mathbb{S}), \xi^*(\mathbb{S}))$ be the solution obtained by solving~\eqref{eq:rbc_linear3}.  Then, we have 
\begin{equation*}
\label{eq:nested_guarantee_final0}
\begin{split}
&\textnormal{P}_{\mathbb{S}}\left[\textnormal{P}_{\bm{x}}\left[
\begin{split}
\begin{split}
\textnormal{P}_{\bm{d}}\left[
\begin{split}
&h(\bm{a}^*(\mathbb{S}),\bm{f}(\bm{x},\bm{d})) \\
&\geq  \gamma h(\bm{a}^*(\mathbb{S}),\bm{x})-\xi^*(\mathbb{S})
\end{split}
\right] \geq 1-\alpha_2 
\end{split}
\end{split}
\right] \geq 1 - \frac{\alpha_1}{l\delta_2}\right]\geq 1-\delta_1
\end{split}
\end{equation*}
\begin{equation*}
\begin{split}
\text{and thus, }\textnormal{P}_{\mathbb{S}}\left[\xi^*(\mathbb{S})=0\Rightarrow \textnormal{P}_{\bm{x}}\left[
\begin{split}
\begin{split}
\textnormal{P}_{\bm{d}}\left[
\begin{split}
\bm{f}(\bm{x},\bm{d}) \in \mathbb{X}
\end{split}
\right] \geq 1-\alpha_2
\end{split} 
\end{split}
\right] \geq 1 - \frac{\alpha_1}{l\delta_2}\right]\geq 1-\delta_1.
\end{split}
\end{equation*}
\end{theorem}

\begin{proof}
The proof proceeds in three main steps.%

\textbf{Step 1: Outer Layer Guarantee via Scenario Approaches}

According to Proposition \ref{prop:pac}, we have that,  with probability at least $1-\delta_1$ over the sample $\mathbb{S}$, the solution satisfies:
\[
\textnormal{P}_{\bm{x},\mathbb{D}^M}\left[\vee_{j=1}^M h(\bm{a}^*(\mathbb{S}),\bm{f}(\bm{x},\bm{d}_j)) < \gamma h(\bm{a}^*(\mathbb{S}),\bm{x})-\xi^*(\mathbb{S})\right] \leq \alpha_1,
\]
where $\{\bm{d}_j\}_{j=1}^M$ is a fresh, independent sample of $M$ disturbances, and $\textnormal{P}_{\bm{x},\mathbb{D}^M}:=\textnormal{P}_{\bm{x}}\times \textnormal{P}_{\bm{d}}^M$. Let $\mathbb{E}_{\text{outer}}$ denote this event. We have $\textnormal{P}_{\mathbb{S}}[\mathbb{E}_{\text{outer}}] \geq 1-\delta_1$. Define for any $\mathbb{S}$ and $\bm{x}$,
\[
\mu_M(\mathbb{S},\bm{x}) := \textnormal{P}_{\bm{d}}^M\left[\vee_{j=1}^M h(\bm{a}^*(\mathbb{S}),\bm{f}(\bm{x},\bm{d}_j)) < \gamma h(\bm{a}^*(\mathbb{S}),\bm{x})-\xi^*(\mathbb{S})\right].
\]
The outer guarantee implies that for $\mathbb{S} \in \mathbb{E}_{\text{outer}}$, 
$\textnormal{E}_{\bm{x}}[\mu_M(\mathbb{S},\bm{x})] \leq \alpha_1$ holds.

\textbf{Step 2: Decomposition on $\textnormal{P}_{\bm{x},\mathbb{D}^M}$ Based on Markov's Inequality.} Consider any $\mathbb{S} \in \mathbb{E}_{\text{outer}}$. By Markov's inequality applied to $\mu_M(\mathbb{S},\bm{x})$, we have
\[
\textnormal{P}_{\bm{x}}\left[\mu_M(\mathbb{S},\bm{x}) \geq  l\delta_2 \right] \leq \frac{\textnormal{E}_{\bm{x}}[\mu_M(\mathbb{S},\bm{x})]}{l\delta_2} \leq \frac{\alpha_1}{l\delta_2}.
\]
This implies $
\textnormal{P}_{\bm{x}}\left[\mu_M(\mathbb{S},\bm{x}) \leq l\delta_2 \right] \geq 1 - \frac{\alpha_1}{l\delta_2}$.

Now, for any $\bm{x}$ with $\mu_M(\mathbb{S},\bm{x}) \leq l\delta_2$, we have
\[
\textnormal{P}_{\bm{d}}^M\left[\vee_{j=1}^M h(\bm{a}^*(\mathbb{S}),\bm{f}(\bm{x},\bm{d}_j)) < \gamma h(\bm{a}^*(\mathbb{S}),\bm{x})-\xi^*(\mathbb{S})\right] \leq l\delta_2,\]
which implies $
\textnormal{P}_{\bm{d}}^M\left[\wedge_{j=1}^M h(\bm{a}^*(\mathbb{S}),\bm{f}(\bm{x},\bm{d}_j)) \geq \gamma h(\bm{a}^*(\mathbb{S}),\bm{x})-\xi^*(\mathbb{S})\right] \geq 1-l\delta_2$.

\textbf{Step 3: Application of Hoeffding's Inequality}

For fixed $(\bm{a},\bm{x},\xi)$, define the true violation probability,
\[
\mu(\bm{a},\bm{x},\xi) := \textnormal{P}_{\bm{d}}\left[h(\bm{a},\bm{f}(\bm{x},\bm{d})) < \gamma h(\bm{a},\bm{x})-\xi\right]
\]
and the empirical violation probability based on fresh disturbances,
\[
\hat{\mu}_M(\bm{a},\bm{x},\xi,\{\bm{d}_j\}_{j=1}^M) := \frac{1}{M}\sum_{j=1}^M 1_{\{h(\bm{a},\bm{f}(\bm{x},\bm{d}_j)) < \gamma h(\bm{a},\bm{x})-\xi\}}(\bm{d}_j).
\]

By Hoeffding's inequality, for any fixed $(\bm{a},\bm{x},\xi)$, we have 
\[
\textnormal{P}_{\bm{d}}^M\left[\mu(\bm{a},\bm{x},\xi) \geq \hat{\mu}_M(\bm{a},\bm{x},\xi,\{\bm{d}_j\}_{j=1}^M) + \alpha_2\right] \leq \exp(-2M\alpha_2^2).
\]

Setting $M \geq \frac{1}{2\alpha_2^2}\ln\frac{1}{(1-l)\delta_2}$ ensures that for $(\bm{a}^*(\mathbb{S}),\bm{x},\xi^*(\mathbb{S}))$,
\[
\textnormal{P}_{\bm{d}}^M\left[\mu(\bm{a}^*(\mathbb{S}),\bm{x},\xi^*(\mathbb{S})) \geq  \hat{\mu}_M(\bm{a}^*(\mathbb{S}),\bm{x},\xi^*(\mathbb{S}),\{\bm{d}_j\}_{j=1}^M) + \alpha_2\right] \leq (1-l)\delta_2.
\]
This holds since $\{\bm{d}_j\}_{j=1}^M \in \mathbb{D}^M$ is independent of $\mathbb{S}$, so $(\bm{a}^*(\mathbb{S}),\xi^*(\mathbb{S}))$ is independent of $\{\bm{d}_j\}_{j=1}^M$.

Let
\begin{enumerate}
    \item $\mathbb{E}_1$: The event that $\wedge_{j=1}^M h(\bm{a}^*(\mathbb{S}),\bm{f}(\bm{x},\bm{d}_j)) \geq \gamma h(\bm{a}^*(\mathbb{S}),\bm{x})-\xi^*(\mathbb{S})$. On this event, $\hat{\mu}_M(\bm{a}^*(\mathbb{S}),\bm{x},\xi^*(\mathbb{S}),\{\bm{d}_j\}_{j=1}^M) = 0$.
    \item $\mathbb{E}_2$: The event that $\mu(\bm{a}^*(\mathbb{S}),\bm{x},\xi^*(\mathbb{S})) < \hat{\mu}_M(\bm{a}^*(\mathbb{S}),\bm{x},\xi^*(\mathbb{S}),\{\bm{d}_j\}_{j=1}^M) + \alpha_2$.
\end{enumerate}

We have $\textnormal{P}_{\bm{d}}^M[\mathbb{E}_1] \geq 1-l\delta_2$ and $\textnormal{P}_{\bm{d}}^M[\mathbb{E}_2] \geq 1-(1-l) \delta_2$. By the union bound, we have
\[
\textnormal{P}_{\bm{d}}^M[\mathbb{E}_1 \cap \mathbb{E}_2] \geq 1 - \delta_2.
\]
Since $\hat{\mu}_M(\bm{a}^*(\mathbb{S}),\bm{x},\xi^*(\mathbb{S}),\{\bm{d}_j\}_{j=1}^M) = 0$ for $\{\bm{d}_j\}_{j=1}^M \in \mathbb{E}_1$, we have
\begin{equation}
\label{17}
\begin{split}
\textnormal{P}_{\bm{d}}^M\left[
\begin{split}
&\mu(\bm{a}^*(\mathbb{S}),\bm{x},\xi^*(\mathbb{S})) < \alpha_2+\hat{\mu}_M(\bm{a}^*(\mathbb{S}),\bm{x},\xi^*(\mathbb{S}),\{\bm{d}_j\}_{j=1}^M) \\
&\wedge \hat{\mu}_M(\bm{a}^*(\mathbb{S}),\bm{x},\xi^*(\mathbb{S}),\{\bm{d}_j\}_{j=1}^M)=0
\end{split}
\right] \geq 1 - \delta_2.
\end{split}
\end{equation}

\oomit{
\textcolor{red}{Actually, according to \eqref{17}, we can obtain $\mu(\bm{a}^*(\mathbb{S}),\bm{x},\xi^*(\mathbb{S})) < \alpha_2$ holds  deterministically with respect to $\textnormal{P}_{\bm{d}}^M$ for $\bm{x}$ satisfying $\mu_M(\mathbb{S},\bm{x}) \leq l \delta_2$: since $1-\delta_2>0$, there exists at least one $\{\bm{d}_j\}_{j=1}^M \in \mathbb{D}^M$ such that $\mu(\bm{a}^*(\mathbb{S}),\bm{x},\xi^*(\mathbb{S})) < \alpha_2+\hat{\mu}_M(\bm{a}^*(\mathbb{S}),\bm{x},\xi^*(\mathbb{S}),\{\bm{d}_j\}_{j=1}^M) \wedge \hat{\mu}_M(\bm{a}^*(\mathbb{S}),\bm{x},\xi^*(\mathbb{S}),\{\bm{d}_j\}_{j=1}^M)=0$ holds, implying $\mu(\bm{a}^*(\mathbb{S}),\bm{x},\xi^*(\mathbb{S})) < \alpha_2$ holds deterministically. However, I proceed with the proof (the following proof is logically correct) in order to derive the full nested structure and to connect it with the conclusion in \cite{xue2020pac}, as clarified in Remark 4. }
}

Therefore, we conclude %that on the event $\mathbb{E}_1 \cap \mathbb{E}_2$,  $\hat{\mu}_M(\bm{a}^*(\mathbb{S}),\bm{x},\xi^*(\mathbb{S}),\mathbb{D}^M)= 0$ and $\mu(\bm{a}^*(\mathbb{S}),\bm{x},\xi^*(\mathbb{S}))  < \hat{\mu}_M(\bm{a}^*(\mathbb{S}),\bm{x},\xi^*(\mathbb{S}),\mathbb{D}^M) + \alpha_2 = \alpha_2$. Therefore,
\begin{equation}
\label{delta_2}
\textnormal{P}_{\bm{d}}^M\left[\mu(\bm{a}^*(\mathbb{S}),\bm{x},\xi^*(\mathbb{S})) < \alpha_2\right] \geq 1 - \delta_2.
\end{equation}

\textbf{Step 4. Conclusion.} \eqref{delta_2} holds for any $\bm{x}$ with $\mu_M(\mathbb{S},\bm{x}) \leq l \delta_2$. %Thus, $\mu(\bm{a}^*(\mathbb{S}),\bm{x},\xi^*(\mathbb{S})) < \alpha_2$ holds deterministically holds for any $\bm{x}$ with $\mu_M(\mathbb{S},\bm{x}) \leq l \delta_2$. Also, 
Since $\textnormal{P}_{\bm{x}}[\mu_M(\mathbb{S},\bm{x}) \leq l\delta_2] \geq 1 - \frac{\alpha_1}{l\delta_2}$, we conclude that for $\mathbb{S} \in \mathbb{E}_{\text{outer}}$,
\[
\textnormal{P}_{\bm{x}}\left[\textnormal{P}_{\bm{d}}^M\left[\mu(\bm{a}^*(\mathbb{S}),\bm{x},\xi^*(\mathbb{S})) < \alpha_2\right] \geq 1 - \delta_2\right] \geq 1 - \frac{\alpha_1}{l\delta_2}.
\]
Thus, we obtain $\textnormal{P}_{\mathbb{S}}\left[\textnormal{P}_{\bm{x}}\left[\textnormal{P}_{\bm{d}}^M\left[\mu(\bm{a}^*(\mathbb{S}),\bm{x},\xi^*(\mathbb{S})) < \alpha_2\right] \geq 1 - \delta_2\right] \geq 1 - \frac{\alpha_1}{l\delta_2}\right]\geq 1-\delta_1$ and further
%\[\textnormal{P}_{\mathbb{S}}\left[\textnormal{P}_{\bm{x}}\left[\mu(\bm{a}^*(\mathbb{S}),\bm{x},\xi^*(\mathbb{S})) < \alpha_2\right] \geq 1 - \frac{\alpha_1}{l\delta_2}\right]\geq 1-\delta_1\]
%and 
\begin{equation*}
\begin{split}
&\textnormal{P}_{\mathbb{S}}\left[\textnormal{P}_{\bm{x}}\left[
\begin{split}
&\textnormal{P}_{\bm{d}}^M\left[
\begin{split}
&\textnormal{P}_{\bm{d}}\left[
\begin{split}
&h(\bm{a}^*(\mathbb{S}),\bm{f}(\bm{x},\bm{d})) \\
&\geq \gamma h(\bm{a}^*(\mathbb{S}),\bm{x})-\xi^*(\mathbb{S})
\end{split}
\right] \\
&\geq 1-\alpha_2 
\end{split}
\right]\\
&\geq 1-\delta_2
\end{split}
\right] \geq 1 - \frac{\alpha_1}{l\delta_2}\right]\geq 1-\delta_1.
\end{split}
\end{equation*}

Further, since $\textnormal{P}_{\bm{d}}\left[ h(\bm{a}^*(\mathbb{S}),\bm{f}(\bm{x},\bm{d})) \geq \gamma h(\bm{a}^*(\mathbb{S}),\bm{x})-\xi^*(\mathbb{S})\right] \geq 1-\alpha_2$ is a deterministic event with respect to $\textnormal{P}_{\bm{d}}^M$, it follows that \[\textnormal{P}_{\bm{d}}^M[\textnormal{P}_{\bm{d}}\left[ h(\bm{a}^*(\mathbb{S}),\bm{f}(\bm{x},\bm{d})) \geq \gamma h(\bm{a}^*(\mathbb{S}),\bm{x})-\xi^*(\mathbb{S})\right] \geq 1- \alpha_2]\in \{0,1\}.\] Moreover, since $\textnormal{P}_{\bm{d}}^M[\textnormal{P}_{\bm{d}}\left[ h(\bm{a}^*(\mathbb{S}),\bm{f}(\bm{x},\bm{d})) \geq \gamma h(\bm{a}^*(\mathbb{S}),\bm{x})-\xi^*(\mathbb{S})\right] \geq 1- \alpha_2] > 0$ (as $1-\delta_2 > 0$), it must be that $\textnormal{P}_{\bm{d}}\left[ h(\bm{a}^*(\mathbb{S}),\bm{f}(\bm{x},\bm{d})) \geq \gamma h(\bm{a}^*(\mathbb{S}),\bm{x})-\xi^*(\mathbb{S})\right] \geq 1- \alpha_2$ holds deterministically with respect to $\textnormal{P}_{\bm{d}}^M$ for such $\bm{x}$ with $\mu_M(\mathbb{S},\bm{x}) \leq l \delta_2$. As a result, we obtain 
\begin{equation*}
\begin{split}
\textnormal{P}_{\mathbb{S}}\left[\textnormal{P}_{\bm{x}}\left[
\begin{split}
&\textnormal{P}_{\bm{d}}\left[h(\bm{a}^*(\mathbb{S}),\bm{f}(\bm{x},\bm{d})) \geq  \gamma h(\bm{a}^*(\mathbb{S}),\bm{x})-\xi^*(\mathbb{S})
\right]\\
&\geq 1-\alpha_2
  \end{split}
  \right] \geq 1 - \frac{\alpha_1}{l\delta_2}\right]\geq 1-\delta_1.
\end{split}
\end{equation*}

The proof is completed with $\xi^*(\mathbb{S})=0$. \qed
\end{proof}

\subsection{The Proof of Theorem \ref{thm:main-vc-final-sbc1}}
\label{appendix:a3}

\begin{theorem}[PAC Safety Certification III Based on SBC]
%\label{thm:main-vc-final-sbc1}
Fix $\delta_1 \in (0,1)$, $\delta_2 \in (0,1)$, $\alpha_1 \in (0,1)$, $\alpha_2 \in (0,1)$, and $l\in (0,1)$, where $\alpha_1<l\delta_2$, % Let $\mathbb{S}=\{(\bm{x}^{(i)},\{\bm{d}^{(i,j)}\}_{j=1}^M)\}_{i=1}^N$ is an i.i.d. sample drawn from  $\textnormal{P}_{\bm{x}} \times \textnormal{P}_{\bm{d}}^M$, where $\{\bm{d}^{(i,j)}\}_{j=1}^M \stackrel{\text{i.i.d.}}{\sim} \textnormal{P}_{\bm{d}}$, $
and let $N \geq \frac{2}{\alpha_1}\left(\ln\frac{1}{\delta_1} + m+1\right)$, and $M \geq \frac{U_a^2}{2\tau^2}\ln\frac{1}{(1-l)\delta_2}$.

%Let $\mathbb{S}=\{(\bm{x}^{(i)},\{\bm{d}^{(i,j)}\}_{j=1}^M)\}_{i=1}^N$ is an i.i.d. training sample drawn from  $\textnormal{P}_{\bm{x}} \times \textnormal{P}_{\bm{d}}^M$, where $\{\bm{d}^{(i,j)}\}_{j=1}^M \stackrel{\text{i.i.d.}}{\sim} \textnormal{P}_{\bm{d}}$, $N \geq \frac{2}{\alpha_1}\left(\ln\frac{1}{\delta_1} + m+1\right)$, and $M \geq \frac{1}{2\tau^2}\ln\frac{4}{\delta_2}$.

Let $(\bm{a}^*(\mathbb{S}), \lambda^*(\mathbb{S}))$ be the solution obtained by solving~\eqref{eq:sbc_linear3}. Then, we have 
$\textnormal{P}_{\mathbb{S}}\left[\textnormal{P}_{\bm{x}}\left[
\textnormal{E}_{\bm{d}}[h(\bm{a}^*(\mathbb{S}),\bm{f}(\bm{x},\bm{d}))] \leq  h(\bm{a}^*(\mathbb{S}),\bm{x})+\lambda^*(\mathbb{S})
\right] \geq 1 - \frac{\alpha_1}{l\delta_2}\right]\geq 1-\delta_1$,
%and thus, 
%\begin{equation*}
%\label{eq:nested_guarantee_final1}
%\begin{split}
%\textnormal{P}_{\mathbb{S}}\left[\textnormal{P}_{\bm{x}}\left[
%\begin{split}
%\textnormal{P}_{\bm{d}}^M\left[
%\begin{split}
%&\textnormal{E}_{\bm{d}}[h(\bm{a}^*(\mathbb{S}),\bm{f}(\bm{x},\bm{d}))]\\
%&> h(\bm{a}^*(\mathbb{S}),\bm{x})
%\end{split}
%\right] \geq 1-\delta_2
%\end{split}
%\right] \geq 1 - \frac{\alpha_1}{l\delta_2}\right]\geq 1-\delta_1,
%\end{split}
%\end{equation*}
and thus, $\textnormal{P}_{\mathbb{S}}\left[\textnormal{P}_{\bm{x}}\left[
\textnormal{P}_{\bm{d}}[\bm{f}(\bm{x},\bm{d})\in \mathbb{X}]\geq 1-h(\bm{a}^*(\mathbb{S}),\bm{x})-\lambda^*(\mathbb{S})
\right] \geq 1 - \frac{\alpha_1}{l\delta_2}\right]\geq 1-\delta_1$.
\end{theorem}

\begin{proof}
The proof proceeds in three main steps.

\textbf{Step 1: Outer Layer Guarantee via Scenario Approaches}

According to Proposition \ref{prop:pac}, we have that, with probability at least $1-\delta_1$ over the sample set $\mathbb{S}$, the solution satisfies:
\[
\textnormal{P}_{\bm{x},\mathbb{D}^M}\left[\frac{1}{M}\sum_{j=1}^M h(\bm{a}^*(\mathbb{S}),\bm{f}(\bm{x},\bm{d}_j)) > h(\bm{a}^*(\mathbb{S}),\bm{x}) + \lambda^*(\mathbb{S})-\tau\right] \leq \alpha_1,\]
where $\{\bm{d}_j\}_{j=1}^M \in \mathbb{D}^M$ denotes a fresh, independent sample of $M$ disturbances, and $\textnormal{P}_{\bm{x},\mathbb{D}^M}:=\textnormal{P}_{\bm{x}}\times \textnormal{P}_{\bm{d}}^M$. Let $\mathbb{E}_{\text{outer}}$ denote this event. We have $\textnormal{P}_{\mathbb{S}}[\mathbb{E}_{\text{outer}}] \geq 1-\delta_1$. Define for any $\mathbb{S}$ and $\bm{x}$,
\[
\mu_M(\mathbb{S},\bm{x}): = \textnormal{P}_{\bm{d}}^M\left[\frac{1}{M}\sum_{j=1}^M h(\bm{a}^*(\mathbb{S}),\bm{f}(\bm{x},\bm{d}_j)) > h(\bm{a}^*(\mathbb{S}),\bm{x}) + \lambda^*(\mathbb{S})-\tau\right].
\]
The outer guarantee implies that for $\mathbb{S} \in  \mathbb{E}_{\text{outer}}$, $
\textnormal{E}_{\bm{x}}[\mu_M(\mathbb{S},\bm{x})] \leq \alpha_1$ holds.

\textbf{Step 2: Decomposition on $\textnormal{P}_{\bm{x},\mathbb{D}^M}$ Based on Markov's Inequality.}
Consider any $\mathbb{S} \in  \mathbb{E}_{\text{outer}}$. By applying Markov's inequality to $\mu_M(\mathbb{S},\bm{x})$, we obtain
\[
\textnormal{P}_{\bm{x}}\left[\mu_M(\mathbb{S},\bm{x}) \geq  l \delta_2\right] \leq \frac{\textnormal{E}_{\bm{x}}[\mu_M(\mathbb{S},\bm{x})]}{l\delta_2} \leq \frac{\alpha_1}{l\delta_2}.
\]
This implies $\textnormal{P}_{\bm{x}}\left[\mu_M(\mathbb{S},\bm{x}) \leq l \delta_2\right] \geq 1 - \frac{\alpha_1}{l\delta_2}$.

Now, for any $\bm{x}$ with $\mu_M(\mathbb{S},\bm{x}) \leq l \delta_2$, we have
\[
\textnormal{P}_{\bm{d}}^M\left[\frac{1}{M}\sum_{j=1}^M h(\bm{a}^*(\mathbb{S}),\bm{f}(\bm{x},\bm{d}_j)) > h(\bm{a}^*(\mathbb{S}),\bm{x}) + \lambda^*(\mathbb{S})-\tau\right] \leq l \delta_2,
\]
which implies
\[
\textnormal{P}_{\bm{d}}^M\left[\frac{1}{M}\sum_{j=1}^M h(\bm{a}^*(\mathbb{S}),\bm{f}(\bm{x},\bm{d}_j)) \leq h(\bm{a}^*(\mathbb{S}),\bm{x}) + \lambda^*(\mathbb{S})-\tau\right] \geq 1-l \delta_2.
\]

\textbf{Step 3: Application of Hoeffding's Inequality} %to Obtain \[\textnormal{P}_{\bm{d}}^M\left[
%\textnormal{E}_{\bm{d}}[h(\bm{a}^*(\mathbb{S}),\bm{f}(\bm{x},\bm{d}))] \leq h(\bm{a}^*(\mathbb{S}),\bm{x})+\lambda^*(\mathbb{S})
%\right] \geq 1-\delta_2.\]}

For fixed $(\bm{a},\bm{x})$, define the empirical mean:
\[
\hat{\mu}_M(\bm{a},\bm{x},\{\bm{d}_j\}_{j=1}^M): = \frac{1}{M}\sum_{j=1}^M h(\bm{a},\bm{f}(\bm{x},\bm{d}_j))
\]
and the true expectation: $
\mu(\bm{a},\bm{x}) := \textnormal{E}_{\bm{d}}[h(\bm{a},\bm{f}(\bm{x},\bm{d}))]$.

By Hoeffding's inequality, for any fixed $(\bm{a},\bm{x})$, we have
\[
\textnormal{P}_{\bm{d}}^M\left[\mu(\bm{a},\bm{x}) \geq  \hat{\mu}_M(\bm{a},\bm{x},\{\bm{d}_j\}_{j=1}^M) + \tau\right] \leq \exp(\frac{-2M\tau^2}{U_a^2}).
\]

Setting $M \geq \frac{U_a^2}{2\tau^2}\ln\frac{1}{(1-l)\delta_2}$ ensures that for ($\bm{a}^*(\mathbb{S}),\bm{x})$, 
\[
\textnormal{P}_{\bm{d}}^M\left[\mu(\bm{a}^*(\mathbb{S}),\bm{x}) \geq  \hat{\mu}_M(\bm{a}^*(\mathbb{S}),\bm{x},\{\bm{d}_j\}_{j=1}^M) + \tau\right] \leq (1-l)\delta_2
\] holds. This holds since $\{\bm{d}_j\}_{j=1}^M$ is independent of $\mathbb{S}$, so $\bm{a}^*(\mathbb{S})$ is independent of $\{\bm{d}_j\}_{j=1}^M$.

Let
\begin{enumerate}
    \item $\mathbb{E}_1$: The event that $\frac{1}{M}\sum_{j=1}^M h(\bm{a}^*(\mathbb{S}),\bm{f}(\bm{x},\bm{d}_j)) \leq h(\bm{a}^*(\mathbb{S}),\bm{x}) + \lambda^*(\mathbb{S}))-\tau$.
    \item $\mathbb{E}_2$: The event that $\mu(\bm{a}^*(\mathbb{S}),\bm{x}) < \hat{\mu}_M(\bm{a}^*(\mathbb{S}),\bm{x},\{\bm{d}_j\}_{j=1}^M) + \tau$.
\end{enumerate}

We have $\textnormal{P}_{\bm{d}}^M[\mathbb{E}_1] \geq 1-l \delta_2$ and $\textnormal{P}_{\bm{d}}^M[\mathbb{E}_2] \geq 1-(1-l) \delta_2$. By the union bound, we obtain
\[
\textnormal{P}_{\bm{d}}^M[\mathbb{E}_1 \cap \mathbb{E}_2] \geq 1 - \delta_2.
\]

On the event $\mathbb{E}_1 \cap \mathbb{E}_2$, we have
\[
\mu(\bm{a}^*(\mathbb{S}),\bm{x}) < \hat{\mu}_M(\bm{a}^*(\mathbb{S}),\bm{x},\{\bm{d}_j\}_{j=1}^M) + \tau \leq h(\bm{a}^*(\mathbb{S}),\bm{x}) + \lambda^*(\mathbb{S}).
\]

Therefore, we obtain
\begin{equation*}
\textnormal{P}_{\bm{d}}^M\left[
\begin{split}
   & \textnormal{E}_{\bm{d}}[h(\bm{a}^*(\mathbb{S}),\bm{f}(\bm{x},\bm{d}))] <  \hat{\mu}_M(\bm{a}^*(\mathbb{S}),\bm{x},\{\bm{d}_j\}_{j=1}^M) + \tau \wedge \\
   &\hat{\mu}_M(\bm{a}^*(\mathbb{S}),\bm{x},\{\bm{d}_j\}_{j=1}^M) + \tau\leq  h(\bm{a}^*(\mathbb{S}),\bm{x}) + \lambda^*(\mathbb{S})
    \end{split}
    \right] \geq 1 - \delta_2,
\end{equation*}
implying 
\begin{equation}
\label{E-delta_2}
\textnormal{P}_{\bm{d}}^M\left[\textnormal{E}_{\bm{d}}[h(\bm{a}^*(\mathbb{S}),\bm{f}(\bm{x},\bm{d}))] \leq   h(\bm{a}^*(\mathbb{S}),\bm{x}) + \lambda^*(\mathbb{S})\right] \geq 1 - \delta_2.
\end{equation}

Since $\textnormal{P}_{\bm{d}}^M[\mathbb{E}_1 \cap \mathbb{E}_2] > 0$ (as $1-\delta_2 > 0$), we obtain $\textnormal{E}_{\bm{d}}[h(\bm{a}^*(\mathbb{S}),\bm{f}(\bm{x},\bm{d}))] \leq  h(\bm{a}^*(\mathbb{S}),\bm{x}) + \lambda^*(\mathbb{S})$ holds deterministically with respect to $\textnormal{P}_{\bm{d}}^M$ for such $\bm{x}$ satisfying $\mu_M(\mathbb{S},\bm{x}) \geq  l \delta_2$.

\textbf{Step 4: Conclusion.}  Since $\textnormal{E}_{\bm{d}}[h(\bm{a}^*(\mathbb{S}),\bm{f}(\bm{x},\bm{d}))] \leq  h(\bm{a}^*(\mathbb{S}),\bm{x}) + \lambda^*(\mathbb{S})$ holds deterministically with respect to $\textnormal{P}_{\bm{d}}^M$ for such $\bm{x}$ satisfying $\mu_M(\mathbb{S},\bm{x}) \geq  l \delta_2$ and $\textnormal{P}_{\bm{x}}[\mu_M(\mathbb{S},\bm{x}) \leq l \delta_2] \geq 1 - \frac{\alpha_1}{l\delta_2}$, we conclude that for $\mathbb{S} \in  \mathbb{E}_{\text{outer}}$,
\[
\textnormal{P}_{\bm{x}}\left[\textnormal{E}_{\bm{d}}[h(\bm{a}^*(\mathbb{S}),\bm{f}(\bm{x},\bm{d}))] \leq h(\bm{a}^*(\mathbb{S}),\bm{x}) + \lambda^*(\mathbb{S})\right] \geq 1 - \frac{\alpha_1}{l\delta_2}.
\]

Since $\textnormal{P}_{\mathbb{S}}[\mathbb{E}_{\text{outer}}] \geq 1-\delta_1$, the overall probability is at least $1-\delta_1$. Also, based on Proposition \ref{pro_one}, we have the conclusion.\qed
\end{proof}
\section{Additional Experimental Results and Implementation Details}
\label{sec:app_ex}

\subsection{Implementation of the Sum-of-Squares Programming Baselines}
\label{sec:app_sos}

In this subsection, we present the details of the RBC-SOS and SBC-SOS baseline methods based on sum-of-squares programming. By leveraging the SOS decomposition for multivariate polynomials, constraints \eqref{eq:rbf} in Definition \eqref{def:rbc} and \eqref{eq:sbf} in Definition \ref{def:sbf} can be recast as semidefinite programmings, which can be solved efficiently in polynomial time via interior-point methods. However, constraints~\eqref{eq:rbf1} and ~\eqref{eq:sbf2} require the inequalities to hold for all $\bm{x}\in \mathbb{R}^n\setminus\mathbb{X}$, which is overly stringent in practice. To address this issue, we introduce relaxed formulations of the constraints that are better suited to SOS programming:
\begin{definition}
    Given $\gamma \in (0,1)$ and a safe set $\mathbb{X}$, a function $h(\cdot):\mathbb{R}^n\rightarrow\mathbb{R}$ is a RBC for the system 
    \eqref{eq:system} if the following inequalities hold:
    \begin{subequations}
    \label{eq:rbf_hat}
    \begin{empheq}[left=\empheqlbrace]{align}
        &h(\bm{x}) < 0, & \forall \bm{x} \in \widehat{\mathbb{X}}\setminus\mathbb{X},  \\
        &h(\bm{x}) \geq 0, &\forall \bm{x}\in \mathbb{X},  \\
        &h(\bm{f}(\bm{x},\bm{d})) \geq \gamma h(\bm{x}), & \forall \bm{x} \in \mathbb{X}, ~\forall \bm{d} \in \mathbb{D}. 
    \end{empheq}
    \end{subequations}
    where $\widehat{\mathbb{X}}$ is a set containing the union of the set $\mathbb{X}$ and all reachable states starting from $\mathbb{X}$ within one step, i.e., 
\begin{equation}
\label{eq:x_hat}
\widehat{\mathbb{X}}\supset \{\bm{y}\in \mathbb{R}^n\mid \bm{y}=\bm{f}(\bm{x},\bm{d}), \bm{x}\in \mathbb{X}, \bm{d}\in \mathbb{D}\}\cup \mathbb{X}.
\end{equation}
\end{definition}
\begin{definition}[\cite{xue2024finite}]
    \label{def:sbf_hat}
    Given a safe set $\mathbb{X}$, a function $h(\cdot):\mathbb{R}^n\rightarrow\mathbb{R}$ is a SBC for the system 
    \eqref{eq:system} if there exists a scalar $\lambda \in [0,1]$ such that the following inequalities hold:
    \begin{subequations}
    \label{eq:sbf_hat}
    \begin{empheq}[left=\empheqlbrace]{align}
        &h(\bm{x}) \geq 0, &\forall \bm{x}\in \mathbb{X}, \\
        &h(\bm{x}) \geq 1, & \forall \bm{x} \in \widehat{\mathbb{X}}\setminus\mathbb{X},  \\
        &\mathbb{E}_{\bm{d}}[h(\bm{f}(\bm{x},\bm{d}))] - h(\bm{x}) \leq \lambda, &\forall \bm{x} \in \mathbb{X},
    \end{empheq}
    \end{subequations}
    where $\widehat{\mathbb{X}}$ is a set satisfying \eqref{eq:x_hat}.
\end{definition}

Compared with constraints \eqref{eq:rbf} and \eqref{eq:sbf}, constraints \eqref{eq:rbf_hat} and \eqref{eq:sbf_hat} replace $\mathbb{R}^n$ with the set $\widehat{\mathbb{X}}$. This substitution substantially relaxes the constraints and makes them easier to satisfy. Assume that $\mathbb{X}=\{\bm{x}\in\mathbb{R}^n \mid g(\bm{x}) \leq  0\}$ and $\widehat{\mathcal{X}} =\{\bm{x}\in\mathbb{R}^n \mid \hat{g}(\bm{x}) \leq  0\}$, where $g(\bm{x})$ and $\hat{g}(\bm{x})$ are polynomial functions over $\bm{x} \in \mathbb{R}^n$, and that the disturbance set is $\mathbb{D} = \{\bm{d}\in\mathbb{R}^d \mid p(\bm{d}) \leq 0\}$, where $p(\bm{d})$ is a known polynomial function. Under these assumptions, constraint \eqref{eq:rbf_hat} can be recast as an SOS programming in \eqref{eq:rbc_sos}, corresponding to the RBC-SOS baseline, while constraint \eqref{eq:sbf_hat} can be recast as an SOS programming in \eqref{eq:sbc_sos}, corresponding to the SBC-SOS baseline. In \eqref{eq:rbc_sos} and \eqref{eq:sbc_sos}, $\sum[\bm{y}]$ denotes the set of sum-of-squares polynomials over variables $\bm{y}$, i.e., 
\[\sum[\bm{y}]=\{p\in \mathbb{R}[\bm{y}]\mid p=\sum_{i=1}^k q^2_i(\bm{y}), q_i(\bm{y})\in \mathbb{R}[\bm{y}],i=1,\ldots,k\},\] where $\mathbb{R}[\bm{y}]$ denotes the ring of polynomials in variables $\bm{y}$.
\begin{align}
    \label{eq:rbc_sos}
    &\min_{\xi,\bm{a},s_1(\bm{x},\bm{d}),s_2(\bm{x},\bm{d}),s_3(\bm{x}),s_4(\bm{x}),s_5(\bm{x})} \xi\\
    \text{s.t.}&\begin{cases}
     h(\bm{a},\bm{f}(\bm{x},\bm{d})) - \gamma h(\bm{a},\bm{x}) + \xi + s_1(\bm{x}, \bm{d}) g(\bm{x}) + s_2(\bm{x}, \bm{d}) p(\bm{d}) \in \sum [\bm{x},\bm{d}], \\
        -h(\bm{a},\bm{x}) - \epsilon_0 + s_3(\bm{x}) \hat{g}(\bm{x}) - s_4(\bm{x}) g(\bm{x}) \in \sum [\bm{x}],\\
        h(\bm{a},\bm{x}) + s_5(\bm{x}) g(\bm{x}) \in \sum[\bm{x}], \\
        s_1(\bm{x}, \bm{d}),s_2(\bm{x},\bm{d}) \in \sum [\bm{x},\bm{d}]; s_3(\bm{x}),s_4(\bm{x}),s_5(\bm{x}) \in \sum [\bm{x}],\\
    \end{cases}\notag
\end{align}

\begin{align}
    \label{eq:sbc_sos}
    &\min_{\lambda,\bm{a},s_i(\bm{x}),i=1,\ldots,4} \lambda + \frac{1}{N_o}\textstyle\sum_{k=1}^{N_o}h(\bm{a},\bm{x}'_k)\\
    \text{s.t.}&
    \begin{cases}
       h(\bm{a},\bm{x}) + s_1(\bm{x}) g(\bm{x}) \in \sum [\bm{x}],\\
       h(\bm{a},\bm{x}) - 1 + s_2(\bm{x}) \hat{g}(\bm{x}) -s_3(\bm{x}) g(\bm{x}) \in \sum[\bm{x}],\\
       h(\bm{a},\bm{x}) + \lambda -\mathbb{E}_{\bm{d}}[h(\bm{a},\bm{f}(\bm{x},\bm{d}))] + s_4(\bm{x}) g(\bm{x}) \in \sum [\bm{x}],\\
       s_1(\bm{x}),s_2(\bm{x}),s_3(\bm{x}),s_4(\bm{x}) \in \sum [\bm{x}],\\
        \lambda \in [0,1],
    \end{cases}\notag
\end{align}

For numerical stability, we constrain each coefficient in $\bm{a}$ to lie within the interval $[-100, 100]$. In \eqref{eq:rbc_sos}, we set $\epsilon_0 = 10^{-6}$ to ensure strict satisfaction of the inequality. The objective in \eqref{eq:sbc_sos} matches that of \eqref{eq:sbc_linear3} to ensure a fair comparison. In the experiments, we keep the degree of $h(\bm{a},\bm{x})$ in the SOS-based methods as close as possible to the degree used in our methods. If the SOS solver encounters numerical issues or exceeds the time limit, we adjust the degree of $h(\bm{a},\bm{x})$ accordingly.

\subsection{Experiments Settings and Hyperparameters}
\label{sec:app_setting}
In this subsection, we describe the experimental setup and hyperparameter configurations used in Tab. \ref{tab:results}. 

All experiments were run on a system with Ubuntu 22.04, a 72-core Intel Xeon Gold 6154 CPU (3.00 GHz), an NVIDIA A100 GPU, and 512 GB of RAM. Linear programs were solved using Gurobi 11.0.3 \cite{gurobi}, while sum-of-squares programs were handled with the YALMIP toolbox in MATLAB  \cite{lofberg2004yalmip} and the Mosek 10.1.21 solver \cite{aps2019mosek}. We impose a computation time limit of 0.5 hour for all examples.

The specific settings for examples \ref{ex:arch}-\ref{ex:lorenz} are shown below.

% \begin{example}
% \label{ex:jet}
%     Consider the discretization of Moore-Greitzer model of a jet engine \cite{aylward2008stability}:
% \begin{equation*}
% \begin{cases}
% x(t+1)=x(t) + 0.1(-y(t)-1.5x^2(t)-0.5x^3(t)),\\
% y(t+1)=y(t) + 0.1(3x(t)-y(t)+d(t)),
% \end{cases}
% \end{equation*}
%  where $d(\cdot)$ is uniformly distributed over $[-1,1]$. The safe set is $\mathbb{X}=\{\,(x,y)\mid x^2+y^2 \leq 1\,\}$.
% \end{example}

% \begin{example}
% \label{ex:vanderpol}
%     Consider the VanderPol oscillator, adapted from \cite{henrion2013convex}, which has been widely used in both physical and biological sciences:
% \begin{equation*}
%     \begin{cases}
% x(t+1)=x(t) - 0.02y(t),\\
% y(t+1)=y(t) + 0.01\big(0.8x(t)+10(x^2(t) - 0.21)y(t)+d(t)\big),
% \end{cases}
% \end{equation*}
% where $d(\cdot)$ follows the standard normal distribution with mean $0$ and variance $1$. The safe set is $\mathbb{X}=\{\,(x,y)\mid x^2+y^2 \leq 1\,\}$.
% \end{example}
\setcounter{example}{1}
\begin{example}
\label{ex:arch}
    Consider the following system adapted from \cite{sogokon2016non}:
\begin{equation*}
    \begin{cases}
x(t+1)=x(t)+0.01\big(x(t)-x^3(t)+y(t)-x(t)y^2(t)+d_1(t)\big),\\
y(t+1)=y(t)+0.01\big(-x(t) + y(t) - x^2(t)y(t) - y^3(t) + d_2(t)\big),
\end{cases}
\end{equation*}
where $d_1(t)$ and $d_2(t)$ are independent and follow the uniform distribution on the interval $[-0.5, 0.5]$. The safe set is $\mathbb{X}=\{\,(x,y)\mid -3 \leq x,y \leq 3 \,\}$.
\end{example}

% \begin{example}
%     \label{ex:bc4}
%     Consider the following system adapted from \cite{zhang2018safety}:
% \begin{equation*}
%     \begin{cases}
% x(t+1)=x(t)+0.01\big(-x(t)+2x^2(t)y(t)\big),\\
% y(t+1)=y(t)+0.01\big(-y(t)+d(t)\big),
% \end{cases}
% \end{equation*}
% where follows a $\text{Beta}(25,25)$ distribution. The safe set is $\mathbb{X}=\{\,(x,y)^{\top}\mid x^2+y^2\leq 1 \,\}$.
% \end{example}

\begin{example}
\label{ex:stable}
    Consider the following system adapted from \cite{ben2015stability}:
\begin{equation*}
\begin{cases}
x(t+1)=x(t)+0.01\big(-x(t)+y(t) -z(t)-x(t)d_1(t)\big),\\
y(t+1)=y(t)+0.01\big(-x(t)(z(t)+1)-y(t)-y(t)d_2(t)\big),\\
z(t+1)=z(t)+0.01\big(0.76524x(t)-4.7037z(t)-z(t)d_3(t)\big),\\
\end{cases}
\end{equation*}
where $d_1(t)$, $d_2(t)$, and $d_3(t)$ are independent and follow uniform distributions on $[1, 2]$, $[1,2]$, and $[2,3]$, respectively. The safe set is $\mathbb{X}=\{\,(x,y,z)\mid -1 \leq x,y,z \leq 1\,\}$.
\end{example}

\begin{example}
    \label{ex:4d}
    Consider a four dimensional system:
    \begin{equation*}
        \begin{cases}
            x_1(t+1) = x_1(t) + 0.01(-x_1(t) + d(t)),\\
            x_2(t+1) = x_2(t) + 0.01(x_1(t)-2x_2(t)),\\
            x_3(t+1) = x_3(t) + 0.01(x_1(t)-4x_3(t)),\\
            x_4(t+1) = x_4(t) + 0.01(x_1(t)-3x_4(t)),
\end{cases}
    \end{equation*}
    where $d(t)$ follows a $\text{Beta}(10,10)$ distribution. The safe set is $\mathbb{X}=\{\,(x_1,\dots,x_4)^\top\mid \sum_{i=1}^4 x_i^2 \leq 1\,\}$.
\end{example}

\begin{example}
    \label{ex:6d} Consider a model adapted from \cite{edwards2024fossil} with a discrete time $\tau_s=0.01$,
    \begin{equation*}
        \begin{cases}
        x_1(t+1)=x_1(t) + \tau_s(x_2(t)x_4(t)-x_1^3(t)),\\
        x_2(t+1)=x_2(t) + \tau_s(-3x_1(t)x_4(t)-x_2^3(t)),\\
        x_3(t+1)=x_3(t) + \tau_s(-x_3(t)-3x_1(t)x_4^3(t)),\\
        x_4(t+1)=x_4(t) + \tau_s(-x_4(t)+x_1(t)x_3(t)),\\
        x_5(t+1)=x_5(t) + \tau_s(-x_5(t)+x_6^3(t)),\\
        x_6(t+1)=x_6(t) + \tau_s(-x_5(t)-x_6(t) + x_3^4(t) - x_6(t) d(t)),\\
        \end{cases}
    \end{equation*}
    with the safe set $\mathbb{X}=\{\,(x_1,\ldots,x_6)\mid \sum_{i=1}^6 x_i^2-1 \leq 0\,\}$. The disturbance $d(t)$ follows the uniform distribution on the interval $[0.5, 1.0]$.
\end{example}

\setcounter{example}{6}
\begin{example}
\label{ex:pendulum}
    Consider a discrete-time system describing a whirling pendulum \cite{chesi2009estimating}:
\begin{equation*}
\begin{cases}
x(t+1)=x(t) + 0.1y(t),\\
y(t+1)=y(t) + 0.1(-\frac{2y(t)}{d(t)}+0.81\sin(x(t))\cos(x(t)) - \sin(x(t))),
\end{cases}
\end{equation*}
 where $d(\cdot)$ is uniformly distributed over $[0.9,1.1]$. The safe set is $\mathbb{X}=\{\,(x,y)\mid -1\leq x,y \leq 1\,\}$.
\end{example}

% \begin{example}
% \label{ex:temp}
%     We consider a temperature regulation problem involving three rooms, modeled by the following discrete-time stochastic system \cite{salamati2024data}: 
% \begin{equation*}
%     \begin{cases}
%         T_1(t+1)= \left(1-\tau_s\left(\alpha+\alpha_e\right)\right) T_1(t)+\tau_s \alpha T_2(t)+\tau_s \alpha_e T_e+d_1(t) \\
%         T_2(t+1)=  \left(1-\tau_s\left(2 \alpha+\alpha_e\right)\right) T_2(t)+\tau_s \alpha\left(T_1(t)+T_3(t)\right)+\tau_s \alpha_e T_e+d_2(t) \\
%         T_3(t+1)=  \left(1-\tau_s\left(\alpha+\alpha_e\right)\right) T_3(t)+\tau_s \alpha T_2(t)+\tau_s \alpha_e T_e+d_3(t)
% \end{cases}
% \end{equation*}
% where $T_1(t)$, $T_2(t)$, and $T_3(t)$ denote the temperatures of the three rooms at time $t$. The Constants $\alpha_e = 8 \times 10^{-3}$ and $\alpha = 6.2 \times 10^{-3}$ represent heat exchange coefficients between rooms and the ambient, and individual rooms, respectively. The ambient temperature is $T_e = 10$, and the discrete time is $\tau_s = 5$. The disturbance terms $d_1(t)$, $d_2(t)$, and $d_3(t)$ are independently and uniformly distributed over $[-0.5,0.5]$.
% The safe set is defined as $\mathbb{X}=\{\,(T_1,T_2,T_3)\mid 17 \leq T_1,T_2,T_3 \leq 30\,\}$.
% \end{example}

\begin{example}
    \label{ex:sank} Consider a model adapted from \cite{sankaranarayanan2013lyapunov} with a discrete time $0.01$,
    \begin{equation*}
        \begin{cases}
        x_1(t+1)=x_1(t) + \tau_s(-x_1(t) + x_2^3(t) - 3x_3(t)x_4(t) + d(t)),\\
        x_2(t+1)=x_2(t) + \tau_s(-x_1(t) - x_2^3(t)),\\
        x_3(t+1)=x_3(t) + \tau_s(x_1(t)x_4(t)-x_3(t)),\\
        x_4(t+1)=x_4(t) + \tau_s(x_1(t)x_3(t)-x_4^3(t)),\\
        \end{cases}
    \end{equation*}
    with $\tau_s=0.01$ and the safe set $\mathbb{X}=\{\,(x_1,x_2,x_3,x_4)^{\top}\mid \sum_{i=1}^4 x_i^2 \leq 1\,\}$. The disturbance $d(t)$ follows the uniform distribution on the interval $[-1,1]$.
\end{example} 

\begin{example}
\label{ex:lorenz}
    Consider a 7-dimensional Lorenz model adapted from \cite{lorenz1996predictability}:
\begin{equation*}
\begin{split}
    &x_i(t+1) = x_i(t) + \tau_s\Big(\big(x_{i+1}(t) - x_{i-2}(t)\big)x_{i-1}(t)- x_i(t) + d_i(t)\Big),~ i = 1,\ldots,7
\end{split}
\end{equation*}
with $\tau_s = 0.01$, $x_{0}(t) = x_7(t)$, $x_{-1}(t) = x_6(t)$, and $x_8(t) = x_1(t)$. The disturbance terms $d_1(t), \ldots, d_7(t)$ are independently and uniformly distributed over $[-1,1]$.
The safe set is defined as $\mathbb{X}=\{\,(x_1,\ldots,x_7)\mid -1 \leq x_1,\ldots,x_7 \leq 1\,\}$. 
\end{example}

All experiments reported in Tab. \ref{tab:results} are conducted with the random number seed set to 0. 
Next, we summarize the parameter settings used in the experiments reported in Tab. \ref{tab:results}. In all experiments, the number of state samples $N_o$ used in the objective is fixed at $1000$. For RBC-I and RBC-II, we set $\gamma$ in Definition \ref{def:rbc} to $0.01$, the constant $C$ in \eqref{barrier_form} to $-1$, and $\kappa$ in \eqref{eq:h_r} to 1. The upper bound $U_{\bm{a}}$ on the coefficients $\bm{a}$ is set to 10, and the upper bound $\overline{\xi}$ on $\xi$ is also set to 10. For SBC-III, the degree of $h_1(\bm{a},\bm{x})$ in \eqref{eq:h_p}, parameter $\tau$ used in \eqref{eq:sbc_linear3},  and the upper bound $U_{\bm{a}}$ on the coefficients $\bm{a}$ follow the specifications listed in Tab. \ref{tab:para_mc}. 
\vspace{-10pt}
\begin{table}[ht]
\centering
\caption{Degree of $h_1$, $\tau$, and $U_{\bm{a}}$ settings of SBC-III in Tab.~\ref{tab:results} experiments.}
\setlength{\tabcolsep}{1.5mm}{
    \begin{tabular}{*{11}{c}}
    \toprule
    Examples & \ref{ex:vinc} & \ref{ex:arch} & \ref{ex:stable} & \ref{ex:4d} & \ref{ex:6d} & \ref{ex:lotka} & \ref{ex:pendulum} & \ref{ex:sank} & \ref{ex:lorenz}  \\ \midrule
    Degree & 2 & 2 & 3 & 4 & 6 & 20 & 20 & 8 & 7\\
    $\tau$ & 0.01 & 0.01 & 0.01 & 0.01 & 0.01 & 0.02 & 0.02 & 0.02 & 0.02 \\
    $U_{\bm{a}}$ & 1.1 & 1.1 & 1.1 & 1.1 & 1.1 & 1.5 & 1.5 & 1.1 & 1.1 \\
    \bottomrule
\end{tabular}}
\label{tab:para_mc}
\end{table}
\vspace{-15pt}
\subsection{Additional Results}
\label{sec:app_results}
In this subsection, we present additional experimental results. We begin with Table~\ref{tab:lotka} from Example~\ref{ex:lotka}, which compares the SBC-III, SBC-SOS, and Monte Carlo methods. We then report the detailed optimization results obtained by SBC-SOS and SBC-III for the experiments reported in Tab.~\ref{tab:results}. 
% Finally, we present and discuss the results obtained when varying the hyperparameters.
% Next, using Examples~\ref{ex:vinc} and~\ref{ex:lotka} as case studies, we examine the outcomes of our method under different random seeds. 
% Finally, using Examples~\ref{ex:vinc} and~\ref{ex:lotka} as case studies, we present and discuss the results obtained when varying the hyperparameters.

\begin{table}[!htbp]
\centering
\caption{Comparison in Example \ref{ex:lotka}.}
\setlength{\tabcolsep}{2mm}{
    \begin{tabular}{*{5}{c}}
    \toprule
    Step & State & $\textnormal{P}_{\text{SBC-III}}(\bm{x})$ & $\textnormal{P}_{\text{SBC-SOS}}(\bm{x})$ & $\textnormal{P}_{\textit{MC}}(\bm{x})$ \\ \midrule
    0 & $(-0.8000,-0.5000)$ & 0.2993 & 0.0000 & 0.4448  \\
    1 & $(-0.8000,0.5902)$  & 0.5440 & 0.0230 & 0.6901  \\
    2 & $(0.0721,-0.8994)$  & 0.7107 & 0.4341 & 0.8337 \\
    3 & $(0.1009,0.8400)$   & 0.7558 & 0.6622 & 0.8980 \\
    4 & $(-0.0343,0.1541)$  & 0.9668 & 0.9987 & 1.0000 \\
    5 & $(-0.0119,-0.1044)$  & 0.9678 & 0.9984 & 1.0000 \\
\bottomrule
\end{tabular}}
\label{tab:lotka}
\end{table}

We first present Tab.~\ref{tab:lotka} from Example~\ref{ex:lotka}. In Tab.~\ref{tab:lotka}, $\textnormal{P}_{\text{SBC-III}}(\bm{x}) = 1-\lambda^*(\mathbb{S})-h(\bm{a}^*(\mathbb{S}),\bm{x})$ and $\textnormal{P}_{\text{SBC-SOS}}(\bm{x})=1-\lambda^*-h(\bm{a}^*,\bm{x})$ denote the results computed by the two methods, respectively. Moreover, $\textnormal{P}_{\textit{MC}}(\bm{x})$ denotes the Monte Carlo estimate, computed by generating $10^6$ successor samples for each state $\bm{x}$ and calculating the fraction that remains in the safe set $\mathbb{X}$.

The detailed optimization results obtained by SBC-SOS and SBC-III are listed in Tab.~\ref{tab:results_sbc}. In Tab.~\ref{tab:results_sbc}, $J^*(\mathbb{S})$ denotes the optimum of the optimization problem \eqref{eq:sbc_linear3}, i.e., $J^*(\mathbb{S}) = \lambda^*(\mathbb{S}) + \frac{1}{N_o}\textstyle\sum_{k=1}^{N_o}h(\bm{a}^*(\mathbb{S}),\bm{x}'_k)$. Similarly, $J^*$ denotes the optimum of \eqref{eq:sbc_sos}. $J^*(\mathbb{S})$ and $J^*$ indicate, to some extent, the conservatism of the resulting lower probability bounds $1-h(\bm{a}^*(\mathbb{S}),\bm{x})-\lambda^*(\mathbb{S})$ and $1-h(\bm{a}^*,\bm{x})-\lambda^*$, respectively, and smaller values are more desirable. As shown in Tab.~\ref{tab:results_sbc}, our method attains values of $J^*(\mathbb{S})$ comparable to or better than $J^*$ obtained by SBC-SOS, demonstrating its effectiveness.
% \begin{table}[ht]
% \centering
% \caption{The optimal value of $\lambda$ in Table \ref{tab:results} Experiments.}
% \setlength{\tabcolsep}{2mm}{
%     \begin{tabular}{*{11}{c}}
%     \toprule
%     Examples & \ref{ex:vinc} & \ref{ex:arch} & \ref{ex:stable} & \ref{ex:4d} &\ref{ex:6d} & \ref{ex:lotka} & \ref{ex:pendulum} & \ref{ex:sank} & \ref{ex:lorenz}   \\ \midrule
%     $\lambda^*(\mathbb{S})$ & 0.01 & 0.01 & 0.01 & 0.01 & 0.01 & 0.0218 & 0.0790 & 0.1870 & 0.3767\\
%     \bottomrule
% \end{tabular}}
% \label{tab:lambda}
% \end{table}
% \vspace{-15pt}
\begin{table*}[!h]
\centering
\caption{Results of SBC-SOS and SBC-III in Table \ref{tab:results} Experiments}
\setlength{\tabcolsep}{5mm}{
    \begin{tabular}{*{5}{c}}
    \toprule
    \multirow{2}{*}{EX} & \multicolumn{2}{c}{SBC-SOS} & \multicolumn{2}{c}{SBC-III} \\ \cmidrule(lr){2-3} \cmidrule(lr){4-5}
    & $J^*$ & $\lambda^*$ & $J^*(\mathbb{S})$ & $\lambda^*(\mathbb{S})$  \\ \midrule
    \ref{ex:vinc}     & 0.1308 & 0.0008 & 0.0100 & 0.0100\\
    \ref{ex:arch}     & 0.0007 & 0.0003 & 0.0100 & 0.0100\\
    \ref{ex:stable}   & 0.1218 & 0.0031 & 0.0100 & 0.0100\\
    \ref{ex:4d}       & 0.2353 & 0.0034 & 0.0100 & 0.0100\\
    \ref{ex:6d}       & 0.4035 & 0.0000 & 0.0100 & 0.0100\\
    \ref{ex:lotka}    & 0.3256 & 0.0007 & 0.1760 & 0.0218\\
    \ref{ex:pendulum} & -      & -      & 0.2892 & 0.0790\\
    \ref{ex:sank}     & 0.4624 & 0.0000 & 0.1949 & 0.1870 \\
    \ref{ex:lorenz}   & 1.0000 & 0.0000 & 0.5141 & 0.3767 \\
\bottomrule
\end{tabular}}
\label{tab:results_sbc}
\end{table*}

\subsection{Impact of Hyperparameters}
\label{sec:app_para}
In this part, we summarize the behavior of RBC-I, RBC-II, and SBC-III under different hyperparameter settings.

\paragraph{RBC-I under Varying Hyperparameters.} 
We use Example~\ref{ex:vinc} to examine how different hyperparameters affect the performance of RBC-I.
Fig.~\ref{fig:vinc_time} reports the computation times for different choices of $\alpha_1$ and $\alpha_2$. For all settings shown in Fig.~\ref{fig:vinc_time},  RBC-I yields $\xi^* = 0$, thereby providing a valid guarantee stated in Theorem~\ref{theo:robust}. Its computation time also remains consistently small.
% Since the sample complexity of RBC-I is relatively low, its computation time does not vary significantly across different $\alpha_1$ and $\alpha_2$ values. This indicates that, in practice, one may choose relatively small $\alpha_1$ and $\alpha_2$ to obtain strong probabilistic guarantees without introducing significant computational overhead.

Moreover, the influence of $\delta$ on the required sample size $N$ is mild. Hence, $\delta$ can typically be set extremely close to zero to obtain a very high confidence level, as in Tab.~\ref{tab:results} where $\delta = 10^{-6}$. In addition, due to the template structure in \eqref{barrier_form}, RBC-I does not require a high polynomial degree for $h_1(\bm{a},\bm{x})$. 
% In all examples, increasing the degree of $h_1(\bm{a},\bm{x})$ has little effect on whether $\xi^*$ equals zero. 
Other hyperparameters, including $\overline{\xi}$, $U_{\bm{a}}$, and $C$, have only a minor influence on both the guarantee results and the computation time.

% \begin{figure}[!htbp]
%     \centering
%     \includegraphics[width=0.4\linewidth]{Fig/vinc_single.pdf}
%     \caption{Computation time (in seconds) of RBC-I under different $\alpha_1$ and $\alpha_2$.}
%     \label{fig:vinc_time}
% \end{figure}

\paragraph{RBC-II under Varying Hyperparameters.} We analyze the influence of hyperparameters on RBC-II using Example~\ref{ex:vinc}.
% To isolate the effect of each hyperparameter, we vary one of $\alpha_1$, $\alpha_2$, or $\delta_2$ while keeping the other two fixed. The resulting computation times are presented in Fig.~\ref{fig:vinc_nested}. Across all settings, RBC-II consistently returns $\xi^*=0$, thereby providing a valid guarantee stated in Theorem~\ref{thm:main-vc-final-corrected}.
The values of $1-\frac{\alpha_1}{l\delta_2}$ and computation times under different $\alpha_1$ and $l$ are shown in Fig.~\ref{fig:vinc_nested_p} and \ref{fig:vinc_nested_t}, while the computation times under different $\alpha_2$ are presented in Fig.~\ref{fig:vinc_nested_alpha2}. Across all settings, RBC-II consistently yields $\xi^*=0$ and thus provides the valid guarantee stated in Theorem~\ref{thm:main-vc-final-corrected}.
Since the required sample sizes $N$ and $M$ satisfy
\[N \geq \frac{2}{\alpha_1}\left(\ln\frac{1}{\delta_1} + m+1\right), \quad M \geq \frac{1}{2\alpha_2^2}\ln\frac{1}{(1-l)\delta_2},\]
smaller $\alpha_1$ or $\alpha_2$, or a larger $l$, lead to stronger probabilistic guarantees but increase the sample size, which leads to longer computation times. 
% Although the main diagonal of Fig.~\ref{fig:vinc_nested_p} shows values of $1-\frac{\alpha_1}{l\delta_2}$ that are all close to  $0.95$ and thus yield the same probabilistic guarantee, each entry corresponds to a different combination of $l$ and $\alpha_1$, and consequently to different sample sizes $N$ and $M$, as shown in Tab.~\ref{tab:vinc_nested_l}. In Example \ref{ex:vinc}, when fixing $1-\frac{\alpha_1}{l\delta_2}=0.95$, the sample size $N\times M$ is minimized at $159354$ when $l\approx0.085$ and $\alpha_1\approx0.004$. For a fixed value of $1-\frac{\alpha_1}{l\delta_2}$, increasing $l$ and decreasing $\alpha_1$ reduce the required state sample size $N$ but increase the disturbance sample size $M$. Therefore, in practice, one may choose $l$ and $\alpha_1$ based on the relative cost of collecting state versus disturbance samples.
In practice, $\alpha_1$, $\alpha_2$, and $l$ should be chosen to balance conservativeness and computational efficiency.

\begin{figure}[!htbp]
    \centering
    \subfigure[\scriptsize RBC-I: Times (in second)]{\includegraphics[width=0.31\linewidth]{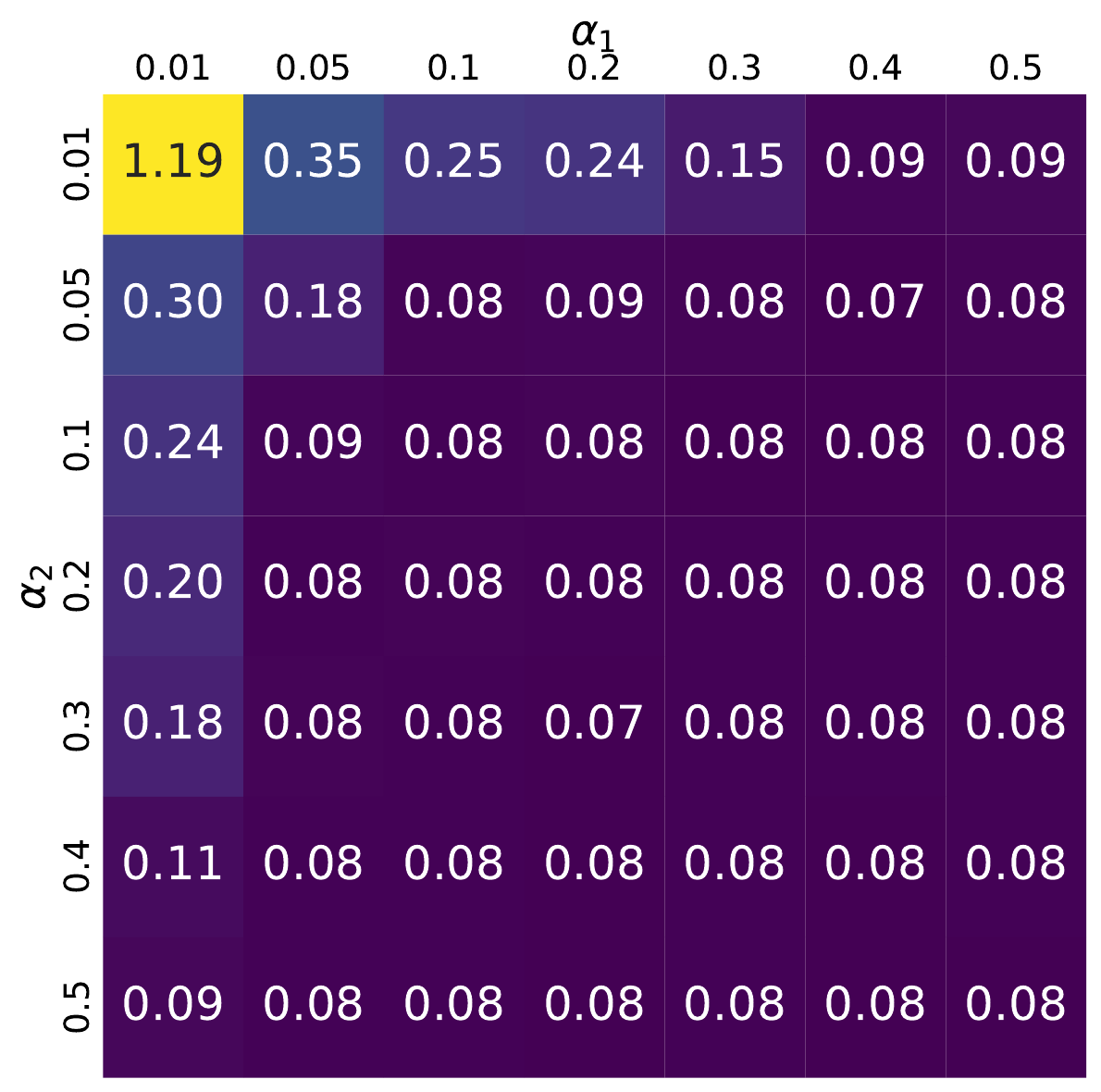}
    \label{fig:vinc_time}}
    \subfigure[\scriptsize RBC-II: Values of $1-\frac{\alpha_1}{l\delta_2}$]{\includegraphics[width=0.31\linewidth]{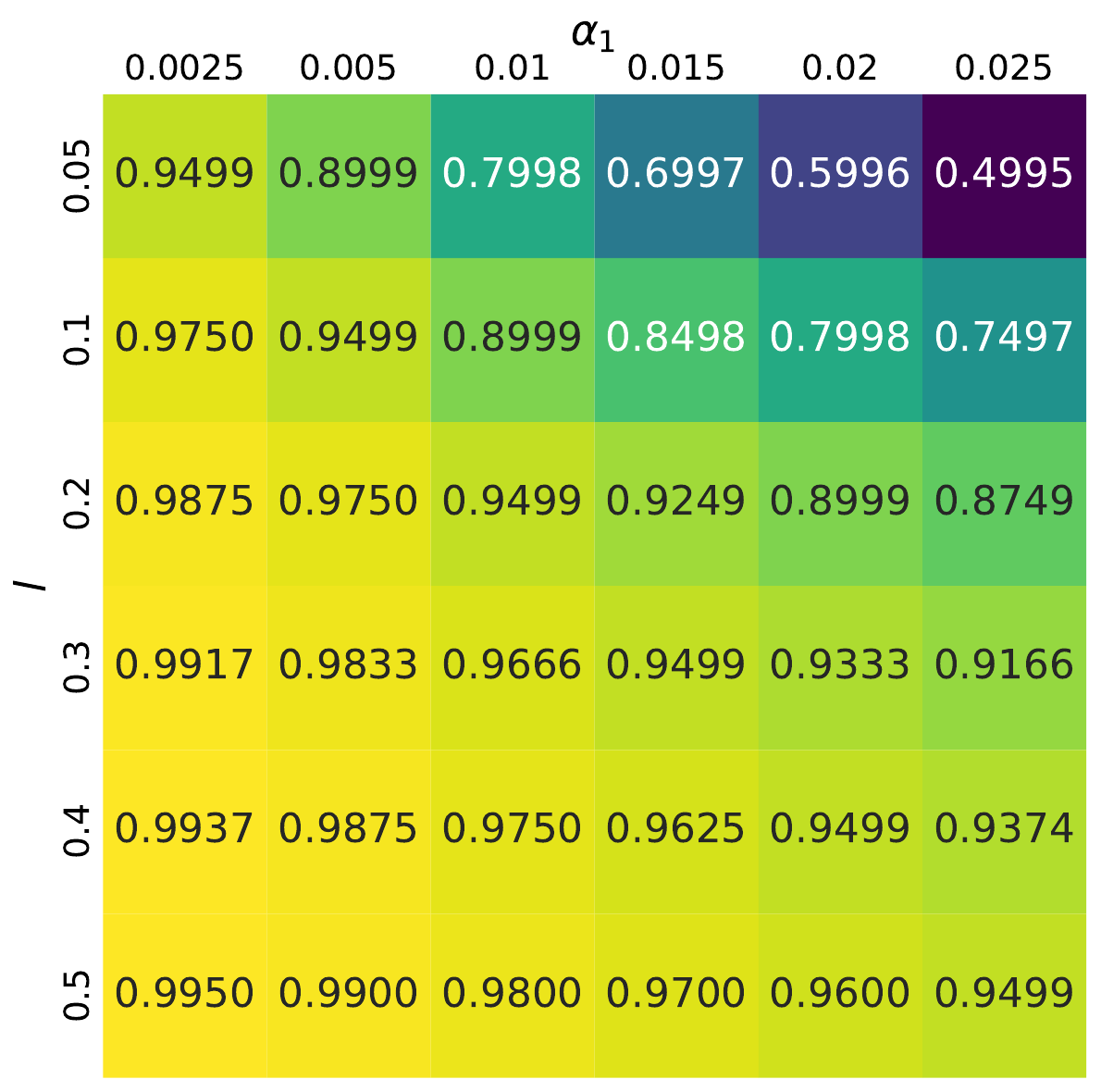}
    \label{fig:vinc_nested_p}}
    \subfigure[\scriptsize RBC-II: Times (in second)]{\includegraphics[width=0.31\linewidth]{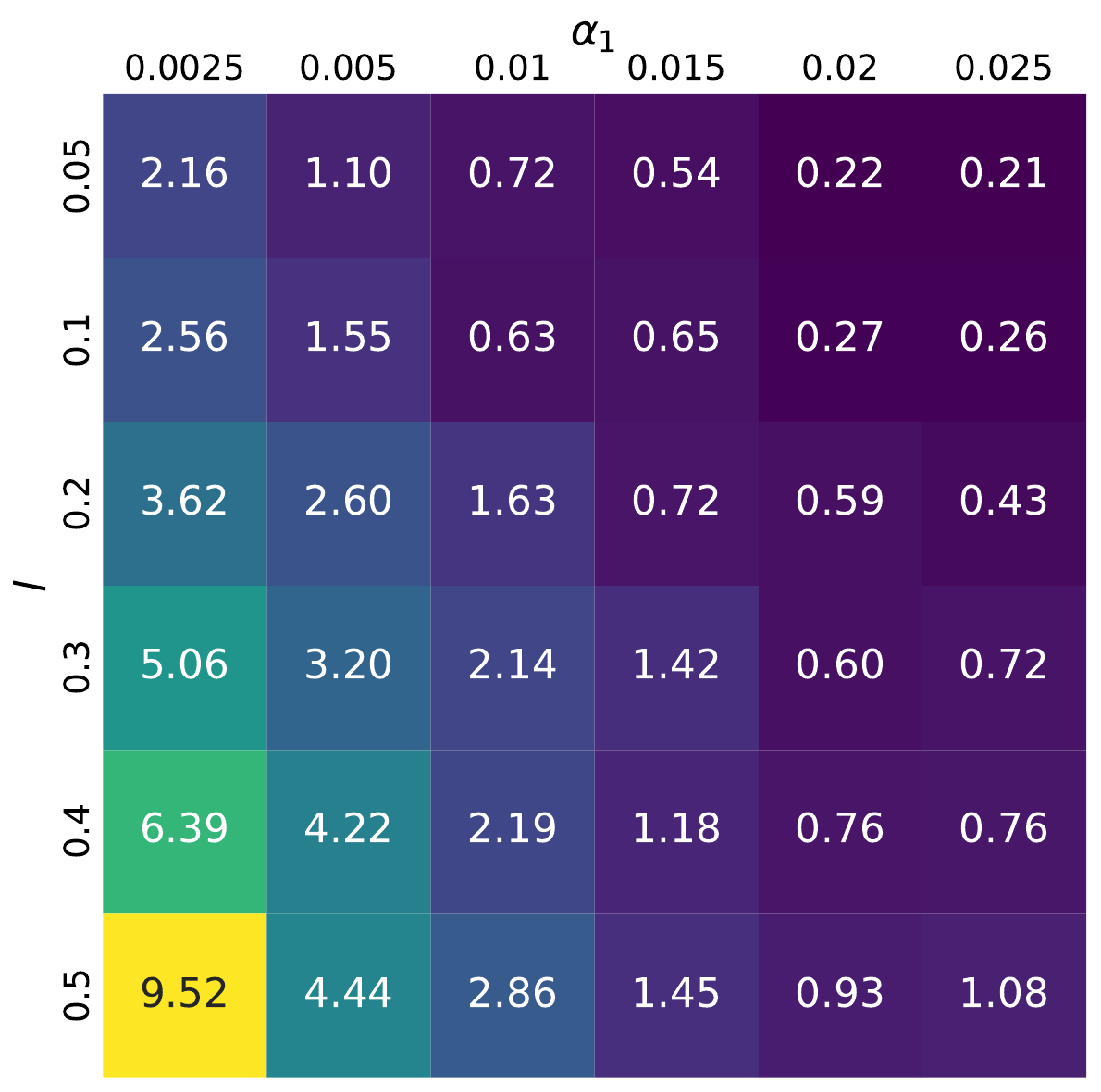}\label{fig:vinc_nested_t}}
    \caption{Performance and computation time.}
    \label{fig:vinc_nested}
\end{figure}

\begin{figure}[!htbp]
    \centering
    \includegraphics[width=0.6\linewidth]{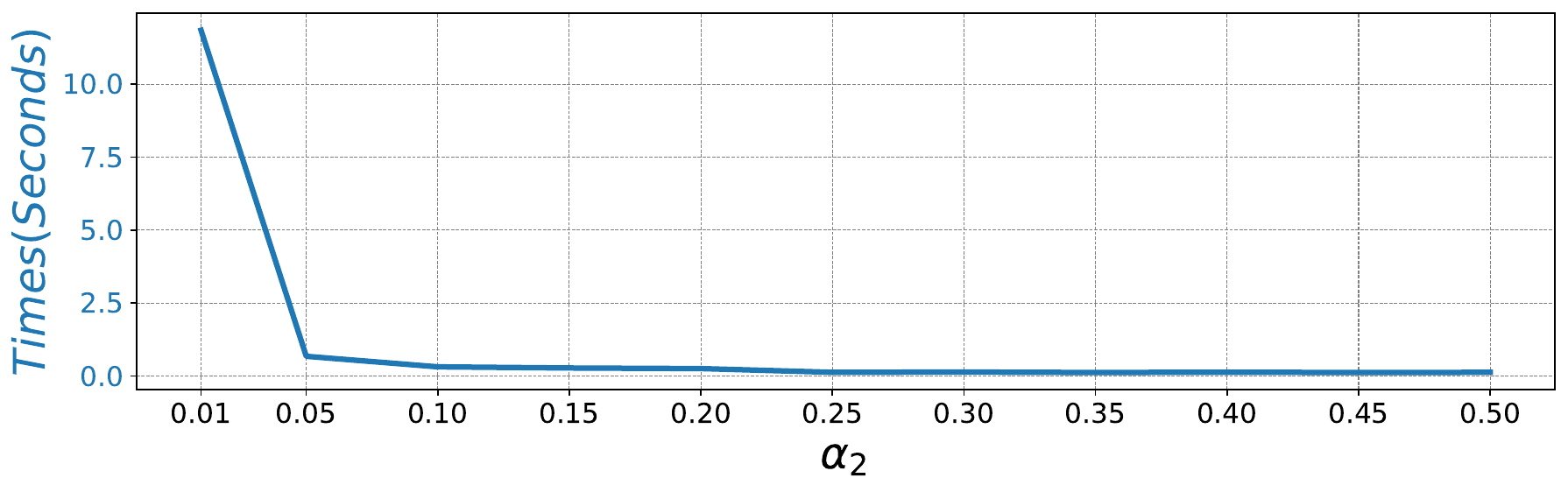}
    \caption{Computation time of RBC-II under different $\alpha_2$.}
    \label{fig:vinc_nested_alpha2}
\end{figure}

% \begin{table}[!htbp]
% \centering
% \caption{Sample sizes under different $l$ settings.}
% \setlength{\tabcolsep}{3mm}{
%     \begin{tabular}{*{6}{c}}
%     \toprule
%     $l$ & $\alpha_1$ & $1-\frac{\alpha_1}{l\delta_2}$ & $N$ & $M$ & $N\times M$ \\ \midrule
%     0.05 & 0.0025 & 0.9499 & 15053 & 11  & 165583\\
%     0.1  & 0.005  & 0.9499 & 7527  & 22  & 165594\\
%     0.2  & 0.01   & 0.9499 & 3764  & 45  & 169380\\
%     0.3  & 0.015  & 0.9499 & 2509  & 72  & 180648\\
%     0.4  & 0.02   & 0.9499 & 1882  & 103 & 193846\\
%     0.5  & 0.025  & 0.9499 & 1506  & 139 & 209334\\
% \bottomrule
% \end{tabular}}
% \label{tab:vinc_nested_l}
% \end{table}

Similar to RBC-I, $\delta_1$ can be set to a value extremely close to zero, such as $10^{-6}$, to ensure a high confidence level, and the polynomial degree of $h_1(\bm{a},\bm{x})$ can be kept low to reduce the disturbance sample size $M$. For $\delta_2$, choosing a value close to one reduces the required disturbance sample size and increases the value of $1 - \frac{\alpha_1}{l\delta_2}$, thereby yielding a less conservative guarantee. 
% In practice, $\delta_2$ is often fixed to a value very close to one, such as $0.999$. 
Other hyperparameters, including $\overline{\xi}$, $U_{\bm{a}}$, and $C$, have only a minor impact on both the guarantee results and the computation time.

\paragraph{SBC-III under Varying Hyperparameters.} We analyze the influence of hyperparameters on SBC-III using Example~\ref{ex:lotka}. Fig.~\ref{fig:lotka_para1} reports the optimal values $J^*(\mathbb{S})$ and the computation times of SBC-III as we vary $\tau$, $U_{\bm{a}}$, and the polynomial degree of $h_1(\bm{a},\bm{x})$. Since the polynomial degree also influences the performance of the SBC-SOS baseline, we include its results and computation times in the third subfigure of Fig.~\ref{fig:lotka_para1} for comparison.

\begin{figure}[!h]
    \centering
    \subfigure{\includegraphics[width=0.6\linewidth]{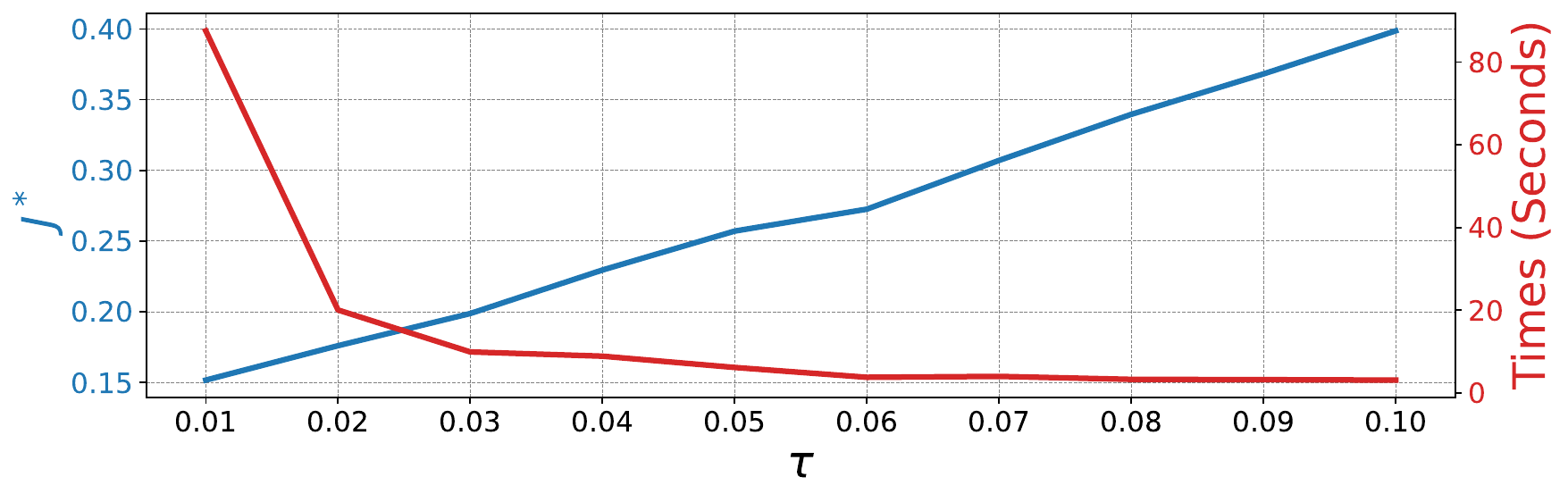}}
    \subfigure{\includegraphics[width=0.6\linewidth]{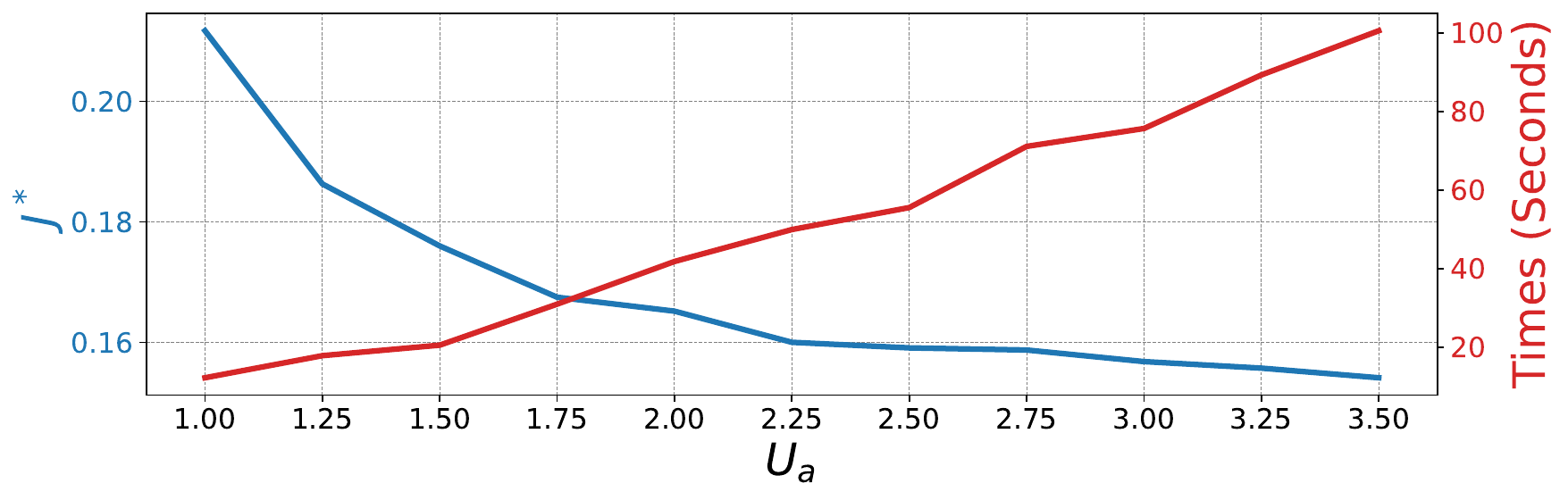}}
    \subfigure{\includegraphics[width=0.6\linewidth]{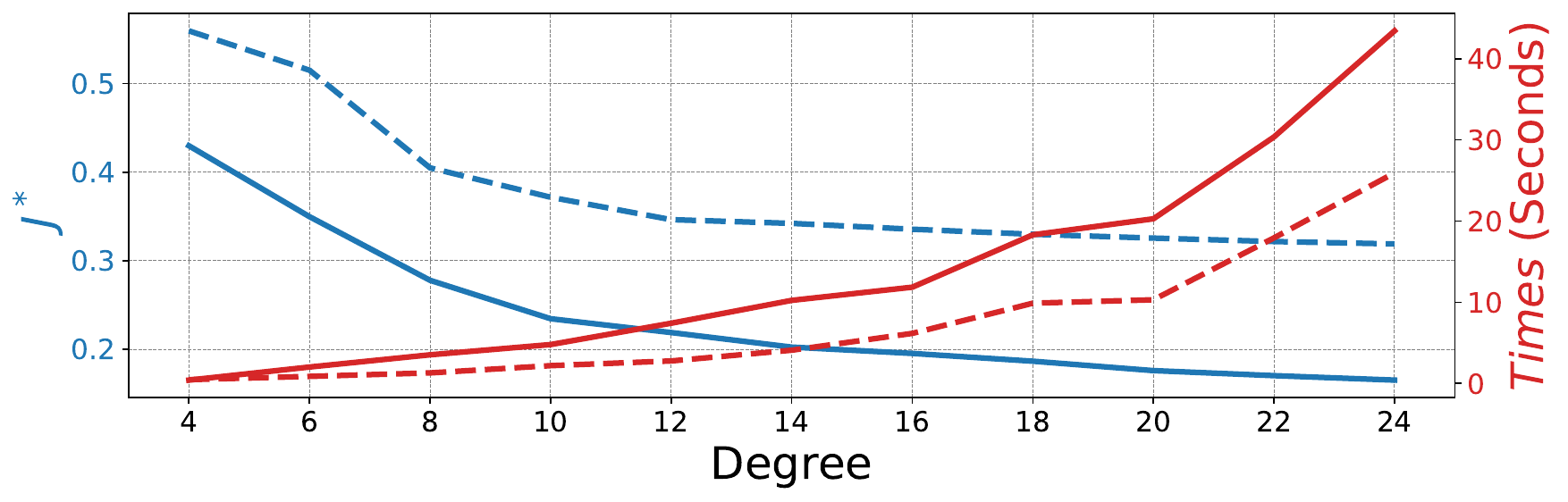}}
    \caption{$J^*(\mathbb{S})$ and computation times of SBC-III under different hyperparameters. Blue lines show the values of $J^*(\mathbb{S})$, and red lines show the computation times. In the third subfigure, solid lines correspond to our SBC-III method, and dashed lines correspond to the SBC-SOS baseline.}
    \label{fig:lotka_para1}
\end{figure}

For parameter $\tau$, a smaller value relaxes the constraint in \eqref{eq:sbc_linear3}, which leads to a better $J^*(\mathbb{S})$. However, the disturbance sample size $M$ in SBC-III must satisfy $M \geq \frac{U_a^2}{2\tau^2}\ln\frac{1}{(1-l)\delta_2}$,
so decreasing $\tau$ requires more disturbance samples and leads to longer computation times. For parameter $U_{\bm{a}}$, increasing its value expands the feasible region of $\bm{a}$, which enables better $J^*(\mathbb{S})$, though this effect gradually saturates as $U_{\bm{a}}$ becomes sufficiently large. A larger $U_{\bm{a}}$ also enlarges the required disturbance sample size $M$, which further increases computation time. Regarding the degree of $h_1(\bm{a},\bm{x})$, Fig.~\ref{fig:lotka_para1} shows that our method consistently yields better $J^*$ than the white-box SBC-SOS baseline across all degrees, highlighting the strength of our approach. A higher degree enhances the expressive power of $h(\bm{a},\bm{x})$ and reduces $J^*(\mathbb{S})$, but it also increases the dimension $m$ of $\bm{a}$, thereby increasing the required state-sample size $N$ and consequently the computation time.

The optimal values $J^*(\mathbb{S})$ and computation times under different $\alpha_1$ and $l$ settings are listed in Fig.~\ref{fig:lotka_para2}. Across these settings, $J^*(\mathbb{S})$ varies only slightly, indicating that SBC-III is robust to variations in $\alpha_1$ and $l$. The trends in sample size and computation time with respect to $\alpha_1$ and $l$ are similar to those observed for RBC-II.
% \textcolor{blue}{Similar to RBC-Nested, the values of $1-\frac{\alpha_1}{l\delta_2}$ along the main diagonal of Fig.~\ref{fig:lotka_a_d} are all close to 0.95, but each point corresponds to a different combination of $\alpha_1$ and $l$ and thus to different sample sizes $N$ and $M$. Increasing $l$ (equivalently, decreasing $\alpha_1$) reduces the required number of state samples $N$, which decreases the number of constraints in \eqref{eq:sbc_linear3} and consequently shortens the solving time. However, choosing $l$ too large may increase the overall sample size $N\times M$. In Example~\ref{ex:lotka} with $1-\frac{\alpha_1}{l\delta_2}=0.95$, selecting $l=0.2$ achieves a favorable balance between sampling cost and computational efficiency, yielding the smallest overall computation time.}
% Moreover, since the state sample size $N$ scales with $1/\alpha_1$ while the disturbance sample size $M$ grows only logarithmically with $1/\delta_2$, the computation time is more sensitive to changes in $\alpha_1$ than in $\delta_2$, as reflected in Fig.~\ref{fig:lotka_a_d_t}. 

\begin{figure}[!h]
    \centering
    \subfigure[$1-\frac{\alpha_1}{l\delta_2}$]{\includegraphics[width=0.31\linewidth]{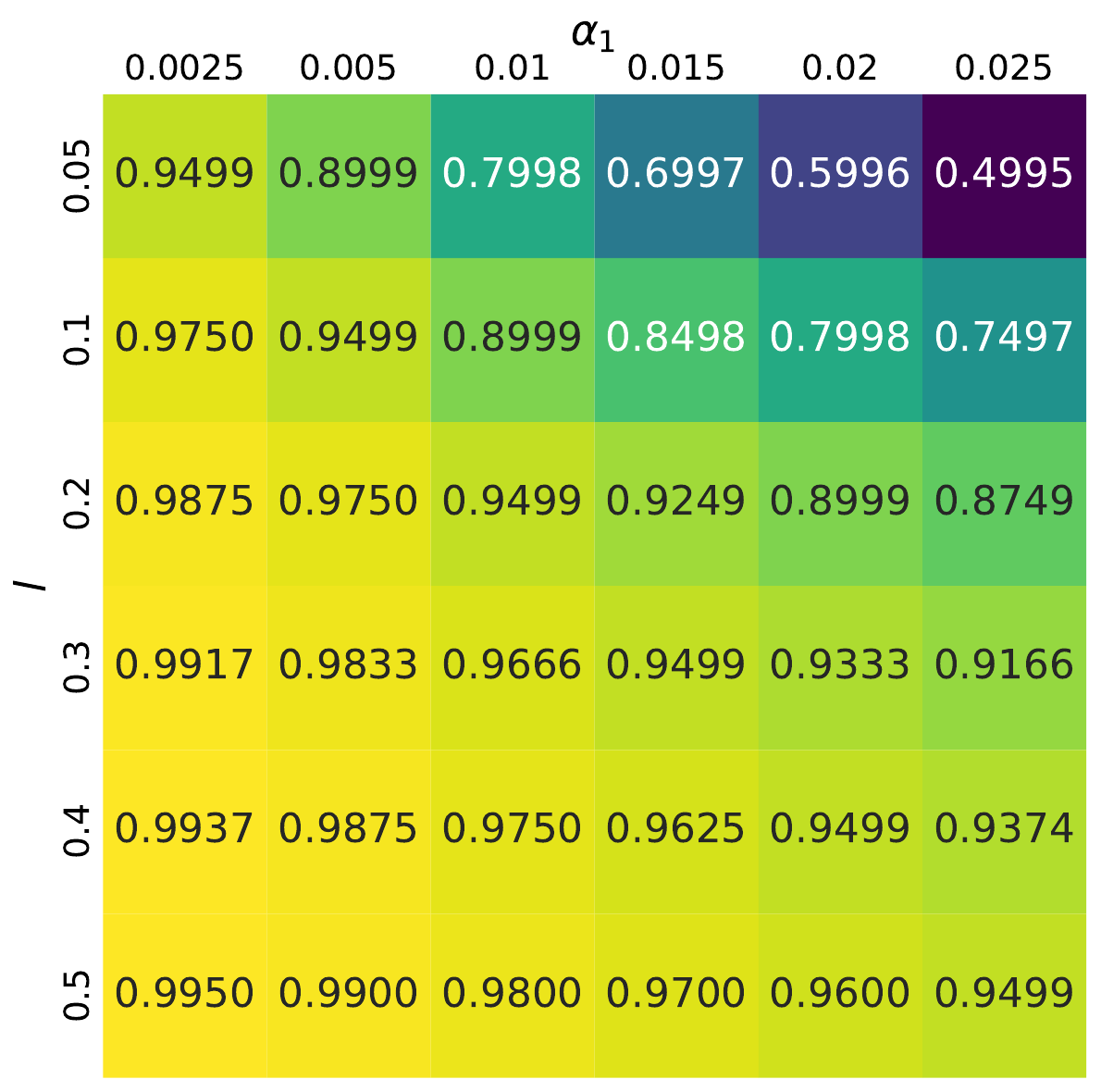}
    \label{fig:lotka_a_d}}
    \subfigure[$J^*(\mathbb{S})$]{\includegraphics[width=0.31\linewidth]{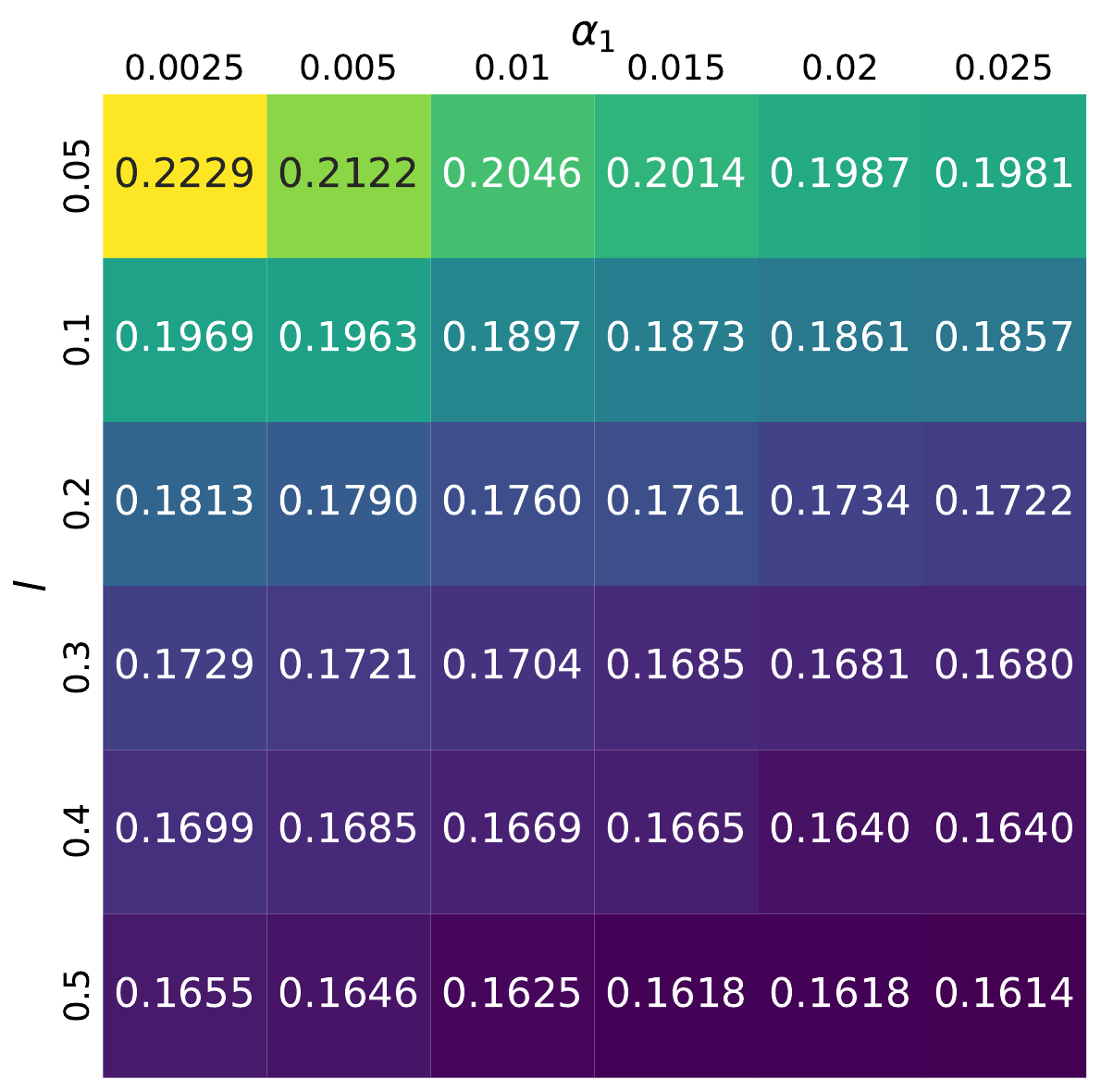}
    \label{fig:lotka_a_d_j}}
    \subfigure[Times (in seconds)]{\includegraphics[width=0.31\linewidth]{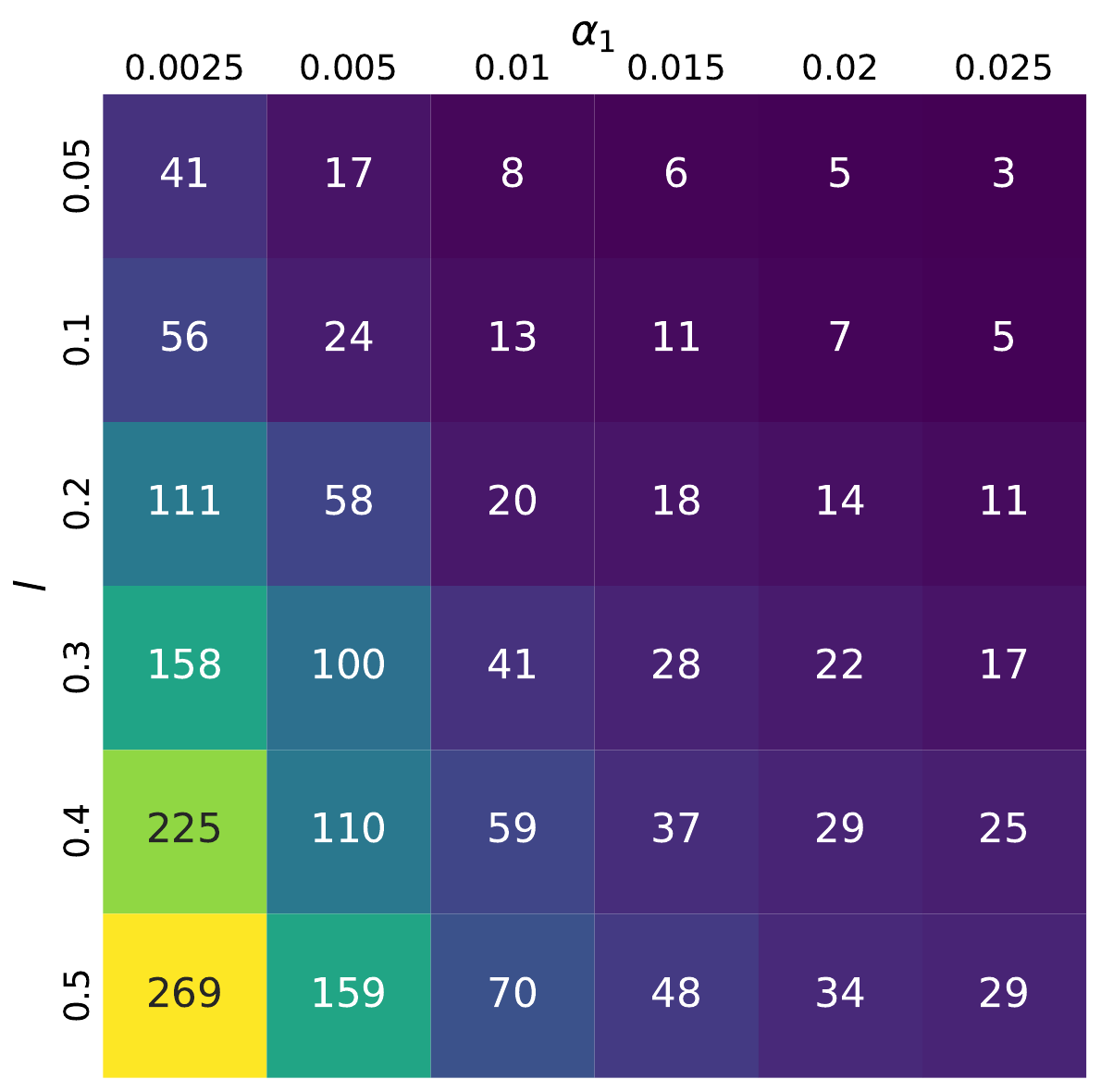}
    \label{fig:lotka_a_d_t}}
    \caption{The values of $1-\frac{\alpha_1}{l\delta_2}$, $J^*(\mathbb{S})$, and computation times of SBC-III under different $\alpha_1$ and $\delta_2$.}
    \label{fig:lotka_para2}
\end{figure}
% Since the required number of samples $N$ depend on $\alpha_1$ and $\alpha_2$, smaller values of these parameters lead to larger sample sizes and hence longer computation times. In practice, $\alpha_1$ and $\alpha_2$ should be chosen to balance conservativeness and computational efficiency.

\subsection{Application of Our Method in CARLA}
\label{sec:carla}

To further demonstrate the practical applicability of the proposed method, we evaluate our PAC one-step safety certification approaches in the CARLA 0.9.15 simulator \cite{Dosovitskiy17}, a high-fidelity platform for autonomous driving research. Specifically, we consider two application cases: a speed-maintenance case and a car-following case, as shown in Fig.~\ref{fig:carla}. In both cases, the precise mathematical form of system dynamics is unknown and stochastic.

\begin{figure}[!h]
    \centering
     \subfigure[\scriptsize Case 1]{\includegraphics[width=0.45\linewidth]{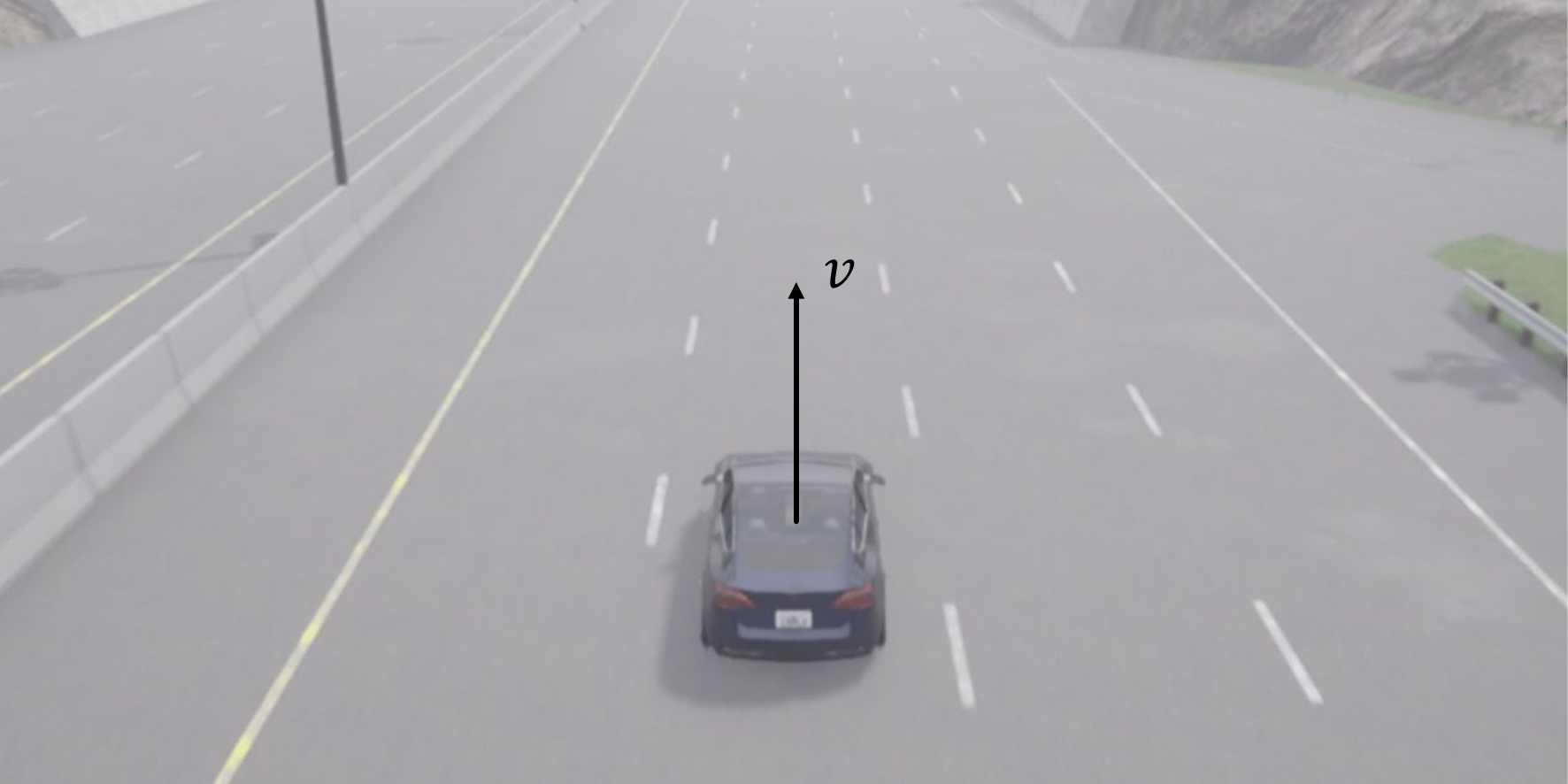}
     \label{fig:carla1}}
    \subfigure[\scriptsize Case 2]{\includegraphics[width=0.45\linewidth]{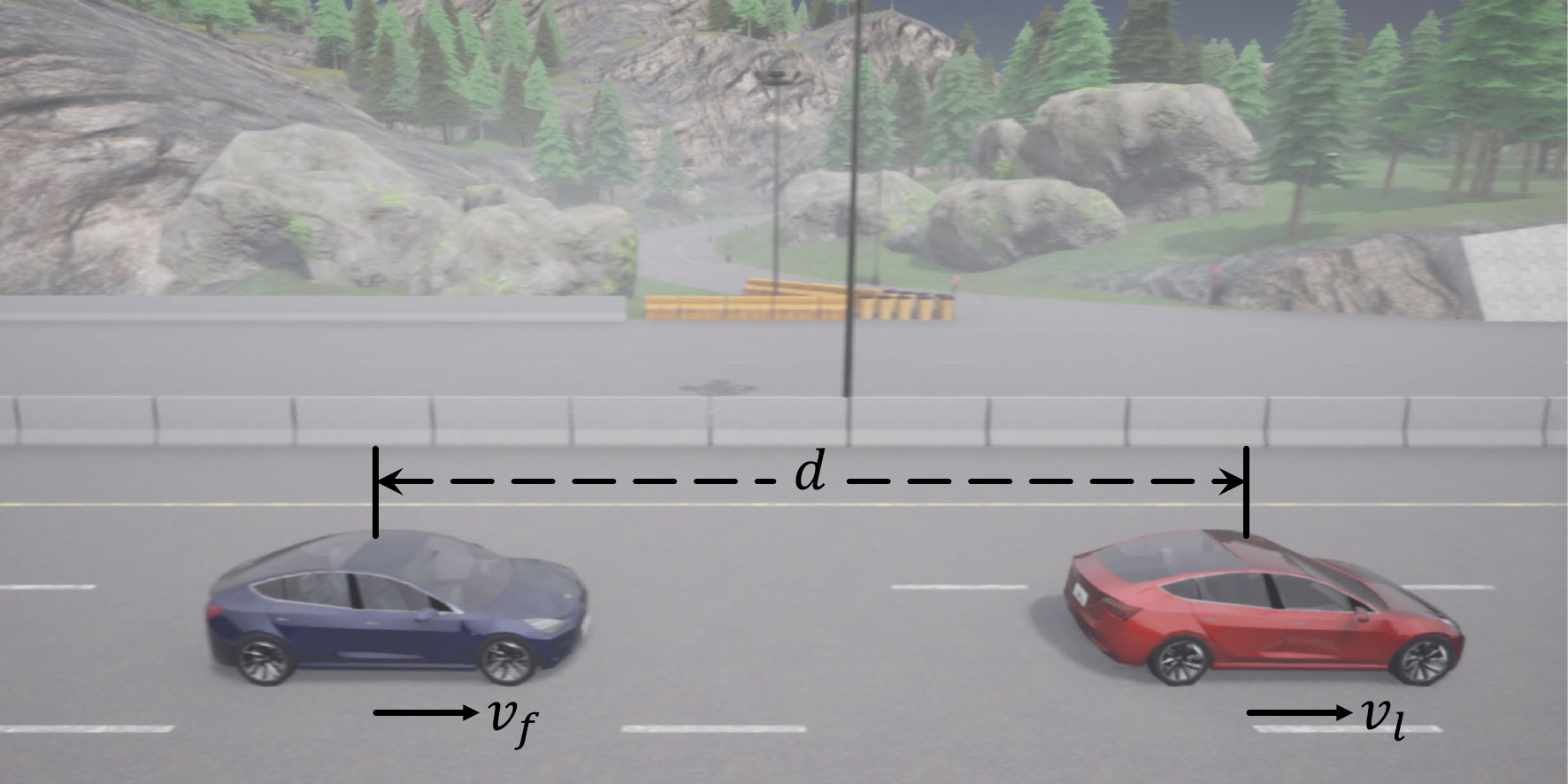}
    \label{fig:carla2}}
    \caption{Illustration of the CARLA simulation environment}
    \label{fig:carla}
\end{figure}

\subsubsection{Case 1 (Speed Maintenance).}
Consider an autonomous vehicle driving along a straight road, with the objective of maintaining its speed $v$ within the range of $20 ~km/h$ to $60 ~km/h$. 
% i.e., the safety set is defined as $\mathbb{X}=\{\bm{v}\in\mathbb{R} \mid 20 \leq v \leq 60\}$. 
% Fig.~\ref{fig:carla1} provides an illustration of the CARLA simulation environment for this case.
The vehicle is controlled by CARLA's built-in PID controller, which attempts to regulate the speed around $40 ~km/h$. We apply our method to assess whether the closed-loop system 
% governed by this PID controller 
can be certified to remain within the prescribed safe speed range.

We first apply RBC-I with $\delta = 10^{-6}$ and $\alpha_1 = \alpha_2 = 0.05$, and we obtain $\xi^*(\mathbb{S})=0$, which yields the following formal guarantee:
\[\textnormal{P}_{\mathbb{S}}\bigg[ \textnormal{P}_{\bm{x}}\Big[\textnormal{P}_{\bm{d}}\!\big[ \bm{f}(\bm{x},\bm{d}) \in \mathbb{X} \big] \ge 0.95 \Big] \ge 0.95 \bigg] \ge 1-10^{-6}.\]
For RBC-II, we set $\delta_1 = 10^{-6}$, $\alpha_1 = 0.005$, $\alpha_2 = 0.05$, $\delta_2 = 0.999$, and $l = 0.1$. RBC-II gives $\xi^*(\mathbb{S})=0$, leading to the guarantee
\[\textnormal{P}_{\mathbb{S}}\bigg[ \textnormal{P}_{\bm{x}}\Big[\bm{f}(\bm{x},\bm{d}) \in \mathbb{X} \big] \ge 0.95 \Big] \ge 0.95 \Big] \ge 1- 10^{-6}.\]
For SBC-III, we use $\delta_1 = 10^{-6}$, $\alpha_1 = 0.005$, $\delta_2 = 0.999$, $\tau=0.02$ and $l = 0.1$. Solving \eqref{eq:sbc_linear3} yields $J^*(\mathbb{S})=0.0200$ and the guarantee
\[\textnormal{P}_{\mathbb{S}}\bigg[ \textnormal{P}_{\bm{x}}\Big[\textnormal{P}_{\bm{d}}\!\big[ \bm{f}(\bm{x},\bm{d}) \in \mathbb{X} \big] \ge 1-h(\bm{a}^*(\mathbb{S}), \bm{x}) - \lambda^*(\mathbb{S}) \Big] \ge 0.95 \bigg] \ge 1- 10^{-6}.\]
where SBC-III further satisfies $1-\lambda^*(\mathbb{S})-h(\bm{a}^*(\mathbb{S}),\bm{x}) \approx 0.98$ for every $\bm{x} \in \mathbb{X}$.

Overall, all three methods provide highly non-conservative guarantees for this case. Using Monte Carlo simulation with 1000 randomly sampled initial speeds, the evolution of the vehicle speed over time is shown in Fig. \ref{fig:carla_sim}. In all sampled trajectories, the system remains strictly within the  safe speed range, which is consistent with the guarantees delivered by the three certification methods.

\begin{figure}
    \centering
    \includegraphics[width=0.6\linewidth]{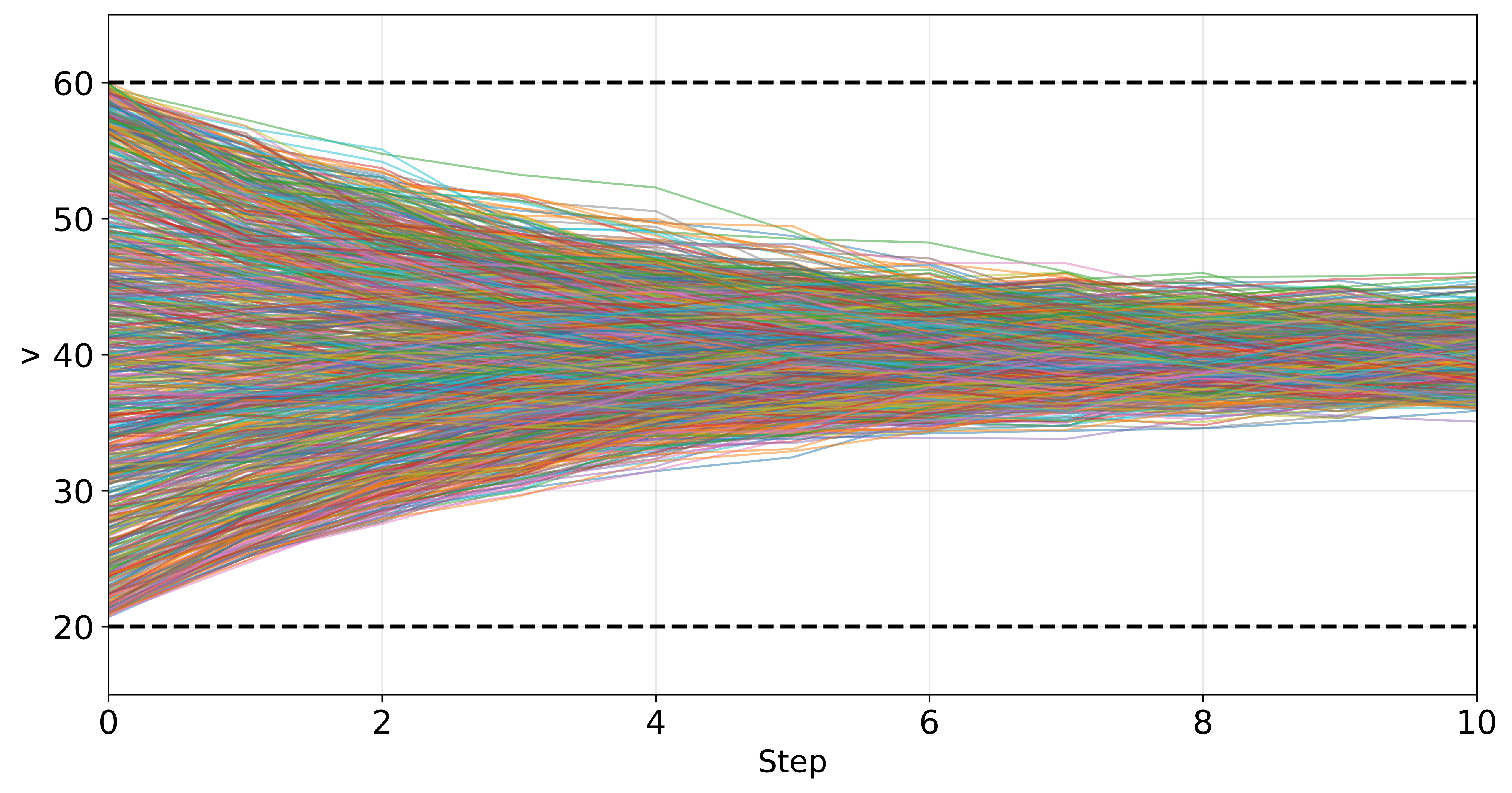}
    \caption{Simulated system trajectories for Case 1.}
    \label{fig:carla_sim}
\end{figure}

\subsubsection{Case 2 (Car Following).}
We next consider a car-following case involving a leading vehicle and a following vehicle, where the goal is to keep their relative distance and relative speed within prescribed safe bounds. 
% Fig.~\ref{fig:carla2} provides an illustration of the CARLA simulation environment for this case. 
The system state is defined by the relative distance $d$ between the two vehicles and the relative speed $v = v_f-v_l$, where $v_f$ and $v_l$ denote the speeds of the following and leading vehicles, respectively. The system is considered safe when $d\in [5,15]$ meters and $v \in [-20, 20] ~km/h$. The speed of the leading vehicle $v_l$ is modeled as a random variable uniformly distributed over $[45, 75] ~ km/h$, introducing stochasticity into the system dynamics. The following vehicle is controlled by a fixed policy obtained using the proximal policy optimization algorithm. 

In this case, RBC-I fails to provide a valid guarantee for $\alpha_1\leq 0.5$ and $\alpha_2 \leq 0.5$, as it consistently results in $\xi^*(\mathbb{S}) > 0$. Similarly, RBC-II is unable to provide a valid guarantee for $\alpha_1\leq 0.25$, $\alpha_2 \leq 0.5$ and $\delta_2 \leq 0.5$. Therefore, we focus on the SBC-III method. Setting $\delta_1=10^{-6}$, $\alpha_1 = 0.01$, $\delta_2 = 0.999$, and $l=0.2$, SBC-III yields $J^*(\mathbb{S})=0.1592$ and $\lambda^*(\mathbb{S}) = 0.0611$. A video included at \href{https://github.com/TaoranWu/PSC-BDSS}{https://github.com/TaoranWu/PSC-BDSS} demonstrates the system simulation starting from the initial state $d=14~m$ and $v=-15~km/h$, together with the corresponding safety certification provided by SBC-III.

\end{document}